\documentclass[12pt]{article}
\usepackage{graphicx}
\RequirePackage{multicol}
\RequirePackage{setspace}
\RequirePackage[leqno]{amsmath}

\makeatletter

\renewcommand\refname{REFERENCES}

  \usepackage{authblk}
\newcommand\TenPtJEL{\@setfontsize\TenPtJEL{10.5}{12}}
\long\def\JEL#1{\long\gdef\@JEL{#1}}

\makeatother
	\usepackage[english]{babel} 
	\usepackage{amsmath, amssymb, textcomp, amsthm, mathtools,mathrsfs} 
	\allowdisplaybreaks
	\usepackage{wasysym} 
	\usepackage{stmaryrd}
	\usepackage{dsfont}

 \usepackage{fancyhdr}
   	\usepackage{booktabs}
    \usepackage{makecell}
   	\usepackage{enumitem}
   	\usepackage{inputenc}
	\usepackage{csquotes}
	\usepackage{floatrow}
    \usepackage{xcolor}

	\usepackage{footnote,footmisc} 

	\usepackage{tikz}
 \usetikzlibrary{calc, decorations.pathreplacing}
\usepackage{comment}
\usepackage{graphicx}
\usepackage{algorithm}
\usepackage{algorithmic}
 \usepackage{subcaption}
 \usepackage{pgfplots}
\usepgfplotslibrary{statistics}
\pgfplotsset{compat=1.16}

	\usepackage{pgfplots}
	\pgfplotsset{compat=1.17}
	\usetikzlibrary{shapes.misc}
	\usetikzlibrary{patterns}
	\usepackage{xcolor}
    \usepackage{hyperref}
	\hypersetup{
    	colorlinks,
    	linkcolor={red!50!black},
    	citecolor={blue!50!black},
    	urlcolor={blue!80!black}
	}
	\usetikzlibrary{decorations.pathreplacing}
	\usetikzlibrary{arrows}
	\usepackage{float} 
	\usepackage{pifont}

\usepackage{floatrow}
	\usepackage{threeparttablex}
    \usepackage{longtable}

\newenvironment{supplement}{}{}
\newcommand{\stitle}{}

\newcommand{\mres}{\mathbin{\vrule height 1.6ex depth 0pt width
0.13ex\vrule height 0.13ex depth 0pt width 1.3ex}}
	\newcommand*\dd{\mathop{}\!\mathrm{d}}

	\usepackage[toc,page,title]{appendix}
	
	\newtheorem{theorem}{Theorem}
	\newtheorem{lemma}{Lemma} 
	\newtheorem{proposition}{Proposition} 
	\newtheorem{definition}{Definition}

    \newtheorem{corollary}{Corollary}
	\newtheorem{assumption}{Assumption}

	\newtheorem*{example*}{Example}
	\usepackage{graphicx}  
	\usepackage{marvosym}

    \newtheorem*{assumptions*}{\assumptionnumber}
\providecommand{\assumptionnumber}{}

\usepackage{caption}
\usepackage{mwe}

\usepackage{pdflscape}      
\usepackage{array}          

\definecolor{cyan}{cmyk}{1, 0.4, 0, 0}

\newcommand{\nddvd}[1]{{\color{blue}[David: #1]}}

\definecolor{mypink}{RGB}{219, 48, 122}

\def\Xint#1{\mathchoice
{\XXint\displaystyle\textstyle{#1}}%
{\XXint\textstyle\scriptstyle{#1}}%
{\XXint\scriptstyle\scriptscriptstyle{#1}}%
{\XXint\scriptscriptstyle\scriptscriptstyle{#1}}%
\!\int}
\def\XXint#1#2#3{{\setbox0=\hbox{$#1{#2#3}{\int}$ }
\vcenter{\hbox{$#2#3$ }}\kern-.6\wd0}}

\def\dashint{\Xint-}

\usepackage{natbib}
\makeatletter

\allowdisplaybreaks

\usepackage[margin=1in]{geometry}
\usepackage{xr}

 \makeatletter
 \def\thanks#1{\protected@xdef\@thanks{\@thanks
        \protect\footnotetext{#1}}}
\makeatother

\title{\vspace{-2cm}\Huge{Free Discontinuity Regression} \\
\normalsize{With an Application to the Economic Effects of Internet Shutdowns}}
\author{\normalsize Florian Gunsilius$^{1}$\thanks{\texttt{fgunsil@emory.edu, dvdijcke@umich.edu}. Authors listed alphabetically.  Python code for this paper can be found at \url{https://github.com/Davidvandijcke/fdr} We thank Austin Wright as well as audiences at the Bank of England, 
Emory, MIT, Northwestern, Princeton, Simon Fraser, UIUC, 
the University of Michigan, Yale, 
the 2023 Munich Econometrics conference, 
the 2023 European Winter Meetings of the Econometric Society, 
and the 2024 ESIF Economics and AI+ML Meeting for helpful comments and suggestions.
All errors are our own.}}
\author{\normalsize David Van Dijcke$^{2,3}$}
\affil{\footnotesize $^{1}$Department of Economics, Emory University \\ 
$^{2}$Department of Economics, University of Michigan\\
$^{3}$Risk Analytics Division, Ipsos Public Affairs}
\date{\vspace{-1em} \small \today}
\begin{document}

\singlespacing

\thispagestyle{empty}
\maketitle
{\vspace{-2em}}
\begin{abstract} 
Sharp, multidimensional changepoints---abrupt shifts in a regression surface whose locations and magnitudes are unknown---arise in settings as varied as gene-expression profiling, financial covariance breaks, climate-regime detection, and urban socioeconomic mapping. Despite their prevalence, there are no current approaches that jointly estimate the location and size of the discontinuity set in a one-shot approach with statistical guarantees. We therefore introduce Free Discontinuity Regression (FDR), a fully nonparametric estimator that simultaneously (i) smooths a regression surface, (ii) segments it into contiguous regions, and (iii) provably recovers the precise locations and sizes of its jumps. By extending a convex relaxation of the Mumford–Shah functional to random spatial sampling and correlated noise, FDR overcomes the fixed-grid and i.i.d. noise assumptions of classical image-segmentation approaches, thus enabling its application to real-world data of any dimension.  This yields the first identification and uniform consistency results for multivariate jump surfaces: under mild SBV regularity, the estimated function, its discontinuity set, and all jump sizes converge to their true population counterparts. Hyperparameters are selected automatically from the data using Stein's Unbiased Risk Estimate, and large-scale simulations up to three dimensions validate the theoretical results and demonstrate good finite-sample performance. Applying FDR to an internet shutdown in
India reveals a 25–35\% reduction in economic activity around the estimated shutdown
boundaries-—much larger than previous estimates. By unifying smoothing, segmentation, and effect-size recovery in a general statistical setting, FDR turns free-discontinuity ideas into a practical tool with formal guarantees for modern multivariate data. \\

   
\noindent \\ 

\end{abstract}
\thispagestyle{empty} 
\newpage
\setcounter{page}{1} 

\clearpage
\onehalfspacing

\section{Introduction}

In many settings of interest, discontinuous changes in regression surfaces can reveal fundamental insights into the underlying problem structure, with both location and size of discontinuities typically unknown and dependent on multiple variables. Such discontinuities arise in many contexts, including algorithmic decision-making in finance and education, marketing segmentation, and tipping points in climate, urban sorting, financial markets, and spatial gene-expression profiling \citep{argyle2020monthly,  brunner2021effects, kuo2002integration, li2010contextual, lamberson2012tipping, card2008tipping, scheffer2009early, hansen2017regression, caetano2017school, aue2009break, chen2023scs, hou2023statistical}.

We introduce \emph{Free Discontinuity Regression} (FDR), a method that reconstructs the regression surface \emph{and} the exact geometry and magnitude of its jumps in one shot. FDR segments the domain into smooth and discontinuous regions without any point-wise smoothness assumption. In contrast to wavelet smoothing, fused-lasso variants, or gradient-threshold rules, our approach embeds smoothing and boundary detection in the same optimization, so the jump set and its sizes emerge as primary estimands rather than as results of a second-stage thresholding procedure, allowing us to develop the first formal identification guarantees for both jump location and size estimation in any dimension. 

Our starting point is the Mumford-Shah (MS) functional \citep{mumford1989optimal}, a workhorse model from mathematical computer vision and extension of classic discrete approaches \citep{geman1984stochastic}, which jointly segments and denoises while incorporating spatial regularity. The MS functional arises naturally as a solution to the limitations of simpler edge detection and thresholding methods \citep{canny1986computational, kass1988snakes}, which almost surely misclassify discontinuities in noisy settings \citep{chan2005image} and typically require multiple estimation steps, complicating statistical inference.\footnote{For empirical validation of the MS functional's superiority over simpler approaches, see \citet{wang2012robust, chan2001active, lucas2022hyperparameter, strekalovskiy2012convex}.} The MS functional is non-convex, so its critical points need not be global optima. This poses a threat to the reproducibility of
estimates and the ability to characterize convergence to the true population function. To resolve it, we employ the method of calibrations \citep{alberti2003calibration} to obtain a convex relaxation with global solutions, implementable via primal-dual algorithms \citep{chambolle2011first}. 

Most classical approaches to the MS functional treat the noise as already baked into \emph{fixed} image $f$. Although $f$ may deviate from the true underlying signal $u$, the analysis typically does not place a probability distribution over all possible noisy images; rather, $f$ is taken as a single, deterministic function. In contrast, our statistical viewpoint treats the observed function $Y_i$ as doubly random: the function values are observed with random noise, and live in a random point cloud. \citet{caroccia2020mumford} recently considered a similar setting for the Mumford-Shah functional on random graphs generated by random point clouds, where the noise is additive and i.i.d. In contrast, we work with a random grid on a random point cloud, and allow for the noise to be correlated. This is important in spatial settings like our application, and even in the classical image setting, where we can expect noise in blurry segments of the image to be correlated. Moreover, we provide sharp identification results in any dimension, in addition to $\Gamma$-convergence---the first results of their kind, to our knowledge. These complement and extend recent results for discrete approximations of the classical non-convex MS functional \citep{caroccia2020mumford, ruf2019discrete, morini2002global, richardson1992limit} and related functionals \citep{chambolle2021learning, garcia2016continuum, trillos2017new, hutter2016optimal, hu2022voronoigram}.

\begin{figure}[htbp!]
\centering
\caption{Simulations: 1D to 3D}
\label{fig:simulated_examples}
\vspace{-1em}
\begin{subfigure}[b]{0.8\linewidth}
\includegraphics[width=\textwidth]{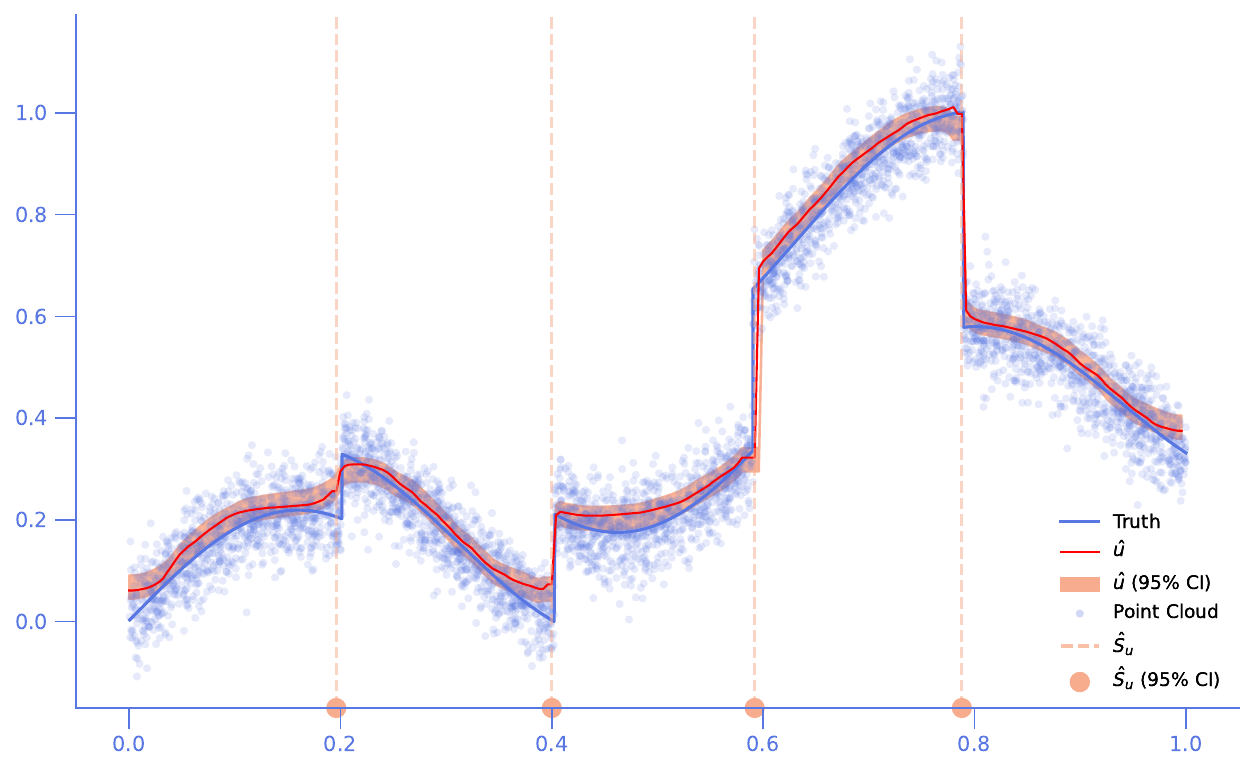}
    \caption{1D: Truth + Point Cloud + $\hat{u}$ + Boundary}
    \label{fig:sims_1D}
\end{subfigure}\hspace{1em}%
\begin{subfigure}[b]{0.24\linewidth}
\includegraphics[width=\textwidth]{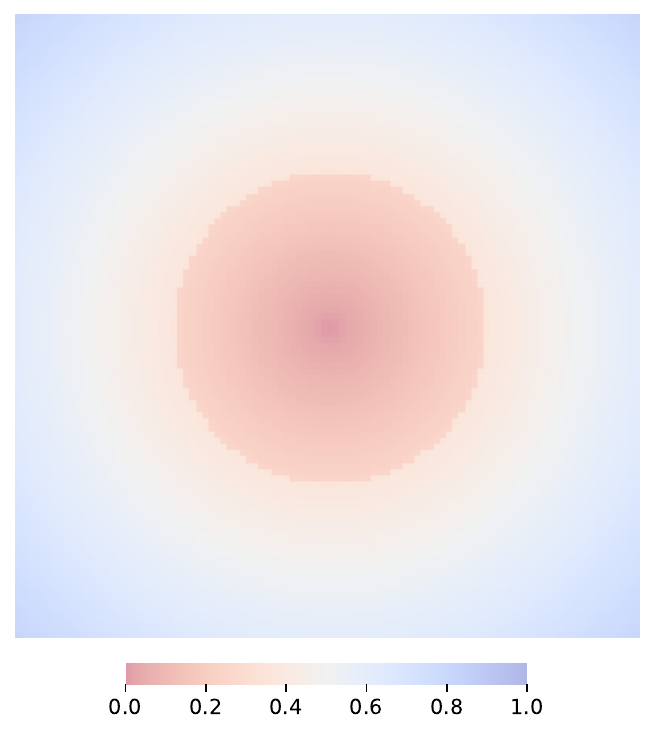}
    \caption{2D: Truth}
\end{subfigure}
\begin{subfigure}[b]{0.24\linewidth}
\includegraphics[width=\textwidth]{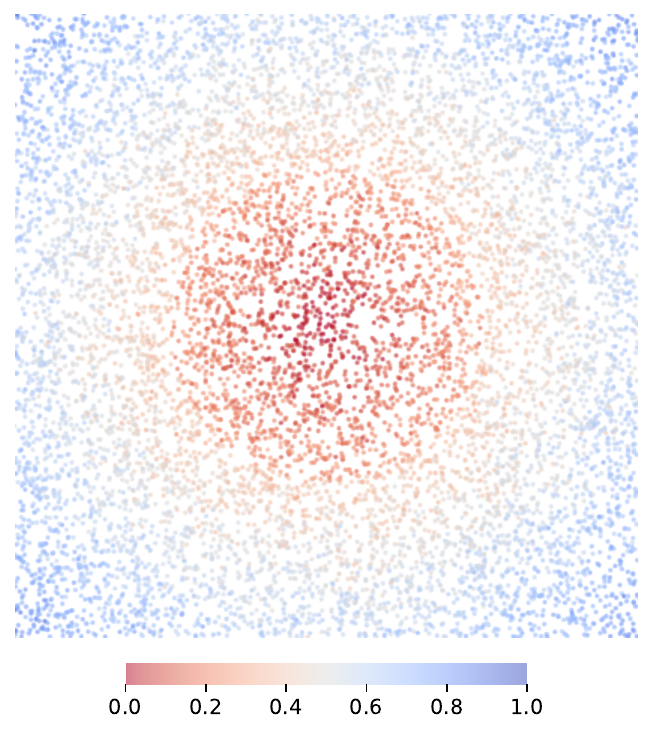}
    \caption{2D: Point Cloud}
    \label{fig:sims_2D}
\end{subfigure}
\begin{subfigure}[b]{0.24\linewidth}
\includegraphics[width=\textwidth]{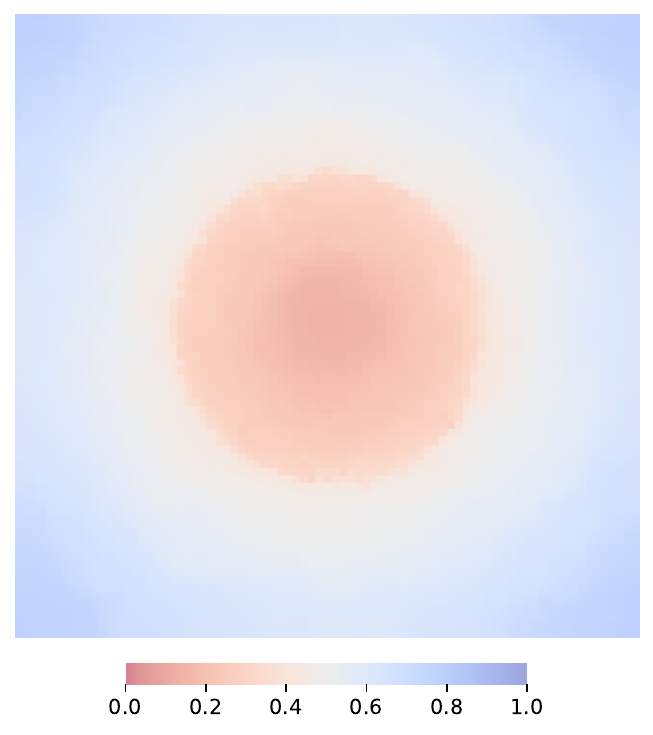}
    \caption{2D: $\hat{u}$}
\end{subfigure}
\begin{subfigure}[b]{0.24\linewidth}
\includegraphics[width=\textwidth]{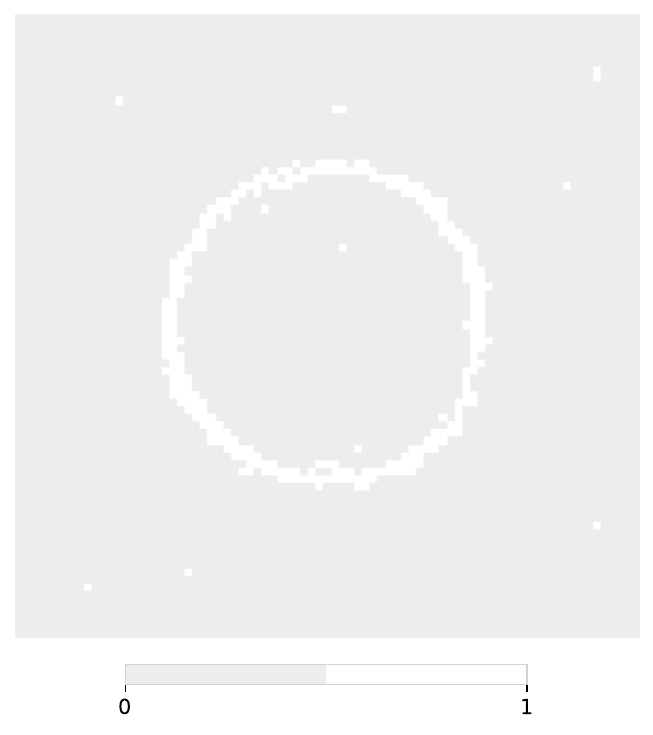}
    \caption{2D: Boundary}
\end{subfigure}\hspace{1em}%

\begin{subfigure}[b]{0.24\linewidth}
\includegraphics[width=\textwidth]{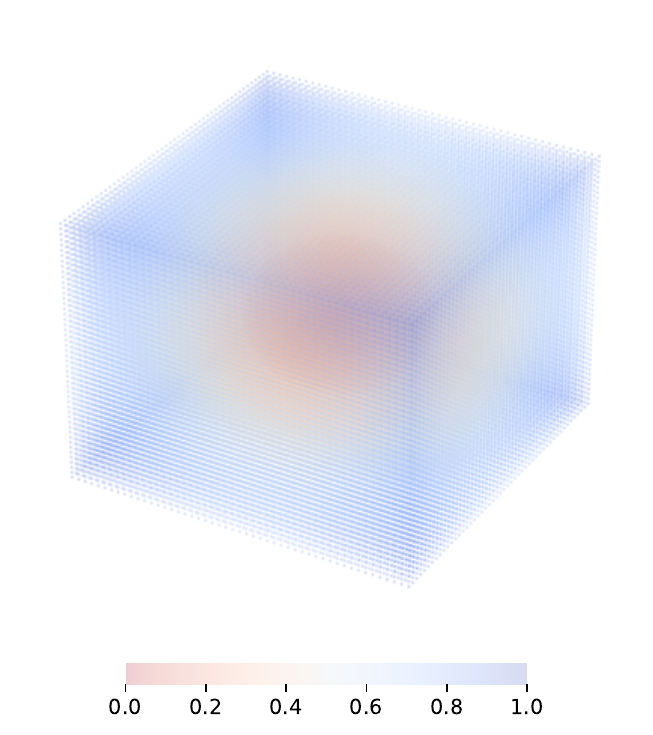}
    \caption{3D: Truth}
\end{subfigure}
\begin{subfigure}[b]{0.24\linewidth}
\includegraphics[width=\textwidth]{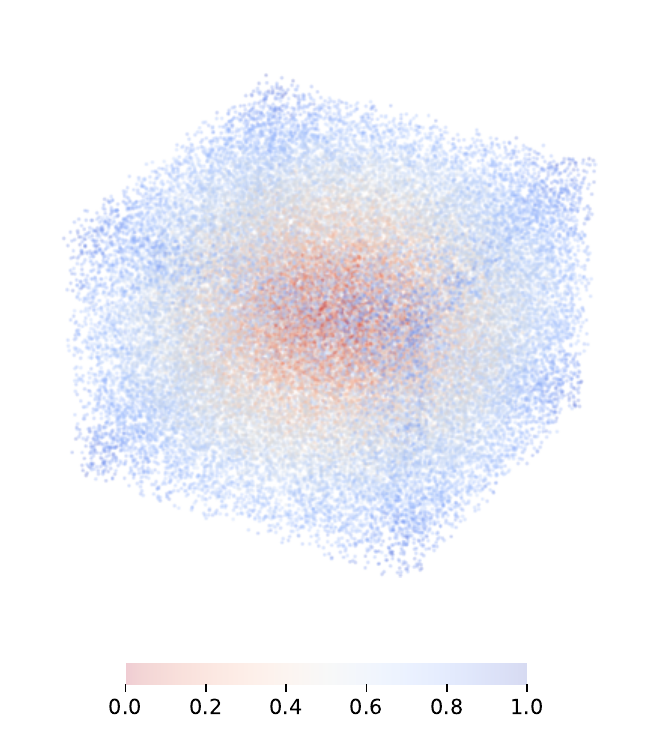}
    \caption{3D: Point Cloud}
\end{subfigure}
\begin{subfigure}[b]{0.24\linewidth}
\includegraphics[width=\textwidth]{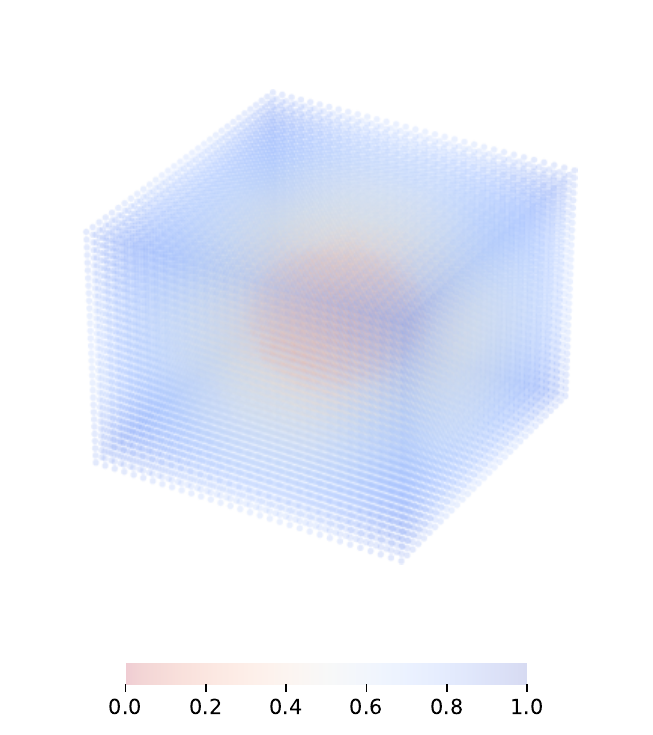}
    \caption{3D: $\hat{u}$}
\end{subfigure}
\begin{subfigure}[b]{0.24\linewidth}
\includegraphics[width=\textwidth]{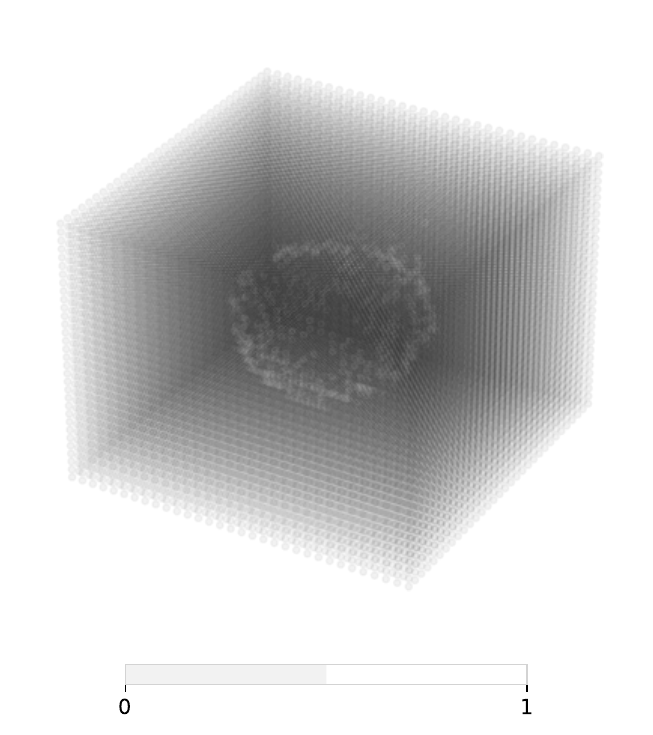}
    \caption{3D: Boundary}
\end{subfigure}
\floatfoot{\textit{Notes}: figures depict, for 1D--3D; the true underlying function without noise, the random point cloud with Gaussian noise ($\sigma= 0.05$), the estimated response function $\hat{u}$ and the estimated boundary locations. 1D case also includes red shaded 95\% confidence bands estimated by subsampling with an estimated rate of convergence \citep{politis1999subsampling}, and red dots below the plot to indicate jump sizes that are different from 0 at 95\% level. Hyperparameters for 1D, 2D chosen by SURE. \textit{Details}: \textbf{1D} jump sizes: 0.1286, 0.2133,  0.3192, -0.4220 ($d=0.2823,  0.4682, 0.7005, -0.9262$), n=5,000, $\lambda= 98.6712, \nu= 0.0001$. \textbf{2D} jump size $\alpha = 0.1126 \, (d=0.75),$ n=10,000; $\lambda=48.3921, \nu=0.0012$ ; \textbf{3D} jump size $\alpha=0.0738 \,(d=0.5),$ n=50,000; $\lambda=10, \nu=0.008$. }
\end{figure}

Figure \ref{fig:simulated_examples} demonstrates our method on 1D-3D simulated examples, showing accurate recovery of both functions and discontinuities from noisy data. We note that, while our main contribution lies in estimating multivariate discontinuities, our method also performs remarkably well in the one-dimensional case, handling multiple jumps of different sizes effortlessly without requiring user input on the number of jumps or the minimum jump size. Further below, we apply the method to estimate the effects of an internet shutdown in Rajasthan, India, in 2021. This one-day shutdown, implemented to prevent exam cheating, created sharp geographic discontinuities in connectivity. Using anonymized mobile activity around economic points of interest and cached location data from satellite-connected apps, we estimate a 25-35\% reduction in economic activity along the shutdown boundary. These findings substantially exceed previous estimates and highlight the asymmetry between internet expansion and shutdowns in modern economies, as well as the critical importance of reliable digital infrastructure for economic performance.

Our method advances several literatures. For a concise overview of how it compares to the main competing methods, see Table  \ref{tab:capabilities_piecewise_smooth}. The rich literature on change point detection methods typically focuses on discontinuity locations rather than magnitudes, or on the one-dimensional setting \citep{page1954continuous, muller1992change, braun2000multiple, killick2012optimal, porter2015regression, donoho1994ideal, harchaoui2008kernel}, while structural break estimation primarily addresses time series \citep{andrews1993tests, bai2003computation, delgado2000nonparametric}. Recent extensions to multivariate domains either impose structural assumptions \citep{park2022jump, herlands2019change, zhu2014spatially}, omit regression estimation \citep{madrid2022change}, require treatment status to be observed \citep{herlands2018automated}, or assume known discontinuity locations \citep{cheng2023regression, narita2021algorithm, abdul2022breaking, cattaneo2022regression, keele2015geographic}. Classical multivariate jump regression approaches \citep{qiu1997jump, korostelev1993minimax, o1994regularized, muller1994maximin, donoho1999wedgelets, li2017bayesian, qiu1998discontinuous} often require knowing partition counts or limit dimensionality. The fused lasso literature \citep{tibshirani2005sparsity, rinaldo2009properties, harchaoui2010multiple} has mostly focused on piecewise constant approximations. Recent approaches have made progress on estimating multivariate, piecewise-smooth regression surfaces with discontinuities, but do not develop any formal approaches for estimating the jump sizes and locations, instead relying on post-estimation thresholding without formal guarantees \citep{park2022jump, caroccia2020mumford, zhu2014spatially}. Finally, note that there is also a rich but distinct literature on ``multiple changepoints'' that deals with the setting where the \emph{outcome variable} is multivariate or functional \citep{aue2009break, matteson2014nonparametric, fryzlewicz2014wild, kovacs2023seeded, madrid2022change, dubey2020frechet, horvath2021monitoring}, as opposed to our setting where the \emph{regressors} are. A notable contribution in that literature is \citet{madrid2023change} who allow the multivariate output signals to be correlated and only satisfy mixing conditions, in line with our setting for a univariate output with multivariate domain.

Our approach, on the contrary, explicitly estimates the piecewise-smooth regression function jointly with its jump locations and sizes. Moreover, we provide formal convergence and identification results for both, which to our knowledge are the first of their kind for the general multivariate setting. Additionally, all existing multivariate estimation approaches assume that the sampling noise is i.i.d. across data points or pixels, while we allow it to be correlated under a general mixing process. This is crucial for frequently occurring spatiotemporally dispersed jump processes like climate regimes, geographic policy discontinuities, or financial time series. Finally, unlike many existing methods, which often work on a fixed lattice, we allow for randomness in both the locations and the sampling error of the data points. 

\begin{scriptsize}

\begin{table}[htbp]
\centering
\caption{Comparison of Statistical Methods for Multivariate Jump Surface Estimation}
\label{tab:capabilities_piecewise_smooth}
\footnotesize
\begin{tabular}{@{}lccccc@{}}
\toprule
\textbf{Method} &
\makecell{Jump\\Domain} &
\makecell{Jump\\Location} &
\makecell{Jump\\Size} &
\makecell{Jump\\Identification} &
\makecell{Correlated\\Sampling/Noise} \\
\midrule
Proposed Method\textsuperscript{1}          & Any $D$ & \checkmark & \checkmark & \checkmark & \checkmark \\
Wild Binary Segmentation\textsuperscript{2} & 1D      & \checkmark & $\times$   & \checkmark & $\times$   \\
PELT\textsuperscript{3}                     & 1D      & \checkmark & $\times$   & $\times$   & \checkmark \\
TV Denoising/Fused Lasso\textsuperscript{4} & Any $D$ & $\sim$\textsuperscript{a} & $\sim$\textsuperscript{a} & $\times$ & $\times$ \\
Trend Filtering on Graphs\textsuperscript{5} & Graph   & $\sim$\textsuperscript{a} & $\sim$\textsuperscript{a} & $\times$ & $\times$ \\
Change Surfaces\textsuperscript{6}          & Any $D$ & $\sim$\textsuperscript{b} & $\sim$\textsuperscript{b} & $\times$ & $\times$\textsuperscript{b} \\
Jump Gaussian Processes\textsuperscript{7}  & Any $D$ & $\sim$\textsuperscript{c} & $\sim$\textsuperscript{c} & $\times$ & $\times$ \\
SVCM\textsuperscript{8}                     & Any $D$ & $\sim$\textsuperscript{a} & $\sim$\textsuperscript{a} & $\times$ & $\sim$\textsuperscript{d} \\
MS on Graphs\textsuperscript{9}             & Graph   & $\sim$\textsuperscript{e} & $\sim$\textsuperscript{e} & $\times$ & $\times$ \\
\bottomrule
\end{tabular}

\vspace{0.6em}
\begin{minipage}{\linewidth}
\scriptsize
\textbf{Legend.} $\times$ = No; $\checkmark$ = Yes; $\sim$ = Partially or requires additional assumptions/extensions.\\
\textbf{Jump Domain} denotes the dimension of the Euclidean point cloud on which the function is observed (or “Graph’’ if defined on a network);  
\textbf{Jump Location/Size} records whether those quantities are explicitly estimated;  
\textbf{Jump Identification} flags formal convergence results;
\textbf{Correlated Sampling/Noise} indicates support for correlated noise in the model.\\[4pt]
\textsuperscript{a} Requires post-estimation gradient thresholding; no explicit discussion of discontinuity identification.\\
\textsuperscript{b} Models smooth transition boundaries with Gaussian processes; no explicit discussion of discontinuity identification.\\
\textsuperscript{c} Estimates local (not global) piecewise-continuous regressions.\\
\textsuperscript{d} Errors are i.i.d.\ samples from a stochastic process; input locations are i.i.d.\\
\textsuperscript{e} Jump recovery done via manual gradient thresholding.\\[4pt]
\textsuperscript{1}\cite{gunsilius2023free}\,
\textsuperscript{2}\cite{fryzlewicz2014wild}\,
\textsuperscript{3}\cite{killick2012optimal}\,
\textsuperscript{4}General version in \cite{hutter2016optimal}\,
\textsuperscript{5}\cite{wang2016}\,
\textsuperscript{6}\cite{herlands2019change}\,
\textsuperscript{7}\cite{park2022jump}\,
\textsuperscript{8}\cite{zhu2014spatially}\,
\textsuperscript{9}\cite{caroccia2020mumford}.
\end{minipage}
\end{table}
\end{scriptsize}

\section{Free Discontinuity Regression}
In this section, we develop the statistical framework for estimating regression functions with unknown discontinuities of any dimension. All mathematical notation is detailed in Appendix \ref{app:notation}.
\subsection{Regression framework}\label{sec:RDD_as_FDR}
We consider $n$ randomly sampled units indexed by $i=1,2,\ldots,n$ for which we observe a potentially multivariate regressor $X_i\in \mathbb{R}^d$ and an outcome of interest $Y_i\in\mathbb{R}$. The regression model is
\begin{equation}\label{eq:model}
Y_i = f(X_i) +\varepsilon_i,\qquad \mathbb{E}[\varepsilon_i\vert X_i] = 0,
\end{equation}
where the $\varepsilon_i$ are unobservable error terms. Our goal is to estimate both the regression surface $f(x)$ and its discontinuity set $S_f$.
To handle multivariate regressors, we need to extend the notion of a univariate discontinuity to multiple dimensions. Following Definition \ref{def:jumppoints} in Appendix \ref{app:notation} \citep[Definition 3.67]{ambrosio2000functions}, we characterize jump points of a function $u$ on $\mathcal{X}$ through two distinct limiting values, $u^{+}(x)$ and $u^{-}(x)$, called \emph{traces}. These traces reflect how $u$ behaves as one approaches a point $x$ in the domain from opposite sides, with an orientation vector $\rho$ defining these ``sides." The collection of all such jump points forms the \emph{(approximate) discontinuity set} $S_u$. 

\subsection{FDR as an estimator based on the Mumford-Shah Functional}
Standard approaches to this estimation problem face several fundamental limitations. Smoothing methods do not explicitly estimate discontinuity locations or sizes, but merely approximate the regression surface while preserving discontinuities. This leads to downward-biased jump size estimates in practice \citep[see e.g.][Fig. 7.1f]{caroccia2020mumford}. Simple thresholding approaches that estimate jump locations via gradient steepness almost surely misclassify discontinuities in noisy settings and fail to take into account the spatial regularity of the function that produces the jumps, resulting in severely biased estimates and high false positive and negative rates for jump detection \citep{chan2005image}. The Mumford-Shah (MS) functional \citep{mumford1989optimal} addresses these issues by combining smoothing and thresholding in a unified framework.
\subsubsection{The Mumford-Shah Functional}
Assuming the random variable $X$ has density $f_X$, our statistical version of the MS functional takes the form:
\begin{multline} \label{eq:MS}
E(u)=\lambda \underbrace{\int_{\mathcal{X}}(f(x)-u(x))^2 f_X(x)\dd x}_{\text{Regression}} + \underbrace{\int_{\mathcal{X} \backslash S_u}|\nabla u(x)|^2 f_X(x)\dd x}_{\substack{\text{Roughness  penalty} \\  \text{away from discontinuity}}} + \underbrace{\nu \mathscr{H}^{d-1}\left(S_u\right)}_{\substack{\text{Discontinuity regularity} \\ \text{ penalty}}}. \end{multline}
The functional combines three components:
1) A regression term that minimizes the mean-squared error between the true function $f$ and its approximation $u$; 2) a roughness penalty that controls the gradient size of $u$ away from the discontinuity set $S_u$; 3) a penalty on the ``size" of the discontinuity set, measured by its $(d-1)$-dimensional Hausdorff measure $\mathscr{H}^{d-1}$. The goal is to estimate the minimizer $u$ of this function, which immediately delivers an estimate of the discontinuity set $S_u$ as well as the jump sizes $u^+(x) - u^-(x)$ for each jump point $x \in S_u$. 

The hyperparameters $\lambda, \nu\geq0$ control the balance between function smoothness and boundary complexity, with the gradient term's weight normalized to 1. The squared Euclidean norm in $\mathbb{R}^d$ is denoted by $|\cdot|^2$.
\subsubsection{Convexification Through Calibrations}
The MS functional's non-convexity through its third term threatens consistent estimation, as optimization routines need not converge to global optima. In practice, this manifests as spurious discontinuities and boundaries (\citealp[Fig. 4-5]{brown2012completely}; \citealp[Fig. 5]{pock2009algorithm}). A common approach to non-convex optimization is to re-estimate the model multiple times with different initializations and select the solution that minimizes the objective function. However, this approach becomes computationally infeasible for functional minimands with dimension larger than 1, as the set of potential initializations grows too large. Moreover, it is not guaranteed to converge to the true global minimum.
We resolve this using calibrations \citep{alberti2003calibration, pock2010global}, which lift the problem to a higher dimension by considering the graph $\Gamma_f$ of $f$. The key is finding a vector field $p:\mathbb{R}^{d+1}\to\mathbb{R}^{d+1}$ on the lifted space that maximizes flux through the graph. This yields the convex relaxation:
\[
E(u)=\sup _{p \in K} \int_{\mathcal{X} \times \mathbb{R}} p \cdot D \mathds{1}_u
\]
with a convex constraint set
\begin{multline*}
K= \left \{p \in C_0\left(\mathcal{X} \times \mathbb{R}, \mathbb{R}^{d+1}\right):\right.\\ p^t(x, t) \geq \frac{\left\lvert p^x(x, t)\right\rvert^2}{4f_X(x)}-\lambda f_X(x)(t-f(x))^2, \quad
 \left|\int_{t_1}^{t_2} p^x(x, s) d s\right| \leq \nu, \thickspace \forall t_1, t_2 \in \mathbb{R}\bigg\}. 
\end{multline*}
Here, $t$ indexes the lifted dimension, and $\mathds{1}_u(x,t)$ is the indicator function of the subgraph of $u(x)$, taking value 1 if $t<u(x)$ and 0 otherwise. The measure $D\mathds{1}_u\coloneqq (D_1\mathds{1}_u,\ldots, D_{d+1}\mathds{1}_u)$ is $(d+1)$-dimensional Radon measure, with $p^x$ and $p^t$ denoting the first $d$ and last components of $p\in C_0\left(\mathcal{X} \times \mathbb{R}, \mathbb{R}^{d+1}\right)$ respectively.
To also achieve convexity over the target function space, we follow \citet{pock2009algorithm} and optimize over the more general space of functions mapping to the unit interval:
\begin{equation*}
v(x,t)\in  C \coloneqq \bigg\{ v\in SBV(\mathcal{X}\times\mathbb{R},[0,1]): \lim _{t \rightarrow-\infty} v(x, t)=1, \lim _{t \rightarrow+\infty} v(x, t)=0 \bigg\}.
\end{equation*}
Here, $SBV(\mathcal{X}\times\mathbb{R},[0,1])$ denotes the space of special functions of bounded variation mapping to the unit interval. This class is particularly well-suited to our problem as it naturally accommodates functions with discontinuities while still being small enough to be controllable.
\begin{definition}\label{def:FDR}
The Free Discontinuity Regression estimator is the $0.5$-level set of the optimizer $v^*$ of the optimization problem
\begin{equation} \label{eq:main_problem}
\inf_{v \in C} E(v) \coloneqq \inf_{v\in C}\sup _{p \in K}\langle p, Dv\rangle\equiv \inf _{v \in C} \sup _{p \in K} \int_{\mathcal{X} \times \mathbb{R}} p \cdot D v.
\end{equation}
\end{definition}
This optimization problem is well-behaved due to its convexity and admits a solution under mild conditions:
\begin{proposition}\label{prop:existence}
The optimization problem \eqref{eq:main_problem} admits a global solution $v^*$ in $SBV(\mathcal{X}\times \mathbb{R})\cap\{v: |Dv|\leq c\}$ for fixed $c<+\infty$ if $f\in SBV(\mathcal{X})$.
\end{proposition}
While the solution $v^*$ need not take the form of an indicator function---which is required for the relaxation to be tight---we prove that it converges to one as $\lambda\to +\infty$ in a certain topology (Theorem \ref{thm:ident}). In practice, our automatic hyperparameter selection procedure selects values of $\lambda$ that are many orders of magnitudes larger than the selected $\nu$, and as such also leads to estimated solutions that are indicator functions in practice. Following \citet{pock2009algorithm}, we obtain the underlying function $u$ by thresholding $v$ at the $0.5$-level set. When $v$ is an indicator function, any threshold $t \in (0,1)$ would yield equivalent results; the midpoint is natural when the solution deviates slightly from an indicator.

\section{Statistical Properties}\label{sec:stats}

This section establishes the mathematical and statistical properties of the estimator. We first show that the solution to \eqref{eq:main_problem} recovers both the true discontinuity set and jump sizes as $\lambda\to+\infty$ with fixed $\nu>0$. We then prove consistency of our empirical estimator via $\Gamma$-convergence \citep{dal2012introduction}. All proofs appear in Appendix \ref{sec:proofs}.

We also provide practical tools for implementation: a data-driven method for selecting tuning parameters $\lambda,\nu$ using Stein's unbiased risk estimate (SURE), and two approaches to uncertainty quantification. The first uses subsampling with estimated convergence rates \citep{politis1999subsampling} and should be used for estimation procedures, while the second employs split conformal inference \citep{lei2018distribution} for computationally efficient uncertainty quantification in high-dimensional prediction. 

\subsection{Identification of the Discontinuity Sets and Jump Sizes}\label{sec:identification}

We first establish that the true function $f$ is identified under mild regularity conditions. Specifically, we show that for fixed $\nu>0$ and as $\lambda\to+\infty$, the solution to \eqref{eq:main_problem} recovers both the correct discontinuity set and jump sizes. This complements and extends previous results \citep{richardson1992limit,morini2002global} which established convergence of the discontinuity set for the classical non-convex Munford-Shah functional. This result is the first to provide mathematical and statistical properties for the convexified problem.

We require two assumptions, following \citet{richardson1992limit}. Let $S_f$ denote the discontinuity set of $f$, which includes the jump set $J_f$ and coincides with it $\mathscr{H}^{d-1}$-almost everywhere by the Federer-Vol'pert theorem \citep[Theorem 3.78]{ambrosio2000functions}.

\begin{assumption}\label{ass:ident1}
The density $f_X$ is bounded above and below everywhere on its support $\mathcal{X}$, i.e., $0<c_x\leq f_X(x)\leq \frac{1}{c_x}<+\infty$ for all $x\in\mathcal{X}$. Further, $f\in SBV(\mathcal{X})$ and there is a constant $c>0$ such that $|f(x)|\leq c$ for $\mathcal{L}^d$-almost every $x\in\mathcal{X}$. Moreover, $\int_{\mathcal{X}}\left\lvert\nabla f\right\rvert^2\dd x +\mathscr{H}^{d-1}(S_f) <+\infty$.
\end{assumption}

\begin{assumption}\label{ass:ident2}
For any $x\in S_f$ it holds that $\mathscr{H}^{d-1}\left(S_f\cap B_\rho(x)\right)>0$ for all $\rho>0$. Moreover, for any set $A\subset\mathcal{X}$ with $\text{dist}(A,S_f)>0$ there exists a constant $L>0$ such that $|f(x)-f(y)|\leq L|x-y|$ for any $x,y\in A$. 
\end{assumption}

Assumption \ref{ass:ident1} imposes standard boundedness on the functions and their jump sizes. Assumption \ref{ass:ident2} ensures regularity both at and away from the discontinuity set. The first part implies connectedness of the discontinuity set when measured with $\mathscr{H}^{d-1}$ on $\mathcal{X}$ (e.g., precluding isolated points in two dimensions). This ensures $\mathscr{H}^{d-1}((\mathcal{X}\cap \bar{S}_f)\setminus S_f) = 0$, meaning the closure $\bar{S}_f$ coincides with $S_f$ $\mathscr{H}^{d-1}$-almost everywhere \citep[p.~337]{ambrosio2000functions}.

\begin{theorem}\label{thm:ident}
Let Assumptions \ref{ass:ident1} and \ref{ass:ident2} hold. Then for fixed $\nu>0$ and in the limit as $\lambda\to+\infty$ every sequence of solutions $v^*(\lambda)$ to \eqref{eq:main_problem} satisfies $\lim_{\lambda\to+\infty}\nabla v^*(\lambda) = 0$ $\mathcal{L}^{d+1}$-almost everywhere. Moreover, the jump set $J_{v^*}(\lambda)$ converges in Hausdorff distance $d_H$ to the graph $\Gamma_f$, i.e.
\begin{equation*}
\lim_{\lambda\to+\infty} d_H\left(J_{v^*}(\lambda),\Gamma_f\right) = 0.
\end{equation*} 
\end{theorem}

The proof of \ref{thm:ident} requires $\lambda$ to diverge ``fast enough'' based on the countable rectifiability \citep{mattila1999geometry} of the jump set: for any given sequence $\{r_m\}$ converging to 0 defining the lower density \citep{ambrosio2000functions} of the jump set, we need $\{\lambda_m\}$ such that $\lambda_mr_m\to+\infty$. This is not a restriction since $\lambda$ can be arbitrarily large. In practice, $\lambda$ must substantially exceed $\nu$ to identify correct jump sizes. Our hyperparameter selection procedure consistently finds such values, with $\lambda$ orders of magnitude larger than $\nu$, yielding convergent estimates of both jump set and sizes as $N$ grows. We now turn to the practical implementation and the corresponding estimator.

\subsection{The Empirical Estimator and Convergence}\label{sec:consistency}

We now develop the empirical version of the estimator for observed data. Consider a random sample $\{X_i,Y_i\}_{i=1}^n$ from distribution $P_{Y,X}$, with errors $\varepsilon_i$, which can be correlated and not identically distributed, as long as they satisfy a law of large numbers. This is important for the spatial setting in our application, where the error terms in the nonparametric regression are correlated in general. Sufficient for this are mixing conditions. In the following, to guarantee consistency of the kernel density estimator $\hat{f}_{X,h(n)}$, we assume that the marginal data generating process $\{X_i\}_{i\in\mathbb{N}}$ is stationary. This can be relaxed \citep{irle1997consistency}.

The key challenges in moving to random data are twofold: 
\begin{enumerate}
    \item[(i)] First, the random location of the points. Unlike computer vision applications \citep{pock2009algorithm} where values lie on a fixed lattice, our observations occur at arbitrary points in the domain. While recent approaches handle this using nearest-neighbor graphs \citep{caroccia2020mumford, chambolle2017accelerated, chambolle2021learning}, such methods are less suitable here as our constraint set $K$ requires explicit dimension-wise operations, particularly summation over one dimension.
    \item[(ii)] The second source of randomness is the error term $\varepsilon_i$. In contrast to existing image segmentation approaches, we allow the terms $\varepsilon_i$ to be dependent, that is, a strongly mixing process. 
\end{enumerate}

To accommodate the random locations, we discretize \eqref{eq:main_problem} on a regular grid following \citet{pock2009algorithm}. Specifically, we overlay an $(N_1 \times N_2 \times \ldots \times N_d) \times S$ pixel grid on the random point cloud $\{X_i\}_{i=1}^n$, defining
$
\mathcal{Q}_N=\{(k_1,\ldots, k_d, s): k_i=1,2, \ldots, N_i; \ldots; k_d = 1,\ldots N_d; s=1,2, \ldots, S\}.
$
To ease notation, we assume without loss of generality that all dimensions are discretized in the same manner, $s\coloneqq k_{d+1}$, $N_1,\ldots,N_{d+1} = N$. Also, denote $k\coloneqq (k_1,\ldots,k_{d+1})$, and $k_{-}\coloneqq (k_1,\ldots,k_d)$. The empirical analogue of $v\in C$ becomes
\[\hat{v}_{Nn}(x)\equiv \hat{v}_{Nn}(\tilde{x},t)\coloneqq \sum_{1\leq k_1,\ldots,k_{d+1}\leq N}v_{k_1,\ldots,k_{d+1}}\mathds{1}\left\{x\in Q_{k_1,\ldots,k_{d+1}}\right\}\]
where $v_k$ is assigned to the center of cube $Q_k$.  We define the empirical analogue $p_N$ of $p\in K$ to take values on the boundary of the respective cubes. That is, for each $j=1,\ldots, d+1$, the corresponding value 
$p_{k_1,\ldots,k_{j}+\frac{1}{2},\ldots,k_{d+1}}$ lies on the boundary $\partial Q_{k_1,\ldots,k_{j},\ldots,k_{d+1}}\cap \partial Q_{k_1,\ldots,k_{j}+1,\ldots,k_{d+1}}$.
This preserves duality in the empirical setting, as the empirical analogues of $p$ are defined on boundaries of the cubes, while the analogues of $v$ are defined in the center \citep{chambolle2021learning}. Problem \eqref{eq:main_problem} becomes
\begin{equation} \label{eq:main_problem_discrete}
\begin{aligned}
&\min _{v \in \tilde{C}_N } \hat{E}_{Nn}(v) \coloneqq  \min_{v \in \tilde{C}_N} \max _{p \in \hat{K}_{Nn}}  \langle p, D_Nv\rangle_N  
\end{aligned}
\end{equation}
with
\begin{equation*}
\begin{aligned}
&\tilde{C}_N=\left\{v: v(k) \in[0,1], v(k_1, \ldots, k_d, 1)=1, v(k_1, \ldots, k_d, N)=0\right\}\\
&\hat{K}_{Nn} =  \left\{p=\left(p^x, p^t\right)^T :   
 p^t(k) \geq \frac{\left\lvert p^x(k)\right\rvert^2}{4\thickspace \hat{f}_{X,h(n)}(k_{-})}-  \lambda\thickspace \hat{f}_{X,h(n)}(k_{-})\thickspace\left(\frac{k_{d+1}}{N}-\hat{f}_{Nn}(k_{-})\right)^2,\right.\\
 &\hspace{2cm}\left.\left|\frac{1}{N}\sum_{s_1 \leq k_{d+1} \leq s_2} p^x(k)\right| \leq \nu\right\},
\end{aligned}
\end{equation*}
where $D_Nv$ is the normalized forward difference and $\langle p, Dv\rangle_N$ the scalar product between two vectors,
\[\langle p_N,D_Nv_N\rangle_N =\sum_{0\leq k_1,\ldots, k_{d+1}\leq N} N\left(v^\uparrow_{k_1,\ldots, k_{d+1}} - v_{k_1,\ldots, k_{d+1}}\right) p_{k_1,\ldots,k_{d+1}}, \]
with $v^\uparrow_{k_1,\ldots,k_j,\ldots, k_{d+1}}\coloneqq v_{k_1,\ldots,k_{j}+1,\ldots,k_{d+1}}$ and $p_k^\uparrow = p_k\cdot e^\uparrow$. The constraints hold for all grid points $k$, and $s_1 \leq k_{d+1} \leq s_2$ iterates over indices with $1\leq s_1 \leq k_{d+1} \leq s_2 \leq S$.

For function estimation, we use weighted averages within each grid cube:
\[\hat{f}_{Nn}(\bar{x}_{k_{-}})=\sum_{i: X_i\in Q_{k_{-}}} w_i Y_i\] 
where $\bar{x}_{k_{-}}$ is the cube center and weights $w_i$ sum to 1; for simplicity, we assume uniform weights, i.e., $w_i = 1 / \#\left\{i: X_i\in Q_k\right\}$. We estimate the density $f_x$ via a standard Nadaraya-Watson estimator 
\[\hat{f}_{X,h(n)} = \frac{1}{nh(n)^d}\sum_{i=1}^n K\left(\frac{X_i-x}{h(n)}\right),\] where  $h(n)$ is a bandwidth and $K(\cdot)$ is a multivariate kernel. We assume that $K:\mathbb{R}^d\to\mathbb{R}$ satisfies $K(x\rho) \geq K(x)$ for all $x\in\mathbb{R}^d$ and $0\leq\rho\leq 1$, $\int_{\mathbb{R}^d} K(x) dx = 1$, $K(x) = O(|x|^{-n-\rho})$ as $|x|\to+\infty$ for some $\rho>0$, and that $\sup_{x\in\mathbb{R}^d} K(x) <+\infty$. In the following, we call such a kernel \emph{regular} \citep{irle1997consistency}.

We make the following assumption on the data-generating process.
\begin{assumption}\label{ass:stats}
Let $(\Omega, \mathcal{A}, P)$ be a probability space and let the following hold:
\begin{enumerate}
\item[(i)] (Mixing assumption and sampling frequency)\\ The number of grid cubes $Q_k$ grows as $N^d(n) = O(n^m)$ for some $0<m<1$. Furthermore, $(Y_i,X_i,\varepsilon_i):\Omega\to \mathbb{R}\times\mathbb{R}^d\times\mathbb{R}$ is an $\alpha$-mixing sequence with $\alpha(u) = O(u^{-s})$ with $s>\max\left\{\frac{r+m}{1-m},\frac{r+1}{\frac{2v}{2v+1}(1-m)}\right\}$ for some $r>1$ and $v\geq 3$. The marginal process $\{X_i\}_{i=1}^n$ is strictly stationary. The penalty term $\nu$ is fixed, and $\lambda_n=o(N(n))$.

\item[(ii)] (Moment conditions)\\
$\mathbb{E}\left[\varepsilon_i \vert X_i\right]=0$ and the moment generating function $\mathbb{E}[\exp(t\varepsilon_i)]$ exists in a neighborhood of zero; also, $\mathbb{E}[|Y_i|^{r+\delta}]<M<+\infty$ for all $i$ and some $\delta>0$. The $X_i$ have a common marginal density $f_X$, which has compact support $\mathcal{X}$ and is bounded away from zero and infinity, i.e., $0< c_x\leq f_X(x)\leq \frac{1}{c_x}<+\infty$ for all $x\in\mathcal{X}$. 

\item[(iii)] (Regularity of kernel and histogram estimator)\\ Let $K:\mathbb{R}^d\to\mathbb{R}$ be a regular kernel and $h(n)\to0$ with $nh(n)^d\to+\infty$ as $n\to+\infty$ and $\sum_{i=1}^\infty (i^2 h(i)^{d+d/l})^{-1} < +\infty$, where $0<l<s$. The weights $w_i$ in the histogram estimator take the form $w_i = 1/ \sum_{i=1}^n 1(X_i\in Q_k)$.
\end{enumerate}
\end{assumption}
These regularity assumptions are standard. Part (i) requires a bound on the mixing parameter $s$ that is very minor, since $r>1$ can be as close to unity as we want. It involves the rate $m$ at which the number $N^d$ of grid points grows with the number of data points $n$. The slower this growth, the smaller the mixing constant can be. The strict stationarity assumption on $\{X_i\}_{i=1}^n$ is done for convenience in the proof and can be relaxed. Further, we allow the penalty term $\lambda_n$ to vary with the data, but very slowly, in order to still obtain the almost sure convergence. The moment conditions in Part (ii) are very minor and allow for unbounded error terms $\varepsilon_i$ as long as their tails decay quickly enough. The regularity for kernel and histogram estimators in Part (iii) are standard and taken from the kernel estimator results in \citet{irle1997consistency}.

Before proving the full statistical result, we provide a consistency result for the simpler setting where we pretend $f(x)$ is known and does not have to be estimated. This result is novel in the image recognition literature and complements recent convergence results in this area \citep[e.g.][]{chambolle2021learning}. It also showcases the complication when introducing the two sources of randomness.
\begin{proposition}\label{prop:consistency_det}
Let Assumption \ref{ass:stats} hold, let $\hat{f}_{Nn}\equiv f$ be a fixed function whose values are given in the center $\bar{x}_k$ of each cube $Q_k$, and set $f_X(x)\equiv 1$ everywhere on $\mathcal{X}$. Then, for fixed penalty terms $\nu,\lambda$, $E_N(v)$ $\Gamma$-converges in the weak$^*$-topology to  
\[E(v)\coloneqq \begin{cases} \sup_{p\in K} \int_{[0,1]^{d+1}} p\cdot Dv&\text{if $v\in C$}\\ +\infty&\text{else} \end{cases} .\]
\end{proposition}

Now we add the randomness to $f$ as well as $f_X$. The fact that $f$ has to be estimated through random draws complicates the estimation procedure significantly. In fact, any estimator $\hat{f}_{Nn}$ for $f$ that averages out the error caused by $\varepsilon_i$, including the proposed histogram estimator, converges to a precise representative $\tilde{f}$ of $f$. This means that it is consistent only for points $x\not\in S_f$. For points $\tilde{x}\in S_f$, $\hat{f}_{Nn}(\tilde{x})$ converges almost surely to $\tilde{f}(\tilde{x}) = (1-\theta) f^+(\tilde{x}) +\theta f^{-}(\tilde{x})$, that is, a convex combination of the values of the traces $f^+(\tilde{x})$ and $f^-(\tilde{x})$. The weight $\theta$ depends on the angle of the hyperplane induced by the approximate differential and the weighting from $f_X$. If all $Q_k$ are cubes and $f_X$ is the density of the uniform distribution, then $\theta=1/2$. Importantly, $0<\theta<1$, which implies that the proposed estimator identifies the location of the jumps of $f$ perfectly. This is captured in the following result.

\begin{theorem}\label{thm:consistency_rand}
Let Assumptions \ref{ass:ident1} -- \ref{ass:stats} hold. Then $\hat{E}_{Nn}(v)$ $\Gamma$-converges in the weak$^*$-topology almost surely to
\[\tilde{E}(v)\coloneqq \begin{cases} \sup_{p\in \tilde{K}}\int_{\mathcal{X}\times\mathbb{R}} p\cdot Dv&\text{if $v\in C$}\\ +\infty&\text{else}, \end{cases}\] 
where $\tilde{K}$ is the same as the constraint $K$, but $f$ replaced by 
\[\tilde{f}(x) = \begin{cases} f(x) &\text{if $x\not\in S_f$} \\ (1-\theta) f^+(x) + \theta f^-(x) &\text{otherwise,}\end{cases}\]
where 
\[\theta = \frac{\int_{Q_kN\cap H^+(x,\rho)} f(x) f_X(x) dx}{\int_{Q_{kN}} f_X(x)dx}\] with $\{Q_{kN}\}_N$ the sequence of grid cubes containing and converging to $\{x\}$ and $H^+(x,\rho) = \{y: \langle y-x,\rho\rangle\geq0\}$ a half space created by the hyperplane $H(x,\rho)$ at $x$ oriented by $\rho \in S^{d-1}$ and separating the traces $f^+(x)$ and $f^-(x)$. 
\end{theorem}
Complementing this with a compactness result yields convergence of the minimizer:

\begin{corollary}\label{corr:consistency_rand}
    Let the assumptions from Theorem \ref{thm:consistency_rand} hold and let $\hat{v}_{Nn}$ be minimizers of $\hat{E}_{N,n}(v)$ with $v\in C\cap \{v\in SBV(\mathcal{X}\times\mathbb{R}): |Dv|\leq 1\}$. Then $\hat{v}_{Nn}\to \tilde{v}$ in the weak$^*$-topology almost surely, where $\tilde{v}$ is a minimizer of the deterministic population objective $\tilde{E}(v)$.  
\end{corollary}

This shows that the estimator is consistent for a precise representative of $f$, identifying the correct jump locations as $\lambda_n\to+\infty$. As shown in the simulations in Figure \ref{fig:simulated_examples}, the proposed estimator also allows to identify the correct jump size, not just the location. The key for this is the weak$^*$-convergence of the optimizer as proved in Theorem \ref{thm:consistency_rand}: by the closure of SBV \citep[Theorem 4.7]{ambrosio2000functions}, this implies that the traces of $v_n$ converge $\mathcal{H}^{d-1}$-almost everywhere to the traces of $v$. Since the traces are identified, this also identifies the jump sizes directly. Finally, standard convergence results for level-sets \citep[e.g.][]{camilli1999note} imply that the $0.5$-level set, which is our estimator for the true function, also converges almost surely. 

In practice, following the vision literature \citep{strekalovskiy2014real},
we detect the empirical discontinuity set by thresholding the estimated
gradient,
\[
   \hat S_u \;=\;
   \bigl\{x\in\mathcal X \;:\;
          \|\nabla\hat u_{Nn}(x)\|\;\ge \sqrt{\nu}\bigr\}.
\]
That gives us the estimated jump locations. To estimate the corresponding jump sizes, for every boundary voxel
\(x\in\hat S_u\), we march in the estimated outward normal direction
\(\hat{\rho}_{N_n}(x)\) and in the inward direction \(-\hat{\rho}_{N_n}(x)\) until we reach the
first voxels \(x_{N_n}^{+}(x)\) and \(x_{N_n}^{-}(x)\) that lie
\emph{outside} the boundary mask \(\hat S_u\).
The empirical jump–size estimator is
\[
   \widehat{\Delta}_{N_n}(x)
   \;=\;
   \hat u_{N_n}\!\bigl(x_{N_n}^{+}(x)\bigr)
   \;-\;
   \hat u_{N_n}\!\bigl(x_{N_n}^{-}(x)\bigr),
\]
i.e.\ the difference between the ``pure'' function values immediately to
the right and left of the interface.  Because the march follows
\(\hat{\rho}_{N_n}(x)\), the construction automatically captures the correct
orientation, so no grid refinement is needed. Then, Corollary \ref{corr:consistency_rand} implies convergence of the empirical normal $\hat{\rho}_{N_n}(x)$ to the population normal to the jump set, and hence of the estimated jump sizes to their population counterparts.

Finally, we construct confidence bands using subsampling \citep[Ch.8]{politis1999subsampling}, which we choose over bootstrap methods due to the estimator's non-smooth nature. The Online Appendix provides implementation details and describes a more computationally efficient alternative using conformal prediction.

\subsection{Data-Driven Choice of Hyperparameters by SURE}

While Theorem \ref{thm:ident} establishes that the population estimator recovers the true jump locations and sizes as $\lambda$ grows for fixed $\nu$, optimal hyperparameter choice in finite samples depends on the data. We select $\lambda,\nu$ by minimizing Stein's unbiased risk estimate (SURE) \citep{stein1981estimation}, which provides an asymptotically unbiased estimate of mean-squared error under Gaussian errors.

Given estimator $\hat{u}_\theta(Y)$ with hyperparameters $\theta = (\lambda, \nu)$, the Stein estimator is
\begin{equation} \label{eq:SURE}
\eta(\hat{u}_\theta(Y)) = \frac1N \lvert Y - \hat{u}_\theta(Y) \rvert^2 - \sigma^2 + 2 \sigma^2 \operatorname{div}_Y \hat{u}_\theta(Y),
\end{equation}
where $\operatorname{div}_Y \hat{u}_\theta(Y) \coloneqq \sum_{i=1}^N \frac{\partial \mathrm{\hat{u}_{\theta,i}}(Y_i)}{\partial Y_i}$ is the divergence with respect to the data. For continuous and bounded operators with well-defined divergence, $\eta(\hat{u}_\theta)$ is unbiased for $MSE(\hat{u}_\theta(Y)) \coloneqq \frac1N \lvert \hat{u}_\theta(Y) - f \rvert^2$ \citep{ramani2008monte}.

Since the estimator lacks a closed-form divergence, we use a Monte Carlo approximation by perturbing the data with random noise $\mathbf{b} \sim N(0, 1)$:
\begin{equation*}
\operatorname{div}_Y \hat{u}_\theta(Y) = \lim_{\delta\to 0} E_{\mathbf{b}} \left\{ \mathbf{b}' \left( \frac{ \hat{u}_\theta(Y + \delta \mathbf{b}) - \hat{u}_\theta(Y) }{ \delta } \right) \right\}.
\end{equation*}
Following \citet{lucas2022hyperparameter}, we average over $R$ draws of $\mathbf{b}^{(r)}$ to obtain 
$\bar{\eta}^R(\hat{u}_\theta(Y)) \coloneqq \frac1R \sum_{r=1}^R \eta_{\delta, \mathbf{b^{(r)}}}(\hat{u}_\theta(Y))$,
setting $\delta=0.01$. Monte Carlo simulations in the next section confirm that this approach selects effective hyperparameters.


\subsection{Simulations} \label{sec:simulations}

{
\begin{table}[htbp!]
\centering
\footnotesize
\caption{Monte Carlo Simulations}
\label{tab:mc}
\begin{subtable}{0.95\textwidth}
\caption{1D}
\begin{subtable}{\textwidth}
\centering
\begin{tabular}{rrrrrr}
\toprule
    N &    MSE &   MSE $\tau_{\mathrm{FD}}$ &    Bias $\tau_{\mathrm{FD}}$ &     FNR &    FPR \\
\midrule
  500 & 0.0789 &    0.1555 &  0.07   & -0.0437 & 0.1933 \\
 1000 & 0.0461 &    0.1761 &  0.029  &  0      & 0.1099 \\
 5000 & 0.0012 &    0.0102 & -0.0117 &  0      & 0.018  \\
\bottomrule
\end{tabular} \\ SURE: $\lambda$ = 98.6712, $\nu$ = 0.0001  \\ 

\end{subtable}
\end{subtable}
\begin{subtable}{0.95\textwidth}
\caption{2D}
\begin{subtable}{\textwidth}
\centering
 d = 0.25 \\ 
\begin{tabular}{rrrrrrrr}
\toprule
     N &   $ \alpha $ &   $ \hat{\alpha} $ &    MSE &   MSE $\tau_{\mathrm{FD}}$ &   Bias $\tau_{\mathrm{FD}}$ &    FNR &    FPR \\
\midrule
  1000 &  0.0377 &       0.0972 & 0.0018 &    0.0088 & 0.0595 & 0.0105 & 0.5722 \\
  5000 &  0.0377 &       0.0475 & 0.0011 &    0.0007 & 0.0098 & 0.259  & 0.0415 \\
 10000 &  0.0377 &       0.0445 & 0.0011 &    0.0003 & 0.0068 & 0.7593 & 0.0014 \\
\bottomrule
\end{tabular} 
\\ SURE: $\lambda$ = 26.0100, $\nu$ = 0.0012  \\ 

 d = 0.50 \\ 
\begin{tabular}{rrrrrrrr}
\toprule
     N &   $ \alpha $ &   $ \hat{\alpha} $ &    MSE &   MSE $\tau_{\mathrm{FD}}$ &    Bias $\tau_{\mathrm{FD}}$ &    FNR &    FPR \\
\midrule
  1000 &  0.0754 &       0.1059 & 0.002  &    0.0081 &  0.0305 & 0.01   & 0.5642 \\
  5000 &  0.0754 &       0.074  & 0.0012 &    0.002  & -0.0014 & 0.0353 & 0.0499 \\
 10000 &  0.0754 &       0.0751 & 0.0012 &    0.001  & -0.0003 & 0.0884 & 0.0075 \\
\bottomrule
\end{tabular} 
\\ SURE: $\lambda$ = 28.1132, $\nu$ = 0.0011  \\ 

 d = 0.75 \\ 
\begin{tabular}{rrrrrrrr}
\toprule
     N &   $ \alpha $ &   $ \hat{\alpha} $ &    MSE &   MSE $\tau_{\mathrm{FD}}$ &    Bias $\tau_{\mathrm{FD}}$ &    FNR &    FPR \\
\midrule
  1000 &  0.1131 &       0.1131 & 0.0025 &    0.0083 & -0      & 0.0127 & 0.5126 \\
  5000 &  0.1131 &       0.1028 & 0.0014 &    0.0036 & -0.0103 & 0.036  & 0.0339 \\
 10000 &  0.1131 &       0.1149 & 0.0013 &    0.0019 &  0.0017 & 0.0336 & 0.0074 \\
\bottomrule
\end{tabular} 
\\ SURE: $\lambda$ = 50.9836, $\nu$ = 0.0014  \\ 

\end{subtable}
\end{subtable}
\floatfoot{\textit{Note}: table shows averaged results from Monte Carlo simulations of 1D and 2D (300 and 100 per row, respectively) smoothly varying functions with jumps and additive Gaussian noise ($\sigma=0.05$) in Figures \ref{fig:sims_1D} and \ref{fig:sims_2D}. The functions have varying true jump sizes for 1D (see the notes in Figure \ref{fig:simulated_examples}) and jump sizes of Cohen's d 0.25, 0.5, 0.75 with true jump sizes indicated by $\alpha$ for 2D. We uniformly sample random point clouds from these functions, where $N$ denotes the sample size, and estimate the function using our FDR estimator based on those point clouds. $\hat{\alpha}$ denotes the estimated jump size, $MSE$ is the mean squared error with respect to the true noise-free image, $\mathrm{MSE} \, \tau_{\mathrm{FD}}$ is the MSE with respect to the true jump sizes, Bias $\tau_{\mathrm{FD}}$ is the bias of the estimated jump sizes, FNR and FPR are the false negative and false positive rate of the estimated jump locations. Hyperparameters $\lambda$, $\nu$ are estimated using a finite-difference Monte-Carlo approximation of Stein's unbiased risk estimate (SURE) with $R=3$ simulations on a $20\times20$ grid \citep{ramani2008monte}, as in Eq. \eqref{eq:SURE}. $N=\frac{1}{20} n$ for 1D simulations, $N=\frac23 n$ for 2D simulations, where $n$ is the raw sample size and $N$ the number of grid cells along each dimension.}
\end{table}
}

We validate the estimator on 1D, 2D, and 3D simulations with piecewise-smooth functions featuring jumps of varying magnitudes (Cohen’s \(d = 0.25,\) \( 0.5,\) \(0.75\)). Additive Gaussian noise is set to \(\sigma = 0.05\). We discretize the lifted dimension into 32 points \citep{pock2009algorithm} and select hyperparameters \(\lambda,\nu\) via a Monte Carlo version of Stein’s unbiased risk estimate (MC-SURE) \citep{ramani2008monte, lucas2022hyperparameter} on a \(20 \times 20\) grid.

Figure~\ref{fig:simulated_examples} shows the true functions, noisy point clouds, and estimated surfaces for sample sizes of \(n = 5{,}000\) (1D), \(n = 10{,}000\) (2D), and \(n = 50{,}000\) (3D). The method accurately recovers both location and magnitude of discontinuities. In the 1D case, we also plot 95\% confidence bands computed via subsampling \citep{politis1999subsampling}, showing that all estimated jumps are significant.

Table~\ref{tab:mc} reports mean squared error (MSE), jump-size bias, and misclassification rates for boundary detection, averaged over 300 (1D) and 100 (2D) simulations. All measures improve as \(n\) grows, confirming our convergence results -- except for a slight and mechanical increase in the false negative rate for the jump locations due to the finer grid sizes. The false positive and negative rates for the estimated jump locations rapidly decrease to zero, indicating highly accurate jump location detection even in small samples. The estimated jump sizes for the 2D case converge to estimates that different by less than 0.5\% from the true jump sizes, except for the smallest jumps, which still converge but at a slightly slower rate.  Notably, the data-driven selection of hyperparameters consistently finds \(\lambda \gg \nu\), aligning with Theorem~\ref{thm:ident}.

\section{Application: The Economic Effects of Internet Shutdowns in India}

Internet shutdowns—deliberate disruptions of internet or electronic communications—have recently drawn global attention, reaching a record high in 2022 \citep{rosson2023weapons}. India, with its burgeoning digital economy, has implemented at least 646 such shutdowns between 2018 and 2023 \citep{InternetShutdowns2024}, more than any other country, primarily to quell protests, communal violence, and cheating in examinations \citep{hrw}. In this section, we exploit the geographic discontinuities caused by these shutdowns to estimate their short-term economic impact.

The FDR estimator is well-suited for this purpose because the areas exposed to the shutdown—where internet connectivity abruptly dropped—were unknown \textit{a priori}. Its multidimensional formulation allows us to explicitly estimate the geography of the targeted areas. This is of independent interest, to understand what areas were affected, but also enables us to precisely estimate the impact of the shutdown by looking at the size of the jumps at the estimated area boundaries. 

We focus on a shutdown imposed by the Rajasthan state government on September 26, 2021, to prevent cheating on the Rajasthan Eligibility Exam for Teachers. This high-stakes exam, not held since 2018, drew hundreds of thousands of candidates in 2021 amid heightened concerns of cheating after several scandals \citep{Purohit}. In the days before the exam, multiple district governments announced a mobile internet shutdown from 6 am to 6 pm.\footnote{Specifically, the District Magistrates or Divisional Commissioners of Ajmer, Jhunjhunu, Kota, Bundi, Baran, Jhalawar, and parts of Udaipur issued suspension orders \citep{Mishra_2021}.}

Several other districts also cut telecom services without notice, ultimately affecting an estimated 25 million people \citep{yeung2021internet}. Because the Rajasthan telecom circle does not perfectly align with administrative boundaries, and connectivity depends on how cell towers, service providers, and mobile devices interact, it was unclear which districts—and what parts of those districts—would be affected. Where the shutdown was enforced, it was complete, as local governments compelled providers to disable both transmission (cell towers) and point-of-contact (mobile phones). This meant users could not connect to towers outside the shutdown area, inducing discontinuous drops in connectivity.

Although the shutdown technically targeted mobile services, it effectively amounted to a near-complete internet blackout. In 2022, Wi-Fi accounted for just 0.08\% of wireless connectivity in India, and only 3.74\% of internet subscribers had wired connections \citep{trai2023indian}. This disruption reverberated through the economy. In Jaipur alone, 80,000 shops reportedly closed \citep{economist_shutdowns}, due to the ubiquity of mobile-based payment systems like UPI or Google Pay. Although cash is accepted, it is less convenient in a system where even small transactions often rely on mobile platforms \citep{nyt_mobile_payments}, leading many ATMs to run dry during similar shutdowns \citep{jaipur_shutdown}. A large share of retail, hospitality, and mobile-app-based transactions also ground to a halt.

More generally, these shutdowns directly disrupt India’s digital economy, which accounted for about 22\% of GDP in 2019 and continues to expand. While wired connections remained available, losing mobile connectivity still undermined supply-chain tracking, process automation, distribution networks, remote work, and customer support—especially for small and rural businesses that rely entirely on mobile internet \citep[p.55]{kathuria2018anatomy}.

A commonly cited framework for estimating the costs of internet shutdowns is the back-of-the-envelope method from \citet{west2016internet},\footnote{See also \citet{top10vpn2023} and \url{https://netblocks.org/projects/cost}.} which uses data on the digital economy’s size, mobile penetration, and the ``digital multiplier'' \citep{quelch2009quantifying}. While transparent, this approach does not measure actual local economic impacts, and there appear to be no academic estimates of these, despite shutdowns’ global prominence. In development economics, however, research on digital infrastructure expansions suggests notable economic benefits. \citet{roller2001telecommunications}, for instance, found that each percentage-point rise in mobile penetration raises economic growth by around 0.15\%. Other studies show welfare improvements, lower living costs, and increased employment from mobile rollouts \citep{jensen2007digital, bjorkegren2022network, couture2021connecting, zuo2021wired, hjort2019arrival}, consistent with digitization reducing search, replication, transportation, tracking, and verification costs \citep{goldfarb2019digital}.

In contrast, an internet shutdown disrupts those welfare-improving channels, creating economy-wide ripple effects. Moreover, losing established connectivity is not simply the inverse of gaining it; digital infrastructures, once in place, play “an ‘enabling function’ across all critical infrastructure sectors” \citep{CISA}. Our findings therefore complement, rather than mirror, existing work on connectivity expansions.

Lastly, using mobile device data to proxy socioeconomic activity is well-established in fields such as remote sensing, network science, and complex systems \citep{vscepanovic2015mobile, kung2014exploring, frias2012relationship, eagle2010network, blumenstock2015predicting}, as well as in industry \citep{naef2014using} and government \citep{worldbank2022ukraine}. More recently, economists and policy researchers have adopted these data to study social and economic questions \citep{kreindler2021measuring, van2023public}.

\subsection{Data}

We analyze anonymized device-level location data from data provider Veraset to assess the effects of the shutdown at a fine spatio-temporal granularity. The data consist of ``pings'', timestamped GPS locations shared by mobile devices with apps that use common software development kits. Inside our sample area around Rajasthan, we observe 126 million unique pings from 3.8 million devices during the four Sundays of September 2021 between 6 am and 6 pm, the shutdown time window.
To capture the disruption to mobile connectivity, we define
\begin{equation}
\label{eq:pings}
\overline{Pings_i} \coloneqq \frac{Pings_{i t_0}}{\frac13 \sum_{t=t_0-3}^{t_0-1} Pings_{i t}},
\end{equation}
where $i$ indexes a 5$\times$5km grid cell overlaid on Rajasthan; $t_0$ indexes September 26, 2021, between 6 am and 6 pm (the shutdown period); and $t_0-3,\ldots,t_0-1$ index the same time window on the three previous Sundays. $Pings_{it}$ counts mobile device pings in grid cell $i$ in period $t$. The normalization by the prior month's Sunday average removes regional differences in absolute activity levels.

To estimate economic effects, we compute an analogous measure $\overline{Econ}_i$ that captures activity around commercial points of interest (POIs). We combine two datasets --- SafeGraph Places (global retail chains) and OpenStreetMap (local businesses) --- to identify approximately 108,000 retail, commercial, and industrial locations (see Appendix \ref{app:data} for detailed classifications). We classify a ping as economically relevant if it originates within 150 meters of a POI. The economic measure is cast onto a coarser 40$\times$40km grid intersected with municipalities to ensure sufficient observations per cell.

Some devices continue emitting signals during the shutdown because location data is determined via satellite connections, which were unaffected. Apps with background location permissions cache these positions for later transmission, enabling measurement of economic activity even in shutdown areas. The representativeness of this cached location data and additional robustness checks are discussed in Appendix \ref{app:data}.

\begin{figure}[htbp!]
\begin{subfigure}[b]{0.4\linewidth}
\includegraphics[width=\textwidth]{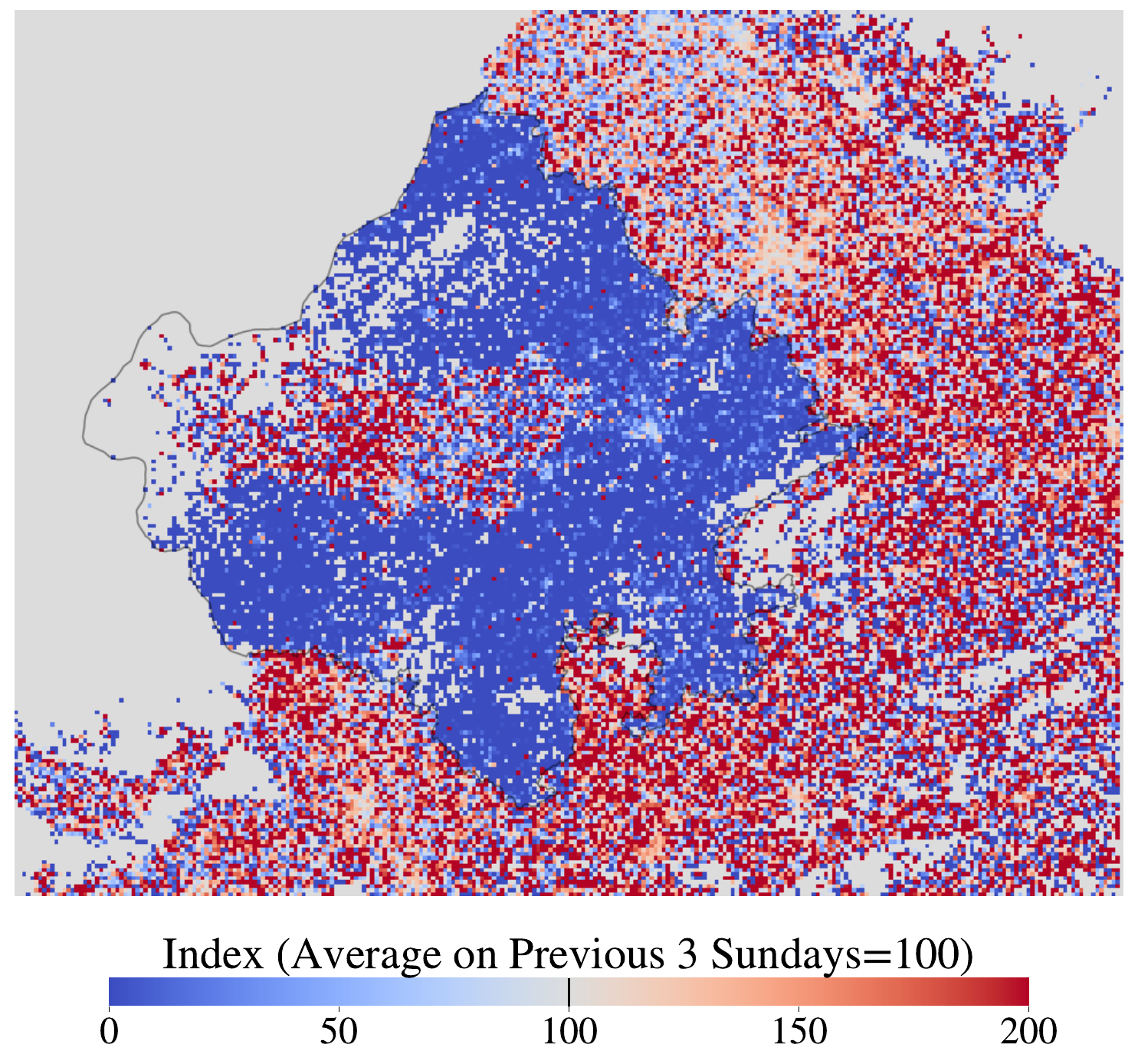}
    \caption{Mobile: Raw Data}
    \label{fig:india_mobile_raw}
\end{subfigure}
\begin{subfigure}[b]{0.4\linewidth}
\includegraphics[width=\textwidth]{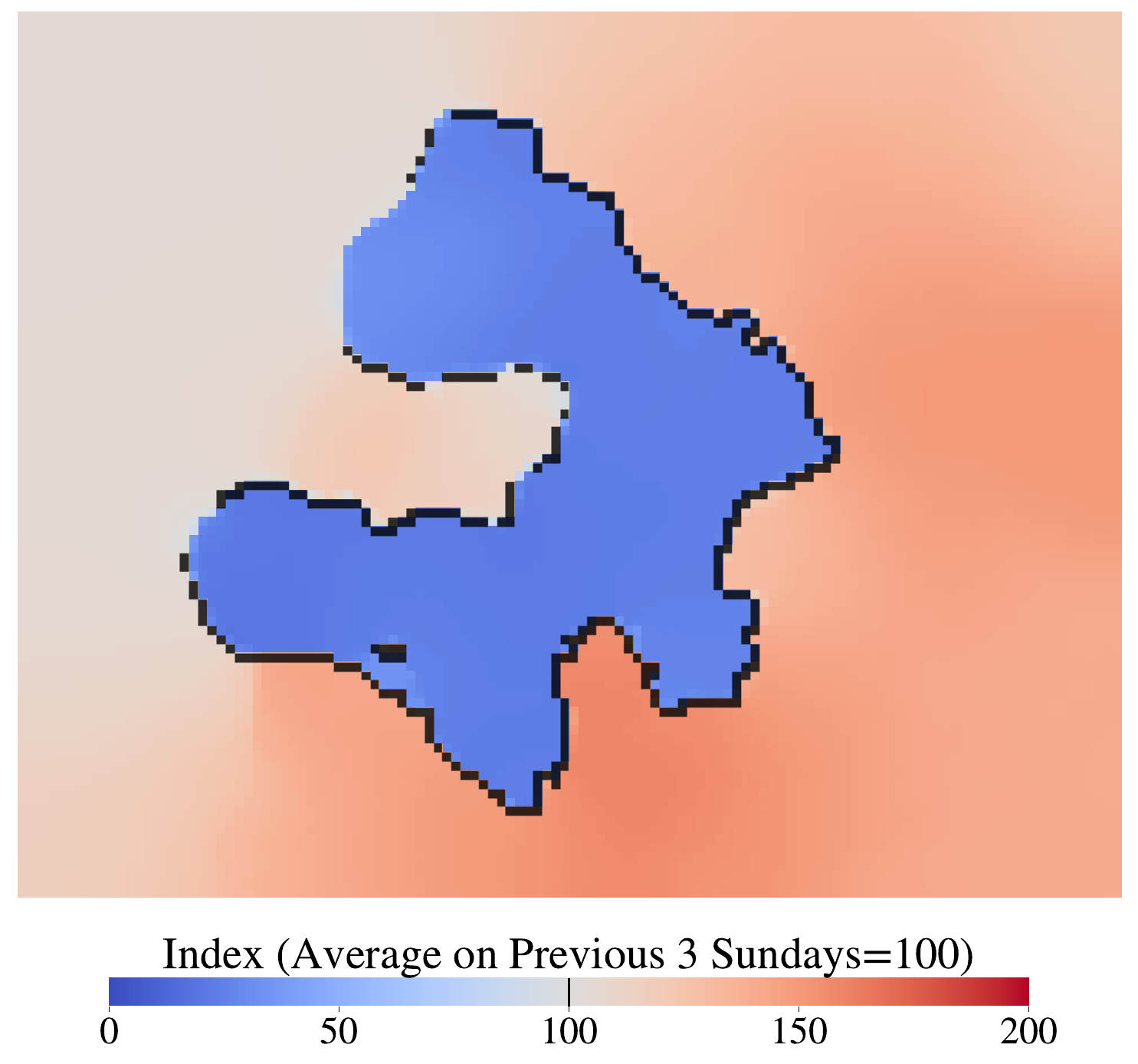}
    \caption{Mobile: $\hat{u}$ \& $S_u$}
        \label{fig:india_mobile_u}
\end{subfigure}
\begin{subfigure}[b]{0.5\linewidth}
\includegraphics[width=\textwidth]{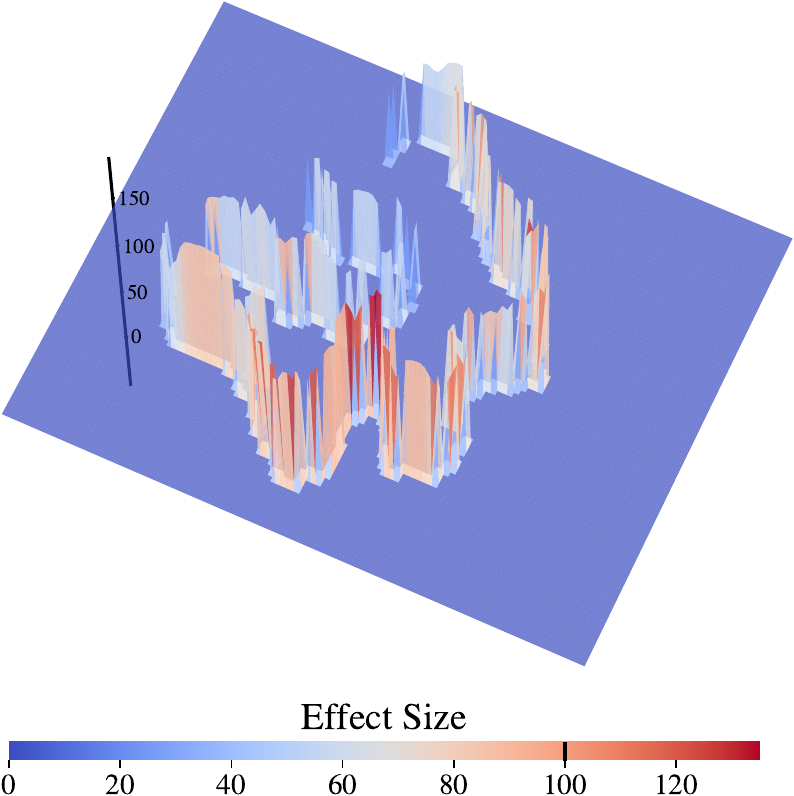}
    \caption{Mobile: $\tau_{\mathrm{FD}}(x)$}
            \label{fig:india_mobile_treatment}
\end{subfigure}
\caption{Internet Shutdown: Mobile Data Effects}
\label{fig:india_mobile}
\floatfoot{ \textit{Notes:} Plots depict in- and outputs of FDR estimation of effects of internet shutdown in Rajasthan state, India, on September 26, 2021, on mobile device signal, as measured by $\overline{Pings}_i$ (see \eqref{eq:pings}) which is a measure of the total mobile device pings per $5\times5$km between 6 am and 6 pm on the day of the shutdown relative to the average in the same time window on the preceding 3 Sundays. $\lambda = 91.2474, \nu=0.0656$ selected by SURE. \textbf{(a)} Shows the raw input data with the fill color of each cell indicating the value of $\overline{Pings}$. The outline of Rajasthan state is indicated by black lines. \textbf{(b)} shows the estimated regression function in color, with the estimated jump set $S_u$ indicated in black. \textbf{(c)} shows the effect curve $\tau_{\mathrm{FD}}(x)$ for those areas that have a jump induced by the shutdown (dropping the discontinuity in the north-east), with the z-axis indicating the magnitude of the drop in terms of $\overline{Pings}_i$.}
\end{figure}


\subsection{Results}

We apply the FDR estimator to the data described above, starting with the mobile device signal in Figure \ref{fig:india_mobile}. 

\paragraph*{Mobile signal} The raw data (Figure \ref{fig:india_mobile_raw}) shows a stark drop-off in signal that aligns with the Rajasthan state boundary, with an unexpected pocket of activity in the northwest intersecting the districts of Jaisalmer, Jodhpur, and Nagaur. Areas without prior signal are normalized to 100, while areas unaffected by the shutdown show signals averaging 130\% of the Sunday September average. This reflects both mechanical noise from low spatial granularity and the broader economic recovery in India at that particular time in the COVID-19 pandemic \citep{Bhaduri_2021}.

\begin{figure}[ht!]
\begin{subfigure}[b]{0.4\linewidth}
\includegraphics[width=\textwidth]{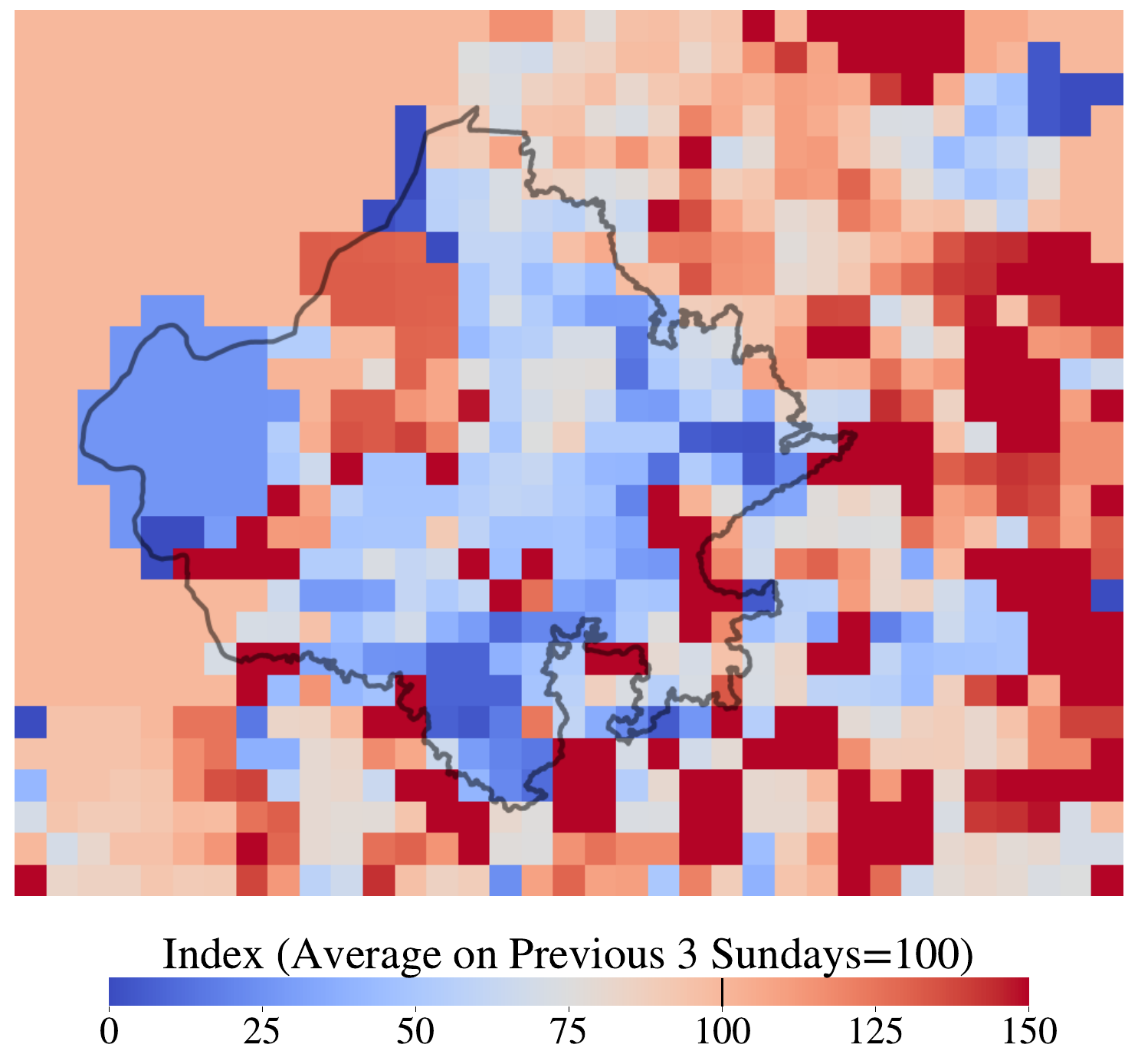}
    \caption{Economic: Raw Data}
\end{subfigure}
\begin{subfigure}[b]{0.4\linewidth}
\includegraphics[width=\textwidth]{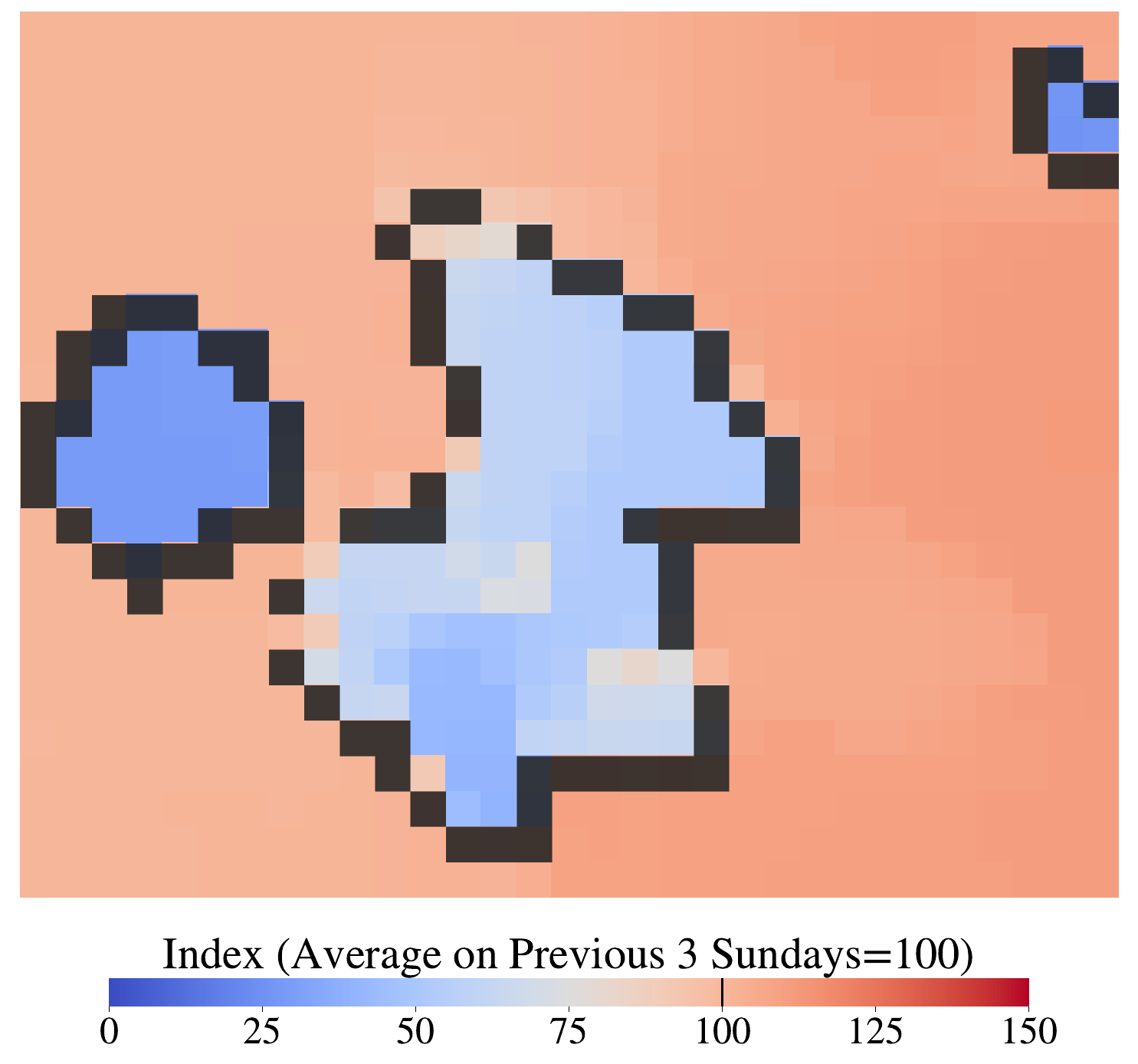}
    \caption{Economic: $\hat{u}$ \& $S_u$}
\end{subfigure}
\begin{subfigure}[b]{0.5\linewidth}
\includegraphics[width=\textwidth]{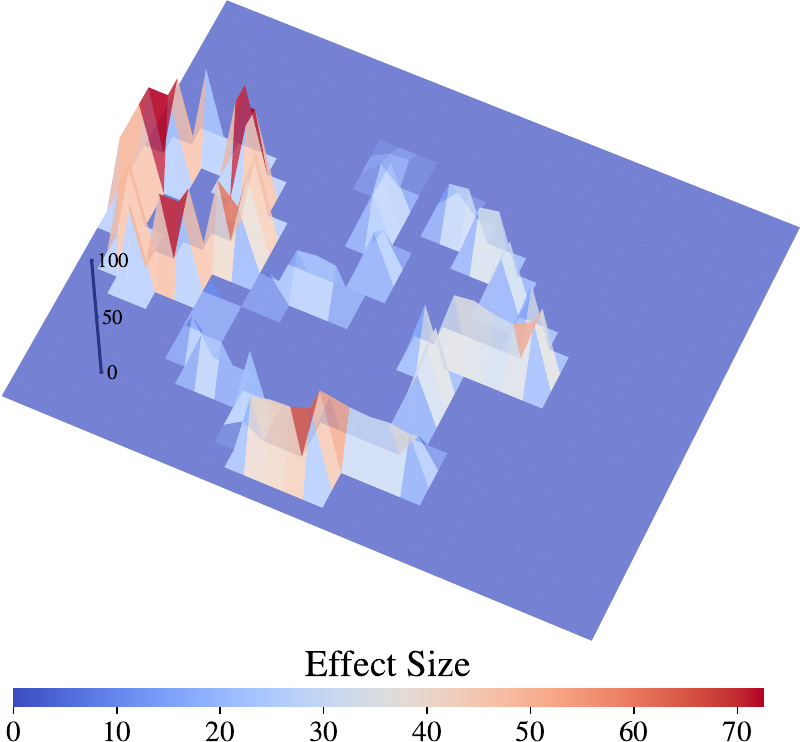}
    \caption{Economic: $\tau_{\mathrm{FD}}(x)$}
\end{subfigure}
\caption{Internet Shutdown: Economic Effects}
\label{fig:india_econ}
\floatfoot{\textit{Notes:} Plots depict in- and outputs of FDR estimation of effects of internet shutdown in Rajasthan state, India, on September 26, 2021, on economic activity, as measured by $\overline{Econ}_i$ (see the ``Economic Activity`` paragraph) which is a measure of the share of mobile device pings per $40\times40$km that fell within a 150m radius of a commerce-related Point of Interest between 6 am and 6 pm on the day of the shutdown relative to the average in the same time window on the preceding 3 Sundays. $\lambda=4.9813, \nu=0.0754$ selected by SURE. \\
\textbf{(a)} Shows the raw input data with the fill color of each cell indicating the value of $\overline{Econ}_i$. The outline of Rajasthan state is solid black. \textbf{(b)} shows the estimated regression function in color, with the estimated jump set $S_u$ indicated in black if the jump size is significantly different from 0 at the 95\% level based on 4$\times$200 subsamples with estimated rate of convergence of approx. $\sqrt{n}$ \citep{politis1999subsampling} and grey otherwise. \textbf{(c)} shows the effect curve $\tau_{\mathrm{FD}}(x)$ for those areas that have a jump induced by the shutdown (dropping the discontinuity in the north-east), with the z-axis indicating the magnitude of the drop in terms of $\overline{Econ}_i$.}
\end{figure}

For hyperparameter selection, we search over $\lambda \in [1,100]$ and $\nu \in [0.001, 0.1]$ to account for the large amount of noise in the mobile data. We winsorize at the 90th percentile to handle outliers from temporary gatherings and set $f_X$ to the uniform distribution after verifying robustness to alternatives. The resulting regression function $\hat{u}$ (Figure \ref{fig:india_mobile_u}) is smooth within the two partitions, with the estimated jump set $S_u$ precisely capturing the state boundary and the unexpected area of continued connectivity within the state. The effect curve (\ref{fig:india_mobile_treatment}) shows an average reduction in mobile signal of 100\% relative to baseline, with higher reductions near Gujarat and Madhya Pradesh. About 25\% of the baseline signal remains inside the shutdown area due to satellite-based location caching, enabling our subsequent economic analysis. Overall, we estimate that the shutdown was highly effective in its goal of disrupting connectivity, inducing a large drop in the mobile device signal throughout the targeted area.

\paragraph*{Economic activity}
Figure \ref{fig:india_econ} shows the results for economic activity. We construct 95\% confidence bands around the function gradient using 200 subsamples at 4 different subsample sizes, estimating a convergence rate of approximately $\sqrt{n}$. Hyperparameters are selected from $\lambda \in [1,5]$ and $\nu \in [0.075, 0.15]$, adjusted for the coarser grid. Despite lower resolution, the estimated jump set closely tracks the mobile signal discontinuities, overlapping with state boundaries and identifying the northwestern pocket of connectivity.

The shutdown reduced economic activity by approximately 35\% in the eastern region and 60\% in the western region, though the western estimate is less reliable due to sparse data and proximity to Pakistan. Outside the shutdown area, activity increased 0-5\%, consistent with India's broader economic recovery at the time \citep{woloszko2020tracking}. Scaling our estimate by 0.7—the most conservative estimate of mobility-GDP correlation from the literature \citep{dong2017measuring, frias2012relationship, spelta2021mobility}—yields a 25\% reduction in economic activity. This substantially exceeds previous estimates of shutdown costs, which predicted impacts around 16.4\% of annualized GDP.

This discrepancy likely reflects two factors: previous estimates focused narrowly on digitized economic activity, missing spillovers to non-digital sectors, and relied on outdated digital economy multipliers from 2009 \citep{quelch2009quantifying} that understate digital infrastructure's current importance, particularly in rapidly digitalizing economies like India.

\paragraph*{Implications}
Our results reveal that arbitrary internet shutdowns inflict larger short-term economic damage than previously understood. While some recovery might occur—though we find no evidence of this in mobile or economic activity the following day (Figure \ref{fig:india_overshooting})—the magnitude of our estimates underscores digital infrastructure's critical role in modern economies. The findings highlight an important asymmetry: while the benefits of expanding internet access accrue gradually, the costs of removing it are immediate and substantial once digital infrastructure becomes integrated into economic activity. This emphasizes not just the importance of expanding internet access, but of ensuring its continued reliability \citep{UN2022}.

{\section{Conclusion}

We introduced \emph{Free Discontinuity Regression} (FDR): the first fully non-parametric framework that \emph{simultaneously} estimates a multivariate regression surface and the geometry—and magnitude—of its unknown discontinuities.  
Built on a statistical reinterpretation of the convexified Mumford–Shah functional, FDR overcomes three long-standing bottlenecks in high-dimensional jump surface estimation:

\begin{enumerate}[leftmargin=2.2em]
\item \textit{Unified estimation.}  Unlike wavelet, fused-lasso, or post-hoc thresholding approaches, FDR brings smoothing and segmentation under a single objective, producing jump locations and sizes \emph{as native outputs} rather than after-thoughts.

\item \textit{Global optimality with guarantees.}  Calibrations convert the non-convex Mumford–Shah functional into a saddle-point problem that primal–dual algorithms solve to \emph{global} optimality.  We prove that, under mild SBV-regularity, the recovered surface and jump set are identified and consistent in any dimension.

\item \textit{Realistic data-generating processes.}  FDR allows a random design, spatially-correlated noise, and randomness in both grid \emph{and} sampling error—assumptions that match modern geospatial, imaging and remote-sensing datasets but are absent from existing theory.
\end{enumerate}

A data-driven version of Stein's Unbiased Risk Estimate selects the two tuning parameters automatically and in a way that is consistent with our identification theorems.  Large-scale simulations from 1D to 3D confirm the predicted convergence of both surfaces and jump sizes, while an application to internet shutdowns in India showcases FDR’s empirical power—revealing a 25–35\% drop in economic activity, far larger than headline estimates. This points to a large asymmetry in the effects of internet expansions and shutdowns, underscoring the internet's critical role in modern economies.

\paragraph*{Outlook}  By combining theoretical guarantees with efficient optimization procedures, the proposed estimator provides a robust and practical statistical estimation tool for multivariate changepoint estimation. Immediate extensions include (i) dynamic or panel settings where jump surfaces evolve over time, through vectorial coupling of the multiple output channels \citep{strekalovskiy2012convex}, (ii) adaptive or multiresolution grids for ultra-high-resolution imagery, (iii) instrumental-variable and causal designs that leverage the estimated discontinuity set as an endogenous boundary, and (iv) computational improvements through statistical versions of the sublabel-accurate approaches in the image literature \citep{mollenhoff2017sublabel}. The hope is that the framework will spur new research into the estimation and application of jumps in high-dimensional processes and further bridge the gap between computer vision and statistical estimation.

\clearpage 
\begin{supplement}
\stitle{}
The supplementary material introduces the mathematical notation (\ref{app:notation}), additional details on the empirical application (\ref{app:data}), detailed proofs of the results in the texts (\ref{sec:proofs}), implementation details (\ref{app:implement}), details on uncertainty quantification (\ref{app:uncertainty}), and additional results (\ref{app:additional_results}).

\begin{appendices}

\begin{appendix}

\setcounter{figure}{0}
\setcounter{table}{0}

\renewcommand{\thefigure}{A-\arabic{figure}}
\renewcommand{\thetable}{A-\arabic{table}}

\section{Mathematical Notation and Definitions} \label{app:notation}
This section contains notation and definitions of mathematical objects used in the main text. These are standard and we refer to \cite{ambrosio2000functions} for further reading.

\begin{definition}[Approximate jump points] \label{def:jumppoints} Let $u$ be locally integrable on $\mathcal{X}$, that is $u \in L_{\text {loc }}^1(P_X)$, where $P_X$ denotes the law of $X$. We say that a point $x\in \mathcal{X}$ is an approximate jump point of $u$ if there exist constants $a, b \in \mathbb{R}$ and an orientation $\rho \in \mathbb{S}^{d-1}$ such that $a \neq b$ and
\begin{equation} \label{eq:jumppoints}
\lim _{\varepsilon \downarrow 0} \frac{\int_{B_\varepsilon^{+}(x, \rho)}|u(x')-a| \dd P_X(x')}{P_X(B_\varepsilon^+(x,\rho))}=0 , \quad \lim _{\varepsilon \downarrow 0} \frac{\int_{B_\varepsilon^{-}(x', \rho)}|u(x')-b| \dd P_X( x')}{P_X(B_\varepsilon^{-}(x', \rho))}=0,
\end{equation}
where 
\begin{equation*}
\left\{\begin{array}{l}
B_\varepsilon^{+}(x, \rho):=\left\{x' \in B_\varepsilon(x):\langle x'-x, \rho \rangle>0\right\} \\
B_\varepsilon^{-}(x, \rho):=\left\{x' \in B_\varepsilon(x):\langle x'-x, \rho \rangle<0\right\}
\end{array}\right.
\end{equation*}
denote the two half balls, oriented by  $\rho$, contained in the $\varepsilon$-ball $B_{\varepsilon}(x)$, and $\mathbb{S}^{d-1}$ is the unit sphere in $\mathbb{R}^d$.
I The triplet $(a, b, \rho)$, uniquely determined by \eqref{eq:jumppoints} up to a permutation of $(a, b)$ and a change of sign of $\rho$, is denoted by $\left(u^{+}(x), u^{-}(x), \rho_u(x)\right)$; with $u^+, u^-$ called the traces of $u$.
The approximate discontinuity set contains all approximate jump points and is denoted by $S_u$.
\end{definition}

A Radon measure $\mu$ is an inner regular and locally finite Borel measure. We denote
\[\dashint_{A} f(x)\dd \mu(x)\coloneqq \frac{\int_{A} f(x)\dd \mu(x)}{\mu(A)}.\] We denote the space of ($k$-times) continuously differentiable functions on $\mathcal{X}$ by $C^k(\mathcal{X})$. $C_0(\mathcal{X})$, the space of all continuous functions that vanish eventually, is the closure in the sup-norm of $C_c(\mathcal{X})$, the space of all continuous functions on $\mathcal{X}$ with compact support.

The $k$-dimensional Hausdorff-measure $\mathscr{H}^k(A)$ of a set $A\subset\mathbb{R}^d$, for $k\in [0,+\infty)$, is defined as \citep[e.g.][Definition 2.46]{ambrosio2000functions}
    \[\mathscr{H}^k(A)\coloneqq \lim_{\delta\downarrow0} \frac{\pi^{k/2}}{2^k\Gamma(1+k/2)}\inf \left\{\sum_{i\in I} \left(\text{diam}(A_i)\right)^k: \text{diam}(A_i)<\delta, A\subset\bigcup_{i\in I}A_i\right\}\] with $\Gamma(x)$ being Euler's Gamma function, $\text{diam}(A)$ the diameter of the set $A$, and where the sums are taken over finite or countable covers of $A$.

We use the following definition of functions of bounded variation.
\begin{definition}[Functions of bounded variation and SBV]\citep[Def. 3.1]{ambrosio2000functions}
Let $u \in L^{1}(\mathcal{X})$; we say that $u$ is a function of bounded variation in $\mathcal{X} \subset \mathbb{R}^d$ if the distributional derivative of $u$, $Du$, is representable by a finite Radon measure in $\mathcal{X}$. i.e. if
\begin{equation} \label{eq:BV}
\int_{\mathcal{X}} u \operatorname{div} \varphi d x=-\sum_{i=1}^d \int_{\mathcal{X}} \varphi_i d D_i u \quad \forall \varphi \in C_i^1(\mathcal{X}) .
\end{equation}
for some $R^d$-valued Radon measure $D u=\left(D_1 u . \ldots D_d u\right)$ in $\mathcal{X}$ and all continuously differentiable test functions $\varphi$. The vector space of all functions of bounded variation in $\mathcal{X}$ is denoted by $B V(\mathcal{X})$. If the Cantor part $D_c u$ of the decomposition of $Du$ is zero, then $u$ is called a special function of bounded variation and denoted as $u\in SBV(\mathcal{X})$.
\end{definition}

We define the weak$^*$-topology in the standard way \citep[Definition 3.11]{ambrosio2000functions}: Consider some $u\in BV(\mathcal{X})$ and a sequence $\{u_n\}\subset BV(\mathcal{X})$. Then we say that $\{u_n\}$ converges to $u$ in the weak$^*$-topology, or $u_n\overset{*}{\rightharpoonup} u$, if $u_n$ converges to $u$ in $L^1(\mathcal{X})$ and the corresponding $Du_n$ converge in the weak$^*$-topology to $Du$, i.e.
\[\lim_{n\to\infty} \int_{\mathcal{X}}\phi\dd Du_n = \int_{\mathcal{X}}\phi\dd Du\qquad\text{ for all $\phi\in C_0(\mathcal{X})$}. \]

For two non-empty subsets $A,B\subset M$ of some metric space $(M,d)$, we define their Hausdorff distance as
\[d_H(A,B)\coloneqq \max\left\{\sup_{a\in A}\mathrm{dist}(a,B), \sup_{b\in B}\mathrm{dist}(b,A)\right\},\] where $\mathrm{dist}(a, B) \coloneqq \inf_{b\in B} d(a,b)$ is the distance of $a$ from $B$.

\section{Data and Measurement}\label{app:data}

Anonymized device-level location data from data provider Veraset allow us to assess the effects of the shutdown at a fine spatio-temporal granularity. The data consist of "pings", timestamped GPS locations shared by the device with a mobile app. Veraset cleans and aggregates such data from thousands of "Software-Development Kits" (SDK), packages of tools that supply the infrastructure for most mobile applications. Location data from the same device can be recombined from various SDKs by use of an anonymized device ID. Inside our sample area, which is a bounding box around the state of Rajasthan, we observe 126 million unique pings on the four Sundays of September 2021 between 6 am and 6 pm, emanating from 3.8 million unique devices. With an estimated smartphone penetration of 60.63\% in India in 2021 and a population in Rajasthan of around 80 million, this means we capture a little under 10\% of all mobile devices in the area.

Few other data sources can proxy for economic activity at the fine spatiotemporal granularity we consider: night lights data are not suited for such rapid temporal variations in economic activity that happen during daytime; credit card penetration is low in India (around 5\%); and data from neither Google Pay nor UPI were being sold to third parties at the time of this paper's writing.

Our device-based measure should be a good proxy of economic activity under two assumptions: 1) the economic activity of the restricted device sample is representative of the wider population; 2) the set of apps that cached devices' background location for the duration of the shutdown was not in some way skewed so as to collect geolocation information at different rates in economic areas during the shutdown. To point 1), though information on location-sharing behavior by Indian smartphone users is sparse, one recent poll of Android users suggests 66.78\% have background location enabled, and Android makes up 95.21\% of the smartphone market share in India \citep{Android_Authority_2022, android_india}. This is further supported by the fact that even after filtering on economic areas for our restricted device sample, we still retain around 68 million unique pings compared to the 126 million total pings. To point 2), our location data comes from 1,000+ different mobile applications which allows location tracking across a wide variety of platforms and thus areas and behaviors. Additionally, background location collection happens continuously so it should not be meaningfully affected by people's differential app usage during the shutdown.

One might expect there to be selection effects on economic activity around the discontinuity, if people near the border cross state lines to avoid the shutdown. We can test for such selection effects by counting the number of devices that are observed crossing the state boundaries during the shutdown and comparing it to the average in the preceding month. We focus on a 40km band around the Rajasthan state boundary as illustrated in Figure \ref{fig:crossers}. We calculate the share of devices near the boundary that appeared in one of the neighboring states during the shutdown period and compare this share to its average in the month prior. 

We find that 2.66\% of devices near the discontinuity cross it between 6 am and 6 pm on the day of the shutdown, while on average 3.57\% do so in the same time window the month prior. Thus, there do not appear to be any self-selection effects associated with the localized internet shutdown. In fact, the lower share of devices that cross suggests that some crossing behavior was disrupted due to the shutdown, likely because of the associated disruption of mobile-based navigation services. This does not contradict the fact that location apps may still be caching devices' location, as route calculation and live navigation require full mobile connection, not just satellite-based geo-positioning. Relatedly, we do not expect mobile phones near the Rajasthan border to connect to telecommunications masts across the border, as the shutdown was implemented at the point of contact (the mobile phone) and not just the transmission points (the masts), as mentioned in the main text.

\section{Proofs}
\label{sec:proofs}

\subsection{Proof of Proposition \ref{prop:existence}}
\begin{proof}
    We work in the weak$^*$-topology \citep[Definition 3.11]{ambrosio2000functions}. Since all $v$ have a uniformly bounded variation by $c<+\infty$ and since all $v$ map to $[0,1]$ and are hence uniformly bounded, it follows from Theorem 4.8 in \citet{ambrosio2000functions} that $C\cap\{v: |Dv|\leq c\}$ is compact in the weak$^*$-topology. 
    
    Now we want to prove that the objective function $E(v)$ is lower semi-continuous in $v$. For this, we analyze the inner optimization problem. Since $C_0$ is a Banach space with respect to the supremum norm, Corollary 6.40 in \citet{aliprantis2006infinite} implies that the objective function of \eqref{eq:main_problem} is jointly continuous as the total variations are bounded by $c<+\infty$. Moreover, the constraint set $K$ does not depend on $v$, so it is trivially upper hemicontinuous in $v$. Therefore, Lemma 17.29 in \citet{aliprantis2006infinite} proves that $E(v)$ is lower semicontinuous in $v$. Thus, the optimization problem $\inf_{v\in C\cap\{v:|Dv|\leq c\}} E(v)$ admits a minimizer in $C$, which shows existence. The fact that the solution is a global minimum follows from the convexity of $C\cap\{v: |Dv|\leq c\}$.
\end{proof}

\subsection{Proof of Theorem \ref{thm:ident}}
\begin{proof}
    We first show that if we set $v^*(x,t) = \mathds{1}_f(x,t)$, then $E(v^*)<+\infty$ for any $\lambda\geq0$ and also as $\lambda\to+\infty$. In fact, since $f\in SBV(\mathcal{X})$ by assumption, fixing the rotation via the inward normal to the graph, the objective function in \eqref{eq:main_problem} can be written as \citep[Lemma 2.10]{alberti2001calibration}
    \begin{align*} E(\mathds{1}_f) = \sup_{p\in K}\int_{\mathcal{X}} \left[ p^x(x,f(x))\cdot \nabla f(x)- p^t(x,f(x))\right]\dd x + \\  \int_{S_f}\left[\int_{f^-}^{f^+}p^x(x,t)\dd t\right]\cdot \rho_f(x)\dd\mathscr{H}^{d-1}(x),\end{align*}
    where $p^x$ denotes the first $d$ dimensions of the vector field $p$ and $p^t$ denotes the lifted dimension. $f^-(x)$ and $f^{+}(x)$ denote the respective traces corresponding to the inward normal orientation $\rho_f(x)$ of the graph of $f$ at the point $x\in S_f$. Note that when using $f(x)$, the constraint set $K$ implies under Assumption \ref{ass:ident1} that 
    \[p^t(x,f(x))\geq \frac{|p^x(x,f(x))|^2}{4f_X(x)}\geq0\] independently of $\lambda$ and $\nu$. Moreover, by Assumption \ref{ass:ident1} it holds that $\int_{\mathcal{X}} |\nabla f|^2\dd x <+\infty$. This, by an application of H\"older's inequality for vector-valued functions in conjunction with the fact that $p\in K$ -- and thus $p^t$ cannot be negative -- and $\nu<+\infty$ implies that the first term on the right-hand side is finite for all $p\in K$. For the second term on the right-hand side note that by assumption $\mathscr{H}^{d-1}(S_f)<+\infty$. Furthermore, since $p\in K$, the integrand over $p^x(x,t)$ is bounded by $\nu<+\infty$, so that the overall $E(\mathds{1}_f)<+\infty$ for all $\lambda\geq0$ and also as $\lambda\to+\infty$. 

We now show that in the limit as $\lambda\to+\infty$ and for fixed $\nu>0$, it holds that $\nabla v^*(\lambda)=0$ $\mathcal{L}^{d+1}$-almost everywhere.
To prove this, consider any $v\in C$. Since $v\in SBV(\mathbb{R}^{d+1})$ it holds by definition that
    \[Dv = \nabla v \mathcal{L}^{d+1}+(v^+-v^-)\rho_v\mathscr{H}^{d}\mres J_v,\] where $J_v$ is the set where $v(x,t)$ jumps, $\rho_v$ is a corresponding orientation, and $\mathcal{L}^{d+1}$ denotes the $(d+1)$-dimensional Lebesgue measure. Note that by the Federer-Vol'pert theorem \citep[Theorem 3.78]{ambrosio2000functions} $J_v$ coincides with $S_v$ $\mathscr{H}^d$-almost everywhere, so that we focus on $J_v$. It therefore holds that
    \[\int p\cdot Dv = \int p\cdot\nabla v\dd \mathcal{L}^{d+1} + \int_{J_v} (v^+-v^-)p\cdot \rho_v\dd\mathscr{H}^d.\] Since $v\in C$, it cannot be constant everywhere because the limits as $t\to \pm \infty$ are different for all $x\in\mathcal{X}$. On the other hand, as $\lambda\to+\infty$, the range of $p^t$ converges to $\mathbb{R}$, which implies that for every $\eta>0$ and every $v\in C$ such that $\nabla v\neq 0$ on a subset of $\mathbb{R}^{d+1}$ of positive $\mathcal{L}^{d+1}$-measure, there exists a large enough $\lambda$ and a corresponding $p\in K$ such that $\int p\cdot \nabla v\dd x>\eta$. This implies that no $v$ for which $\nabla v$ is not constant on a set of positive $\mathcal{L}^{d+1}$ measure can be a solution to \eqref{eq:main_problem} as $\lambda\to+\infty$. Therefore, in the limit as $\lambda\to+\infty$, the only changes in $v$ must lie in $J_v$. 
    
    We now prove the rest of the theorem. First, we prove that if $v\in C$ is such that for some $\eta>0$
    \begin{equation}\label{eq:contradic}
    \Gamma_f\not\subset J^\eta_v(\lambda),
    \end{equation} where $A^\eta\coloneqq\{x'\in\mathcal{X}: \text{dist}(x',A)\leq \eta\}$ is the $\eta$-enlargement of $A$, then it cannot be a solution to \eqref{eq:main_problem} in the limit as $\lambda\to+\infty$. Then we prove that if for some $\eta>0$
    \begin{equation}\label{eq:contradic2}
    J_v(\lambda)\not\subset \Gamma^\eta_f
    \end{equation} then it cannot be a solution to \eqref{eq:main_problem} in the limit as $\lambda\to+\infty$. \\

    \noindent \emph{Showing \eqref{eq:contradic} leads to a contradiction.}\\
    The idea is to show that $\lim_{\lambda\to+\infty} E(v(\lambda)) = +\infty$ for any sequence $\{v(\lambda)\}\subset C$ that satisfies \eqref{eq:contradic}. So suppose that $\{v(\lambda)\}\subset C$ is some sequence that satisfies \eqref{eq:contradic}. This implies that in the limit as $\lambda\to+\infty$, it holds that there is some $(x',t')\in \Gamma_f$ such that $\text{dist}((x',t'),J_v(\lambda))>\eta>0$.

    This can happen in two ways. Either, $x'\not\in S_f$ or $x'\in S_f$. Recall that under the first part of Assumption \ref{ass:ident2} it holds that $\mathscr{H}^{d-1}(\bar{S}_f\setminus S_f)=0$, where $\bar{S}_f$ is the closure. We now consider each case one by one.
    
    \emph{First subcase.} In the first case, it follows that there is some $\delta>0$ such that $\text{dist}(x',S_f)>\delta$. We now use the Lipschitz assumption on $f$ to show that $\mathscr{H}^d(\Gamma_f\setminus J_v)>0$. In fact, it holds that $B_\delta(x')\cap S_f=\emptyset$. Now by the second part of Assumption \ref{ass:ident2} there exists a Lipschitz constant $L<+\infty$ such that
    \[|f(x')-f(x'')|\leq L|x'-x''|\leq L\delta\] for all $x''\in B_\delta(x')\subset\mathcal{X}$. Now we pick $\delta>0$ small enough such that the set $Z\coloneqq B_\delta(x')\times B_{L\delta}(f(x'))$ has diameter $\eta$. The diameter of $Z$ is $\text{diam}(Z) = \delta\sqrt{2+L^2}$, so that if we define $\delta< \frac{\eta}{\sqrt{2+L^2}}$, then $\text{diam}(Z)<\eta$. We can do this under Assumption \ref{ass:ident2}. This implies in particular that $Z\cap J_v=\emptyset$.

    Now the fact that $\mathscr{H}^d(Z\cap\Gamma_f) = \mathscr{H}^d((Z\cap\Gamma_f)\setminus J_v)>0$ follows from an application of the area formula since $\Gamma_f$ is a Lipschitz graph. In fact, we get \citep[e.g.][p.~88]{ambrosio2000functions}
    \[\mathscr{H}^{d}(Z\cap\Gamma_f) = \int_{B_\delta(x')} \sqrt{1+ \left\lvert \nabla f\right\rvert^2}\dd\mathcal{L}^{d}>0.\] 
    Denote 
    \[W\coloneqq \left\{(x,t)\in J_v: x\in B_\delta(x')\right\}.\] We have just shown that $J_v$ does not contain $Z\cap\Gamma_f$ on a set $B_\delta(x')$ of positive $d$-dimensional Hausdorff measure, which implies that for $\mathscr{H}^d$-almost every $(x,t)\in W$, $t\neq f(x)$. Moreover, by definition $|v^+(x,t)-v^-(x,t)|>0$ for $\mathscr{H}^d$-almost all $(x,t)\in J_v$. Since $f\in SBV(\mathcal{X})$ it holds by definition that $J_v$ is measurable with respect to $\mathscr{H}^d$. Therefore,
    \begin{multline}\label{eq:split} \lim_{\lambda\to+\infty}\sup_{p\in K}\int_{J_v} (v^+-v^-)p\cdot\rho_v\dd\mathscr{H}^d \\= \lim_{\lambda\to+\infty}\sup_{p\in K}\left[\int_{W} (v^+-v^-)p\cdot\rho_v\dd\mathscr{H}^d+ \int_{J_v\setminus W} (v^+-v^-)p\cdot\rho_v\dd\mathscr{H}^d\right].
    \end{multline}
   Now note that it must hold that $\mathscr{H}^d(W)>0$. In fact, since $v\in SBV(\mathcal{X}\times\mathbb{R})$, it holds that $J_v$ is countably $\mathscr{H}^d$-rectifiable \citep[e.g.][chapter 4]{ambrosio2000functions}. Since $B_\delta(x')$ is compact, $W$ is also compact. This directly implies by Proposition 2.66 in \citet{ambrosio2000functions} that $\mathscr{H}^d(W)\geq \mathcal{L}^d(\pi_{\mathcal{X}}(W))$, where $\pi_{\mathcal{X}}$ is the projection onto $\mathcal{X}$. But $\pi_{\mathcal{X}}(W) = B_{\delta}(x')$, which is of positive $\mathcal{L}^d$ measure. The first term on the right hand side of \eqref{eq:split} diverges to $+\infty$ because the range of $p^t$ increases to $\mathbb{R}$ as $\lambda\to+\infty$ and $t''\neq f(x'')$ for $\mathcal{L}^d$-almost every $x''\in B_\delta(x')$. This implies that such a $v$ cannot be a solution to \eqref{eq:main_problem} in the limit as $\lambda\to+\infty$.

   \emph{Second subcase.} We now suppose $x\in S_f$. 
   Since $x\in S_f$, we now have to work with the points $(x,f^+(x))$ and $(x,f^-(x))$. Without loss of generality focus on the former and suppose that $\text{dist}((x,f^+(x)), J_v(\lambda))>\eta$ for all $\lambda>0$ and as $\lambda\to+\infty$. By Assumption \ref{ass:ident1}, we know that $\mathscr{H}^{d-1}(S_f)<+\infty$, so that \citep[e.g.][Theorem 4.7]{mattila1999geometry} $\mathcal{L}^d(S_f)=0$ on $\mathcal{X}$. The first part of Assumption \ref{ass:ident2} implies that $\mathcal{L}^d(\bar{S}_f)=0$ since it implies $\mathscr{H}^{d-1}(\bar{S}_f\setminus S_f)=0$. This further implies that the projection $\pi_\mathcal{X}(B_\eta((x,f^+(x))))$ is such that it intersects $\mathcal{X}\setminus \bar{S}_f$ in such a way that 
   \[\mathcal{L}^d(\pi_\mathcal{X}(B_\delta((x,f^+(x))))\cap (\mathcal{X}\setminus \bar{S}_f))>0.\]
   Moreover, by Assumption \ref{ass:ident2} we know that $f(x)$ is Lipschitz away from $\bar{S}_f$, which directly implies that there must be points $(x',f(x'))$ in $B_{\delta}((x,f^+(x)))$ with $\text{dist}((x',f(x')), J_v)>0$. This brings us back to the first subcase above and we can derive a contradiction this way and proves that \eqref{eq:contradic} cannot be a solution in the limit as $\lambda\to+\infty$.\\

   \noindent\emph{Showing \eqref{eq:contradic2} leads to a contradiction.}\\
   Suppose \eqref{eq:contradic2} holds and denote $Z^\eta\coloneqq J_v\setminus \Gamma_f^\eta$ for some fixed $\eta>0$. By our assumption, we know that $Z^\eta$ is not empty, so there is some $(x,t)\in (\mathcal{X}\times\mathbb{R})\cap Z^\eta$.  As before $J_v$ must be countably $\mathscr{H}^d$-rectifiable since $v\in SBV(\mathcal{X}\times\mathbb{R})$. This implies \citep[e.g.][Lemma 15.5 (2)]{mattila1999geometry} that $Z^\eta$ is countably $\mathscr{H}^d$-rectifiable. One direction of a theorem based on results by Besicovitch-Marstrand-Mattila \citep[Theorem 2.63]{ambrosio2000functions} implies that 
   \[\theta_d\left(Z^\eta,(x,t)\right)=1\qquad\text{$\mathscr{H}^d$-almost every $(x,t)\in\mathcal{X}\times\mathbb{R}$},\] where
   \[\theta_d\left(A,y\right)\coloneqq \liminf_{r\downarrow 0}\frac{\mathscr{H}^d\left(B_r(y)\right)}{\omega_dr^d}\] is the lower density of the Hausdorff measure with respect to the volume of the $d$-dimensional unit ball $\omega_d$. We may hence assume that there exist some $r_0>0$ and $c>0$ such that
   \[\mathscr{H}^d(Z^\eta\cap B_r((x,t)))\geq  c r^d,\qquad (x,t)\in Z, \thickspace 0<r<r_0. \]

 Recall that for any $(x,t)\in Z^\eta\cap B_r((x,t))$ it must hold that $t\neq f(x)$. Since $Z^\eta$ is by definition $\mathscr{H}^d$-measurable, it holds by the same argument as in \eqref{eq:split} that
    \begin{equation}\label{eq:split2} 
    \begin{aligned}
    &\lim_{\lambda\to+\infty}\sup_{p\in K}\int_{J_v} (v^+-v^-)p\cdot\rho_v\dd\mathscr{H}^d \\
    =& \lim_{\lambda\to+\infty}\sup_{p\in K}\left[\int_{Z^\eta\cap B_r((x,t))} (v^+-v^-)p\cdot\rho_v\dd\mathscr{H}^d+  \int_{J_v\setminus (Z^\eta\cap B_r((x,t)))} (v^+-v^-)p\cdot\rho_v\dd\mathscr{H}^d\right].
    \end{aligned}
    \end{equation}
   Now focus on the first term on the right and recall that we only consider the limit as $\lambda\to+\infty$. Hence, for every constant $M>0$ and every $r>0$ there exists $\lambda>0$ such that the first term on the right hand side of \eqref{eq:split2} is greater than $M$. This means that the first term on the right hand side of \eqref{eq:split2} diverges to $+\infty$ as $\lambda\to+\infty$ because we can always pick a large enough $\lambda$ for every $r$ and because the range of $p^t$ increases to $\mathbb{R}$ as $\lambda\to+\infty$, and $t\neq f(x)$ for every $(x,t)\in Z$. This shows the contradiction to \eqref{eq:contradic2} for every $\eta>0$ and proves that $\lim_{\lambda\to+\infty} d_H(J_{v},\Gamma_f)=0$ for $\lambda$ that diverges fast enough. 
\end{proof}

\subsection{Proof of Proposition \ref{prop:consistency_det}}
We now proof convergence for a \emph{deterministic} analogue of the problem, for which the estimators $\hat{f}_{Nn}$ and $\hat{f}_{X,h(n)}$ are known functions whose values are given in the center of each pixel. We do this for two reasons. First, it clarifies the proof for the statistical setting. Second, it provides a novel convergence result in the mathematical literature on image recognition, complementing recent convergence results \citep{caroccia2020mumford, chambolle2021learning, ruf2019discrete}. Throughout, we assume that $v \in [0,1]^{d+1}$ without loss of generality since all functions are defined on a compact subset of $\mathbb{R}^{d+1}$ by Assumption \ref{ass:stats}.

To prove Proposition \ref{prop:consistency_det}, we require a lemma, which relates the discrete problem to the approximation of the continuous problem via cubes.
In the following, and to lighten the notational burden in proofs, we define  \[v^\uparrow_{k_1,\ldots,k_j,\ldots, k_{d+1}}\coloneqq v_{k_1,\ldots,k_{j}+1,\ldots,k_{d+1}}\] the forward value of $v_{k_1,\ldots,k_{j},\ldots,k_{d+1}}$ for a given dimension $k_j$. As in the main text, we write $k\coloneqq k_1,\ldots,k_{d+1}$, and hence $v_k \coloneqq v_{k_1,\ldots,k_{d+1}}$. We write $v_k^j$ if we want to emphasize one specific dimension $j=1,\ldots, d+1$, otherwise the dimension is defined by the context.
\begin{lemma}\label{lem:cubeapprox}
For the empirical analogues $Dv_N$ of $Dv$ and $p_N$ of $p$ as defined in the main text, it holds that
\[\int_{[0,1]^{d+1}} p\cdot Dv_N = \frac{1}{N^{d+1}}\langle p_N,D_Nv_N\rangle_N,\] with
\[\langle p_N,D_Nv_N\rangle_N =\sum_{0\leq k_1,\ldots, k_{d+1}\leq N} N\left(v^\uparrow_{k_1,\ldots, k_{d+1}} - v_{k_1,\ldots, k_{d+1}}\right) p^\uparrow_{k_1,\ldots,k_{d+1}}.\]
\end{lemma}
\begin{proof}[Proof of Lemma \ref{lem:cubeapprox}]
Note that $v_N\in SBV\left([0,1]^{d+1}\right)$ for all $N\in\mathbb{N}$, since it is a piecewise constant function and the partition $\mathcal{Q}_N$ is a Cacciopolli partition by Theorem 4.16 in \citet{ambrosio2000functions} in combination with Theorem 4.5.11 in \citet{federer2014geometric}. 
We can therefore write
\begin{align*}
    \int_{[0,1]^{d+1}} p\cdot Dv_N
    =\sum_{0\leq k\leq N}\int_{Q_{k}}p\cdot \nabla v_{k}\dd x+\sum_{0\leq k\leq N}\int_{\partial Q_k^{\uparrow}\cap\partial Q_k}\left(v_k^\uparrow - v_k\right)p\cdot s\dd\mathscr{H}^d,
\end{align*}
where the orientation $s\in \mathbb{S}^{d+1}$ is chosen in the ``forward direction'', i.e. from $j$ to $j+1$, which means that it takes the form of unit vectors $e_j\in\mathbb{R}^{d+1}$ with zeros everywhere except a $1$ in one of the $j$ positions. Since $v_k$ is constant on $Q_k$ it holds that $\nabla v_k=0$ on $Q_k$, so that the first term vanishes. 

For the second term, we have
\begin{align*}
    \sum_{0\leq k\leq N}\int_{\partial Q_k^\uparrow\cap\partial Q_k}\left(v_k^\uparrow - v_k\right)p_k\cdot s\dd\mathscr{H}^d
    = \sum_{0\leq k\leq N} \left(v_k^\uparrow - v_k\right)p_k^\uparrow\frac{1}{N^d},
\end{align*}
where we define $p_k^\uparrow = p_k\cdot e^\uparrow$, where $e^\uparrow$ is the unit vector in the forward direction for respective $j$. The inequalities follow because $p_k$ is constant on the boundary of the respective cube where it is defined and since the Hausdorff measure of a face of the hypercube in $d+1$ dimensions with sidelength $\frac{1}{N}$ is $\frac{1}{N^d}$. Also recall that the orientation $s$ is in terms of forward differences, so that $p_k$ in the above expression is the one that lies on the boundary of $Q_k$ and the corresponding $Q_k^\uparrow$. 

We therefore have
\begin{align*}
    \int_{[0,1]^{d+1}} p\cdot Dv_N =& \frac{1}{N^d}\sum_{0\leq k\leq N}\left(v_k^\uparrow - v_k\right)p_k^\uparrow\\
    =&\frac{1}{N^{d+1}}\sum_{0\leq k\leq N}N\left(v_k^\uparrow - v_k\right)p_k^\uparrow\\
    \equiv & \frac{1}{N^{d+1}}\langle p_N,D_Nv_N\rangle_N.
\end{align*}
\end{proof}
With this result, we are ready to prove Proposition \ref{prop:consistency_det}. 

\begin{proof}[Proof of Proposition \ref{prop:consistency_det}]
In the following, we denote weak$^*$-convergence by $\overset{*}{\rightharpoonup}$. 
To prove $\Gamma$-convergence we need to show \citep[Proposition 8.1]{dal2012introduction}
\begin{itemize}
    \item[(i)] for every $v\in C$ and every sequence $v_N\in \tilde{C}_N$ with  $v_{N}\overset{*}{\rightharpoonup} v$ it holds that
    \[E(v)\leq\liminf_{N\to\infty} E_N(v_N)\] and
    \item[(ii)] for every $v\in C$ there exists a sequence $v_N\in \tilde{C}_N$ with $v_{N}\overset{*}{\rightharpoonup}v$ such that 
    \[E(v)\geq \limsup_{N\to\infty} E_N(v_N).\]
\end{itemize}

\emph{Part (i)}: Let $v_N\overset{*}{\rightharpoonup} v$ and assume $\liminf_{N\to+\infty} E_N(v_N)<+\infty$. Let $p\in C_c^\infty([0,1]^{d+1},\mathbb{R}^{d+1})\cap K$.

As in the main text, the discrete approximation of $p$ is achieved by decomposing $[0,1]^{d+1}$ into hypercubes of sidelength $\frac{1}{N}$. We define $p_k^{j+\frac{1}{2}}$ as the average flux through the boundary of two adjancent hypercubes $Q_k^{j+1}$ and $Q_k^j$, i.e.
\[p_k^{j+\frac{1}{2}} = N^d\int_{\partial Q_k^{j+1}\cap \partial Q_k^j}p\cdot s\dd\mathscr{H}^d,\] where $s\in\mathbb{S}^d$ is the orientation chosen in the direction from $j$ towards $j+1$. We then define 
\[p_N = \sum_{0\leq k\leq N}p_k.\] 

Now we need to analyze the above constructed $p_N$ in terms of $\hat{K}$, which consists of two constraints. Let us start with the first. Each $p_k$ is a $(d+1)$-dimensional vector, and we have to distinguish between the first $d$ entries of this vector and the last entry. To do this, we will write $p_k^x$ as the vector consisting of the first $d$ entries and $p_k^t$ as the last entry of $p_k$. Furthermore, for the cube $Q_k$ we denote its centerpoint by $x_k$.

We can now analyze
\begin{equation}\label{eq:tob1}
\begin{aligned}
    \left\lvert \frac{1}{N}\sum_{\kappa_1\leq k_{d+1}\leq \kappa_2}  p_k^x\right\rvert =& \left\lvert\frac{1}{N}\sum_{\kappa_1\leq k_{d+1}\leq \kappa_2}^{} p_k^x+ p^x(x_k) - p^x(x_k)\right\rvert\\
    \leq & \left\lvert \frac{1}{N}\sum_{\kappa_1\leq k_{d+1}\leq \kappa_2}^{} p_k^x- p^x(x_k)\right\rvert + \left\lvert \frac{1}{N}\sum_{\kappa_1\leq k_{d+1}\leq \kappa_2} p^x(x_k)\right\rvert
\end{aligned}
\end{equation}
for $\kappa_1,\kappa_2\in \{1,\ldots,N\}$. Since $p$ is smooth and in particular Lipschitz, we can bound the first term in terms of the Lipschitz constant and the edge length. Let $L<+\infty$ be the Lipschitz constant and note that the Euclidean distance between $x_k$ and the center of one of the boundaries of $Q_k$ is $\frac{1}{2N}$. We can hence bound the first term as
\begin{align*}
    \left\lvert \frac{1}{N}\sum_{\kappa_1\leq k_{d+1}\leq \kappa_2}^{} p_k^x - p^x(x_k)\right\rvert
    \leq &\frac{1}{N}\sum_{\kappa_1\leq k_{d+1}\leq \kappa_2} \left\lvert {}p_k^x - p^x(x_k)\right\rvert\\
    \leq & \frac{1}{N}\sum_{\kappa_1\leq k_{d+1}\leq \kappa_2} \frac{L}{2N}\\
    \leq & \frac{L}{2N}.
\end{align*}
The second term in \eqref{eq:tob1} is bounded above by $\nu+ o(1)$ since $p\in K$ and a Riemann sum argument. So overall, we have
\[\frac{1}{N}\left\lvert \sum_{\kappa_1\leq k_{d+1}\leq \kappa_2} p_k^x\right\rvert\leq \nu+ o(1)+ \frac{L}{2N}.\]
So for $p_N\in \hat{K}$ we would have to normalize each ${}p_k^x$ as
\[\tilde{p}^x_k = {}p_k^x\frac{\nu}{\nu+o(1)+\frac{L}{2N}}.\]

Let us now analyze the second part of the constraint $\hat{K}$. Recall that the constraint for $K$ reads
\[\left\lvert p^x(x,t)\right\rvert^2\leq 4\left(p^t(x,t)+\lambda(t-f(x))^2\right).\]
By Assumption \ref{ass:stats} the term in square brackets is bounded for all $t$ and $x$. 
We now analyze
\[\left\lvert {}p_k^x\right\rvert^2\leq \left(\left\lvert {}p_k^x-p^x(x_k)\right\rvert+\left\lvert p^x(x_k)\right\rvert\right)^2,\]
where $x_k$ is again the center point of the cube $Q_k$. 
The first term can be bounded in the same way as above:
\[\left\lvert p_k^x-p^x(x_k)\right\rvert\leq\frac{L}{2N}.\]
By the fact that $p\in K$ the second term satisfies 
\begin{align*}
    \left\lvert p^x(x_k)\right\rvert \leq& 2\sqrt{p^t(x_k)+\lambda\left(t_k - f(\tilde{x}_k)\right)^2}\\
    \leq & 2\sqrt{\left\lvert p^t(x_k) -p_k^t\right\rvert+{}p_k^t+\lambda\left(t_k - f(\tilde{x}_k)\right)^2}\\
    \leq & 2\sqrt{\frac{L}{2N}+{}p_k^t+\lambda\left(t_k - f(\tilde{x}_k)\right)^2},
\end{align*}
where we define $x_k\equiv (\tilde{x}_k,t_k)\in \mathbb{R}^{d+1}$ and the second line follows from $(a-b)\leq |a-b|$. Note that $f$ is a bounded function by Assumption \ref{ass:stats}, which in our discretization scheme is approximated in the center of each cube. 
Putting both terms together gives
\begin{align*}
    \left\lvert {}p_k^x\right\rvert^2\leq &\left(\frac{L}{2N}+\sqrt{\frac{2L}{N}+4\left({}p_k^t+\lambda\left(t_k-f(\tilde{x}_k)\right)^2\right)}\right)^2\\
    \leq &\left(\frac{L}{2N} + \sqrt{\frac{2L}{N}} + \sqrt{4\left({}p_k^t+\lambda\left(t_k-f(\tilde{x}_k)\right)^2\right)}\right)^2\\
    = &\left(\frac{L}{2N}+\sqrt{\frac{2L}{N}}\right)^2 + 2\left(\frac{L}{2N}+\sqrt{\frac{2L}{N}}\right)\sqrt{4\left({}p_k^t+\lambda\left(t_k-f(\tilde{x}_k)\right)^2\right)}  +  \\ 
    & 4\left({}p_k^t+\lambda\left(t_k-f(\tilde{x}_k)\right)^2\right) \\
    & \coloneqq C_p,
\end{align*}
where we used the inequality $\sqrt{a+b}\leq \sqrt{a}+\sqrt{b}$ in the second line.

So in order to make $\tilde{p}_N\in \hat{K}_N$, we would need to normalize 
\begin{align*}
    \tilde{p}^x_k = p_k^x \left(\frac{\nu}{\nu+o(1)+\frac{L}{2N}}\right)\qquad\text{and}\qquad
    \tilde{p}^x_k = p_k^x\sqrt{\frac{4\left({}p_k^t+\lambda\left(t_k-f(\tilde{x}_k)\right)^2\right)}{C_p}}.
\end{align*}
Taking the minimum of these two expressions for $\tilde{p}_k^x$, and calling this minimum $0\leq m(N)\leq 1$ we have $\tilde{p}_k^x = m(N) p_k^x\leq p_k^x$ with $m(N)\to1$ as $N\to\infty$. It also follows that
\[\frac{1}{N^{d+1}}\langle p_N,D_Nv_N\rangle_N \leq m(N)E_N(v_N).\]

We now show the convergence of the liminf using the information we have just derived. First, since $v_N\overset{*}{\rightharpoonup}v$ is holds by definition of weak$^*$-convergence that 
\[\int_{[0,1]^{d+1}} p\cdot Dv = \lim_{N\to\infty} \int_{[0,1]^{d+1}} p\cdot Dv_N = \lim_{N\to\infty} \frac{1}{N^{d+1}}\langle p_N,D_Nv_N\rangle_N,\] where the second inequality follows from Lemma \ref{lem:cubeapprox}.
From the above argument, letting $N\to+\infty$ and then taking the supremum over $p$ we get
\[E(v)\equiv \sup_{p\in K}\int_{[0,1]^{d+1}} p\cdot Dv\leq \liminf_{N\to\infty} E_N(v_N).\] 
Note that the Neumann boundary conditions of the population problem are preserved in the finite sample version and consistency for those follows immediately by construction. This shows the first part of $\Gamma$-convergence. 

\emph{Part (ii)}: For the second part we have to construct a recovering sequence \citep{dal2012introduction, chambolle2021learning}, which is a sequence $v_N$ of discrete functions of the form
\[\tilde{v}_N(x)\coloneqq \sum_{0\leq k\leq N} \tilde{v}_k \mathds{1}\{x\in Q_k\}\]
with $\tilde{v}_N\overset{*}{\rightharpoonup} v$ and 
\[\limsup_{N\to\infty} E_N(\tilde{v}_N)\leq E(v). \]

By \citet[Theorem 3.9]{ambrosio2000functions} we can approximate $v\in BV([0,1]^{d+1})$ by a sequence $v_N\in C_c^\infty([0,1]^{d+1})$ of mollifiers in $L^1$ with 
\[\lim_{N\to\infty} \int_{[0,1]^{d+1}} \left\lvert \nabla v_N\right\rvert\dd x = \left\lvert Dv\right\rvert.\] We then construct the $\tilde{v}_N$ by 
\[\tilde{v}_N(x) = \sum_{0\leq k\leq N} v_N(x_k)\mathds{1}\{x\in Q_k\},\]
where $x_k$ is the center point of the corresponding cube $Q_k$. 
But note that 
\begin{align*}
    &\lim_{N\to\infty}\lvert \tilde{v}_N - v_N\rvert_{L^1([0,1]^{d+1})} \\=& 
    \lim_{N\to\infty}\int_{[0,1]^{d+1}} \left\lvert \sum_{0\leq k\leq N} v_N(x_k)\mathds{1}\left\{x\in Q_k\right\} - v(x)\right\rvert\dd x\\
    \leq & \lim_{N\to\infty}\int_{[0,1]^{d+1}} \left\lvert \sum_{0\leq k\leq N} v_N(x_k)\mathds{1}\left\{x\in Q_k\right\} - v_N(x)\right\rvert\dd x + o(1)\\
    = & \lim_{N\to\infty}\int_{[0,1]^{d+1}} \left\lvert \sum_{0\leq k\leq N} v_N(x_k)\mathds{1}\left\{x\in Q_k\right\} - \sum_{0\leq k\leq N} v_N(x)\mathds{1}\left\{x\in Q_k\right\}\right\rvert\dd x + o(1)\\
    \leq & \lim_{N\to\infty}\sum_{0\leq k\leq N}\int_{Q_k} \left\lvert  v_N(x_k) - v_N(x)\right\rvert\dd x + o(1)\\
    \leq &\lim_{N\to\infty} \sum_{0\leq k \leq N}\frac{L_N+o(1)}{2N} \frac{1}{N^{d+1}} +o(1)\\
    \leq &\lim_{N\to\infty} \sum_{0\leq k \leq N}\frac{L_N+o(1)}{2N^{d+2}}  +o(1)\\
    \leq &\lim_{N\to\infty} \frac{L_N+o(1)}{2N} +o(1) = 0,
\end{align*}
where we used the mean-value inequality and 
\[L_N\coloneqq \int_{[0,1]^{d+1}}\left\lvert \nabla v_N\right\rvert\dd x,\] which we know converges to a bounded quantity since $v\in BV([0,1]^{d+1})$. This shows that $\tilde{v}_N\to v$ in $L^1([0,1]^{d+1})$. 

Now we need to show that $\tilde{v}_N\overset{*}{\rightharpoonup}v$, for which we have to show that $\sup_{N\in\mathbb{N}}\left \lvert D\tilde{v}_N\right\rvert<+\infty$. Since $v_N\to v$ in $L^1([0,1]^{d+1})$, it holds that 
\[|Dv_N| = \int_{[0,1]^{d+1}}\left\lvert \nabla v_N\right\rvert \dd x <+\infty\] for $N$ large enough \citep[Theorem 3.9]{ambrosio2000functions}.  
By construction, for any $N\in\mathbb{N}$ it holds that
\begin{align*}
    \int_{[0,1]^{d+1}} p\cdot D\tilde{v}_N = &\sum_{0\leq k\leq N} \int_{Q_k} p\cdot D\tilde{v}_N\\
    = & \sum_{0\leq k\leq N} \int_{Q_k} p\cdot \nabla\tilde{v}_N + \sum_{0\leq k\leq N}\int_{\partial Q_k^{\uparrow}\cap\partial Q_k}\left(\tilde{v}_k^\uparrow - \tilde{v}_k\right)p\cdot s\dd\mathscr{H}^d\\
    =& \sum_{0\leq k\leq N} \left(\tilde{v}_k^\uparrow - \tilde{v}_k\right)p^\uparrow\frac{1}{N^d}\\
    \leq &\frac{1}{N^d} \sum_{0\leq k\leq N} \frac{L}{N} p^\uparrow,
\end{align*}
where $p^\uparrow$ is the value of $p$ on the boundary of one of the forward directions. Since $p$ is continuous and $[0,1]^{d+1}$ is compact, it follows
\[\sup_{N\in\mathbb{N}}\sup_{p\in C_0([0,1]^{d+1})} \int_{[0,1]^{d+1}} p\cdot D\tilde{v}_N\leq\sup_{N\in\mathbb{N}}\sup_{p\in C_0([0,1]^{d+1})} \frac{1}{N^d} \sum_{0\leq k\leq N} \frac{L}{N} p^\uparrow<+\infty, \]
which shows that $\sup_{N\in\mathbb{N}}\left \lvert D\tilde{v}_N\right\rvert<+\infty$. This in turn implies that that $\tilde{v}_N\overset{*}{\rightharpoonup} v$, which shows that $\tilde{v}_N$ is the required recovering sequence, since the Neumann boundary conditions are also preserved as before. 
\end{proof}

\subsection{Proof of Theorem \ref{thm:consistency_rand}}
\begin{proof}
The proof follows along similar lines to the proof of Proposition \ref{prop:consistency_det}, except that we have to account for the randomness of the estimators $\hat{f}_{Nn}(\tilde{x}_k)$ and $\hat{f}_{X,h(n)}(\tilde{x}_k)$. The only random element in the optimization problem is one constraint in $\hat{K}$, and this is what we analyze. All the other arguments are the same as in the deterministic case and are omitted. 

Recall that the constraint for $\hat{K}_{Nn}$ reads
\[\left\lvert p^x(x,t)\right\rvert^2\leq 4\hat{f}_{X,h(n)}(\tilde{x}_k)\left(p^t(x,t)+\lambda_n \hat{f}_{X,h(n)}(\tilde{x}_k)\left(t-\hat{f}_{Nn}(\tilde{x}_k)\right)^2\right).\]
As before we analyze
\[\left\lvert {}p_k^x\right\rvert^2\leq \left(\left\lvert {}p_k^x-p^x(x_k)\right\rvert+\left\lvert p^x(x_k)\right\rvert\right)^2,\]
where $x_k$ is again the center point of the cube $Q_k$. 
The first term can be bounded in the same way as in Proposition \ref{prop:consistency_det}:
\[\left\lvert {}p_k^x-p^x(x_k)\right\rvert\leq\frac{L}{2N}.\]
By the fact that $p\in K$ the second term satisfies 
\begin{align*}
    \left\lvert p^x(x_k)\right\rvert \leq& 2\sqrt{\hat{f}_{X,h(n)}(\tilde{x}_k )p^t(x_k)+\lambda_n\hat{f}^2_{X,Nn}(\tilde{x}_k )\left(t_k - \hat{f}_{Nn}(\tilde{x}_k)\right)^2}\\
    \leq & 2\sqrt{\hat{f}_{X,h(n)}(\tilde{x}_k ) \left\lvert p^t(x_k) - {}p_k^t\right\rvert+\hat{f}_{X,h(n)}(\tilde{x}_k ){}p_k^t+\lambda_n\hat{f}^2_{X,Nn}(\tilde{x}_k )\left(t_k - \hat{f}_{Nn}(\tilde{x}_k)\right)^2}\\
    \leq & 2\sqrt{\hat{f}_{X,h(n)}(\tilde{x}_k )\frac{L}{2N}+ \hat{f}_{X,h(n)}(\tilde{x}_k ){}p_k^t+\lambda_n\hat{f}_{X,h(n)}^2(\tilde{x}_k )\left(t_k - \hat{f}_{Nn}(\tilde{x}_k)\right)^2},
\end{align*}
where we define $x_k\equiv (\tilde{x}_k,t_k)\in \mathbb{R}^{d+1}$ and the second line follows from $(a-b)\leq |a-b|$. The difference to the deterministic case is that $\hat{f}_{Nn}$ and $\hat{f}_{X,h(n)}$ are random estimators of $f(x)$ and $f_X(x)$, which in our discretization scheme are imputed in the center of each cube. 
Putting both terms together gives
\[\left\lvert {}p_k^x\right\rvert^2\leq \left(\frac{L}{2N}+\sqrt{\hat{f}_{X,h(n)}(\tilde{x}_k)\frac{2L}{N} + \hat{f}_{X,h(n)}(\tilde{x}_k){}p_k^t+\lambda_n\hat{f}^2_{X,Nn}(\tilde{x}_k)\left(t_k-\hat{f}_{Nn}(\tilde{x}_k)\right)^2}\right)^2.\]

Just as in the discrete case, we bound this further by
\begin{multline}\label{eq:largep}
\left\lvert {}p_k^x\right\rvert^2\leq\left(\frac{L}{2N}+\sqrt{\hat{f}_{X,h(n)}(\tilde{x}_k)\frac{2L}{N}}\right)^2 \\ + 2\left(\frac{L}{2N}+\sqrt{\hat{f}_{X,h(n)}(\tilde{x}_k)\frac{2L}{N}}\right)\sqrt{4\hat{f}_{X,h(n)}(\tilde{x}_k)\left({}p_k^t+\lambda_n\hat{f}_{X,h(n)}(\tilde{x}_k)\left(t_k-\hat{f}_{Nn}(\tilde{x}_k)\right)^2\right)} \\+ 4\hat{f}_{X,h(n)}(\tilde{x}_k)\left({}p_k^t+\lambda_n\hat{f}_{X,h(n)}(\tilde{x}_k)\left(t_k-\hat{f}_{Nn}(\tilde{x}_k)\right)^2\right)
\end{multline}

\noindent\emph{Addressing the randomness:}\\
Let us consider the two terms $\hat{f}_{X,h(n)}(\tilde{x}_k)$ and $\hat{f}_{Nn}(\tilde{x}_k)$
one by one. First, by definition, $\hat{f}_{X,h(n)}(\tilde{x}_k)$ is a Nadaraya-Watson estimator. Corollary 3 in \citet{irle1997consistency} in conjunction with Assumption \ref{ass:stats} and the assumptions stated in Theorem \ref{thm:consistency_rand} guarantees that $\hat{f}_{X,h(n)}(x)$ converges almost surely for almost every $x\in\mathbb{R}^d$ to $f_X(x)$. To see this, note that strict stationarity of $\{X_i\}$ implies asymptotic stationarity and in particular ``summability'' of the first two terms in the last equation on page 131 of \citet{irle1997consistency}. The last term of the summability condition of \citet{irle1997consistency} follows since under Assumption \ref{ass:stats} $s>1$, which implies that there exists a small enough $k>0$ such that $\sum_{i=1}^{\infty} \alpha(i)^{1/k}<+\infty$; in fact, we can set $0<k<s$. The last assumption in Corollary 3 of \citet{irle1997consistency} follows directly from Assumption \ref{ass:stats}. Then, by construction, $\tilde{x}_k$ is a Lebesgue point, which implies that $\hat{f}_{X,h(n)}(\tilde{x}_k)$ converges almost surely to $f_X(\tilde{x}_k)$ since the latter is a density. The Continuous Mapping Theorem \citep[e.g.][Theorem 1.3.6]{van1996weak} implies that $\hat{f}_{X,h(n)}^2(\tilde{x}_k)$ also converges almost surely to $f^2_X(\tilde{x}_k)$.

Now write $Z_{ik}\coloneqq 1\{i:X_i\in Q_k\} (f(X_i)+\varepsilon_i)$ with $S_k\coloneqq \sum_{i=1}^n Z_{ik}$ for the numerator; and $I_{ik}\coloneqq 1\{X_i\in Q_k\}$ with $M_k \coloneqq \sum_{i=1}^n I_{ik}$. Then we can write \[\hat{f}_{Nn}(\tilde{x}_k) = \frac{1}{\sum_{i=1}^n 1\left\{i: X_i\in Q_k\right\}}\sum_{i: X_i\in Q_k} (f(X_i)+\varepsilon_i) = \frac{S_k}{M_k}\].  

Fix a cube $Q_k$ and define the success probability of the Bernouilli random variable $I_{ik}$ as $\pi_k = \int_{Q_k} f_X(x) dx$. By Assumption \ref{ass:stats}, $\pi_k\geq \frac{c_x}{N^d}$. By Theorem 3.49 in \citet{white2014asymptotic}, the process $\{I_{ik}\}_i$ is $\alpha$-mixing of the same size since $\{(X_i,\varepsilon_i)\}_i$ is $\alpha$-mixing and $I_{ik}$ is a measurable function of the process. 

We want to bound $|\hat{f}_{Nn}(\tilde{x}_k) - \frac{\mathbb{E}[S_k]}{\mathbb{E}[M_k]}|$. We can write:
\[\left\lvert \hat{f}_{Nn}(\tilde{x}_k) - \frac{\mathbb{E}[S_k]}{\mathbb{E}[M_k]}\right\rvert = \left\lvert\frac{S_k}{M_k} - \frac{\mathbb{E}[S_k]}{\mathbb{E}[M_k]}\right\rvert.\] Note that this is the correct relationship to bound since we want to compare the estimator to $\mathbb{E}[S_k | X_i\in Q_k] = \frac{\mathbb{E}[S_k]}{\mathbb{E}[M_k]}.$ Now write
\begin{align*}
    \left\lvert \frac{S_k}{M_k} - \frac{\mathbb{E}[S_k]}{\mathbb{E}[M_k]}\right\rvert  &\leq  \left\lvert \frac{S_k}{M_k} - \frac{\mathbb{E}[S_k]}{M_k}\right\rvert +\left\lvert \frac{\mathbb{E}[S_k]}{M_k}- \frac{\mathbb{E}[S_k]}{\mathbb{E}[M_k]}\right\rvert\\
    &= \frac{1}{M_k} \left\lvert S_k - \mathbb{E}[S_k]\right\rvert + \mathbb{E}[S_k]\left\lvert \frac{1}{M_k} - \frac{1}{\mathbb{E}[M_k]}\right\rvert,
\end{align*}
and consider the first term first.

To bound the first term, consider the two events $\{\sum_i I_{ik}<\frac{n c_x}{2N^d}\}$ and $\{\sum_i I_{ik}\geq \frac{n c_x}{2N^d}\}$. The first one is the event that there are ``not enough'' data points in $Q_k$, the ``bad event''. 

 We get: 
\begin{align*}
    P\left(\sum_{i=1}^n I_{ik} < \frac{n c_x}{2N^d}\right) &= P\left(\sum_{i=1}^n I_{ik} - \sum_{i=1}^n \mathbb{E}[I_{ik}]< \frac{n c_x}{2N^d} - \sum_{i=1}^n \mathbb{E}[I_{ik}]\right)\\
    &\leq P\left(\sum_{i=1}^n I_{ik} - \sum_{i=1}^n \mathbb{E}[I_{ik}]< -\frac{n c_x}{2N^d}\right)\\
    & \leq P\left(\left\lvert\sum_{i=1}^n I_{ik} - \sum_{i=1}^n \mathbb{E}[I_{ik}]\right\rvert>\frac{n c_x}{2N^d}\right)\\
    &\leq 4\exp\left(-\frac{c_x^2}{32}q\right)+ 22\left(1+\frac{8}{c_x}\right)^{1/2}q \alpha\left(\frac{n}{2N^dq}\right)\eqqcolon A(n, N^d,q)
\end{align*}
for any integer $q\in [1,\frac{n}{2N^d}]$, where the second line follows from $p_k\geq \frac{c_x}{N^d}$ and every $I_{ik}$ is a Bernoulli random variable; and the last line follows from a Hoeffding inequality for $\alpha$-mixing processes \citep[Theorem 1.3]{bosq2012nonparametric}. Above, we can set $q$ to be any slowly changing function of $\frac{n}{N^d}$, like $\lfloor\sqrt{n}\rfloor$ or $\lfloor\log(n)\rfloor$, to trade-off the speed of convergence in both terms. Also, note that this will hold for any $k$ since the smallest probability of landing in a cube is $c_x \frac{1}{N^d}$ by Assumption \ref{ass:stats}. Since we have bounded this ``bad event'', we can now turn to the good event and condition on this.

 We want to bound this with a similar concentration bound, but since $\varepsilon_i$ do not need to have bounded support, we need the more general concentration bound \citep[Theorem 1.4]{bosq2012nonparametric}. Since $f$ is bounded on $\mathcal{X}$, $f(X_i)$ is bounded for all $i$; in addition, $\varepsilon_i$ possesses a moment generating function in a local neighborhood of zero by Assumption \ref{ass:stats}.  Therefore $Y_i=f(X_i)+\varepsilon_i$ satisfies Cram\'er's condition from Theorem 1.4 in \citet{bosq2012nonparametric}. For any $\delta>0$, any $\frac{n}{N^d}\geq 2$, any integer $q\in [1,\frac{n}{2N^d}]$, and any each $k\geq 3$ we hence get:
\begin{align*}
    & P\left(\left\lvert S_k - \mathbb{E}[S_k]\right\rvert>\delta\frac{nc_x}{2N^d}\right)\\
    \eqqcolon &P\left(\left\lvert S_k - \mathbb{E}[S_k]\right\rvert>\tilde{\delta}\frac{n}{N^d} \right)\\
    \leq &
a_1 \exp\left( -\frac{q\tilde{\delta}^2}{25m_2^2 + 5c\tilde{\delta}} \right) 
+ a_2(k)\alpha\left(\left\lfloor \frac{n}{(q+1) N^d} \right\rfloor \right)\eqqcolon B(n, N^d,q, \delta),
\end{align*}
where 
\begin{align*}
a_1 &= 2\frac{n}{qN^d} + 2 \left( 1 + \frac{\tilde{\delta}^2}{25m_2^2 + 5c\tilde{\delta}} \right), 
\quad m_2^2 = \max_{1 \leq t \leq n} \mathbb{E}X_t^2, \\
a_2(k) &= 11\frac{n}{N^d} \left( 1 + \frac{5m_k}{\tilde{\delta}} \right)^{\frac{2k}{2k+1}}, 
\quad m_k = \max_{1 \leq t \leq n} \|X_t\|_k.
\end{align*}
The same argument as above holds, i.e., we can pick an appropriate $q$ that balanced the convergence of both terms, and hence proves convergence to zero of this event. 

In order to put both bounds together and bound the entire first term, note that
\[\left\{\frac{1}{M_k} \left\lvert S_k - \mathbb{E}[S_k]\right\rvert>\delta, M_k\geq \frac{n}{2N^d}\right\}\subset \left\{\left\lvert S_k - \mathbb{E}[S_k]\right\rvert>\delta\frac{n}{2N^d}\right\},\]
so that 
\begin{align*}
P\left(\frac{1}{M_k}\left\lvert S_k-\mathbb{E}[S_k]\right\rvert>\delta\right) \leq& P\left(M_k< \frac{n}{2N^d}\right) + P\left(\left\lvert S_k-\mathbb{E}[S_k]\right\rvert>\delta \frac{n}{2N^d}\right)\\
\leq &A(n,N^d,q_1)+B(n,N^d,q_2)
\end{align*}
for integers $q_1,q_2\in [1,\frac{n}{2N^d}]$.

Now we need to bound the second term $\mathbb{E}[S_k]\left\lvert M_k^{-1} - (\mathbb{E}[M_k])^{-1}\right\rvert$. We approach this in the same way as the previous bound. First, notice that 
\[|M_k^{-1} - (\mathbb{E}[M_k])^{-1}|\leq \left\lvert \frac{|M_k-\mathbb{E}[M_k]|}{M_k\mathbb{E}[M_k]}\right\rvert.\] Now we split the problem into the ``bad'' event $\{M_k<\frac{n}{2N^d}\}$, which we can bound in the exact same way as above by $A(n,N^d,q)$, and the ``good event'' $\{M_k\geq \frac{n}{2N^d}\}$.
For the good event, we have
\[\mathbb{E}[S_k]\left\lvert M_k^{-1}- \left(\mathbb{E}[M_k]\right)^{-1}\right\rvert \leq \frac{\mathbb{E}[S_k]}{\mathbb{E}[M_k]} \frac{|M_k-\mathbb{E}[M_k]|}{M_k}.\]
Now note that 
\[\mathbb{E}[S_k] = \sum_{i=1}^n \mathbb{E}[f(X_i) +\varepsilon_i)|X_i\in Q_k]\cdot P(X_i\in Q_k)\leq \sum_{i=1}^n C\frac{1}{N^d}= C\frac{n}{N^d},\] because $\mathbb{E}[\varepsilon_i|X_i] = 0$, the law of iterated expectations, and $f\leq C<+\infty$ by Assumption \ref{ass:stats}. Similarly, $\mathbb{E}[M_k] = \frac{n}{N^d}$, so that 
$\frac{\mathbb{E}[S_k]}{\mathbb{E}[M_k]}\leq C.$

Now focus on the good event and bound for any $\delta>0$
\begin{align*}
    &P\left(\left\lvert M_k - \mathbb{E}[M_k]\right\rvert>\frac{\delta}{2C}\cdot \frac{n}{N^d}\right)\\
    \leq & 4\exp\left(-\frac{\delta}{32C^2}q\right)+22\left(1+\frac{8C}{\delta}\right)^{1/2} q\alpha\left(\frac{n}{2qN^d}\right) \\
    \eqqcolon & C(n,N^d,q)
\end{align*}
for any integer $q\in[1,\frac{n}{N^d}]$, using the bound established above for $I_{ik}$.
Notice that
\[\left\{\mathbb{E}[S_k]\left\lvert M_k^{-1}-(\mathbb{E}[M_k])^{-1}\right\rvert>\delta, M_k\geq \frac{n}{2N^d}\right\}\subset \left\{|M_k-\mathbb{E}[M_k]|>\frac{\delta}{2C}\cdot \frac{n}{N^d}\right\}.\]
We can therefore bound the second term for $\delta>0$
\begin{align*}
P\left(\mathbb{E}[S_k]\left\lvert M_k^{-1}-\left(\mathbb{E}[M_k]\right)^{-1}\right\rvert >\delta\right)\leq &P\left(M_k< \frac{n}{2N^d}\right) + P\left(|M_k-\mathbb{E}[M_k]|>\frac{\delta}{2C}\cdot \frac{n}{N^d}\right)\\
&\leq A(n,N^d,q) + C(n, N^d, q).
\end{align*}
This implies that we can bound the entire expression as
\[P\left(\left\lvert \hat{f}_{Nn}(\tilde{x}_k)-\frac{\mathbb{E}[S_k]}{\mathbb{E}[M_k]}\right\rvert>\delta\right)\leq 2 A(n,N^d,q) + B(n,N^d,q,\delta) + C(n,N^d,q,\delta).\]
This bound holds for every cube $Q_k$. To obtain a uniform bound over the entire grid, we take the union bound
\[P\left(\max_{1\leq k\leq N^d}\left\lvert \hat{f}_{Nn}(\tilde{x}_k)-\frac{\mathbb{E}[S_k]}{\mathbb{E}[M_k]}\right\rvert>\delta\right)\leq N^d\left(2 A(n,N^d,q) + B(n,N^d,q,\delta) + C(n, N^d,q,\delta)\right).\]
We want this expression to converge to zero. 
Analyze the terms $A(n,N^d,q)$ and $C(n,N^d,q,\delta)$ first. These terms depend on $q$ and have two main terms: the first is exponential in $q$ and the second takes the form $q\cdot \alpha(\frac{n}{2q N^d})$. We get to choose $q\in [1,\frac{n}{2N^d}]$. So pick $q = \left(\frac{n}{N^d}\right)^\rho$ for some $0<\rho<1$. 
The exponential term is of the form $N^d\cdot \exp\left(-c\left(\frac{n}{N^d}\right)^\rho\right)$ for some $0<\rho<1$ and some fixed constant $c$. Since $\frac{n}{N^d}\to+\infty$, the exponential term dominates and this term vanishes exponentially fast as $n\to+\infty$. 

The term involving the $\alpha$-mixing parameter is of the form 
$N^d\cdot(\frac{n}{N^d})^\rho\alpha\left(\frac{n}{2qN^d}\right)$, where $\alpha(m) = O(m^{-s})$ by Assumption \ref{ass:stats}. We want the expression 
\[T(n)= N^d\cdot\left(\frac{n}{N^d}\right)^\rho\cdot\alpha\left(\frac{n^{1-\rho}}{2N^{d(1-\rho)}}\right)\leq N^d\cdot\left(\frac{n}{N^d}\right)^\rho\cdot C\cdot \left(\frac{n^{1-\rho}}{2N^{d(1-\rho)}}\right)^{-s}\]
to be bounded above by $n^{-r}$ for some $r>0$, in order to show almost sure convergence over all bins by a Borel-Cantelli-argument further down. We do this by picking an appropriate $\rho\in(0,1)$. Some calculations show that if 
\[\rho\leq \frac{s-m(1+s)-r}{(1+s)(1-m)}\] then $T(n) = O(n^{-r})$ for some fixed constant for all $n$. Also note that this $\rho\in(0,1)$ because $s>\max\left\{\frac{r+m}{1-m},\frac{r+1}{\frac{2k}{2k+1}(1-m)}\right\}$ by Assumption \ref{ass:stats}.

Now turn to analyzing $B(n,N^d,q,\delta)$, which is again split into an exponential term and the $\alpha$-mixing term. For $q=\left(\frac{n}{N^d}\right)^\rho$ for $\rho\in (0,1)$, the exponential term dominates, so focus again on the $\alpha$-mixing term. We have to analyze $a_2(k)\alpha\left(\left\lfloor\frac{n}{(q1)N^d}\right\rfloor\right)$ in Theorem 1.4 of \citet{bosq2012nonparametric}. To simplify, we assume that $\frac{n}{(q1)N^d}$ is an integer without loss of generality. We want to guarantee that $N^d\frac{n}{N^d} \left(\frac{n}{\left(\left(\frac{n}{N^d}\right)^\rho+1\right)N^d}\right)^{-s}$ is bounded above by $n^{-r}$ for some $r>1$. Some more calculations show that if 
\[\rho\leq \frac{s\frac{2k}{2k+1}(1-m)-r-1}{s\frac{2k}{2k+1}(1-m)},\] then $T(n)= O(n^{-r})$ for some fixed constant $C<+\infty$ over all $n$. Again, since $s>\max\left\{\frac{r+m}{1-m},\frac{r+1}{\frac{2k}{2k+1}(1-m)}\right\}$ by Assumption \ref{ass:stats}, $\rho\in(0,1)$.

Hence picking 
\[\rho\leq \min\left\{\frac{s-m(1+s)-r}{(1+s)(1-m)},\frac{s\frac{2k}{2k+1}(1-m)-r-1}{s\frac{2k}{2k+1}(1-m)}\right\}\] ensures that
\[P\left(\max_{1\leq k\leq N^d}\left\lvert \hat{f}_{Nn}(\tilde{x}_k)-\frac{\mathbb{E}[S_k]}{\mathbb{E}[M_k]}\right\rvert>\delta\right) \leq C n^{-r}\] for some constant $C<+\infty$.

We now show the almost sure convergence uniformly over all cubes $Q_k$. For each $n\in\mathbb{N}$, consider the event 
\[A_n\coloneqq\left\{\max_{1\leq k\leq N^d}\left\lvert \hat{f}_{Nn}(\tilde{x}_k)-\frac{\mathbb{E}[S_k]}{\mathbb{E}[M_k]}\right\rvert>\delta\right\},\] and the above computations show that $P(A_n)\leq C n^{-r}$. Hence,
\[\sum_{i=1}^n P(A_n)\leq C \sum_{i=1}^n n^{-r}<+\infty\] and the Borel-Cantelli Lemma implies that $A_n$ occurs infinitely often with probability zero. Thus, for large enough $n$, almost surely $A_n=0$, which shows the almost sure convergence of $\hat{f}_{nN}(\tilde{x}_k)$ to $\lim_{N\to+\infty} \frac{\mathbb{E}[S_k]}{\mathbb{E}[M_k]}$. 

If $f$ is continuous at $\tilde{x}_k$, then, focusing on a sequence of cubes $\{Q_{k}\}_N$ shrinking in such a way that they all contain $\tilde{x}_k$, the classical Lebesgue-Besicovitch differentiation theorem \citep[e.g.][section 1.7.1]{evans2018measure} implies that $\hat{f}_{Nn}(\tilde{x}_k)= \lim_{N\to+\infty} \frac{\mathbb{E}[S_k]}{\mathbb{E}[M_k]} = \frac{f\left(\tilde{x}_k\right) f_X\left(\tilde{x}_k\right)}{f_X\left(\tilde{x}_k\right)} =  f(\tilde{x}_k)$.

If $\tilde{x}_k\in S_f$ the limit value depends on how the cubes shrink while containing $x_k$. We can write
\begin{align*}
\lim_{N\to+\infty}\hat{f}_{Nn}(\tilde{x}_k) &= \lim_{N\to+\infty}\frac{\int_{Q_{kN}} f(x) f_X(x) dx}{\int_{Q_{kN}} f_X(x) dx} \\
&= \lim_{N\to+\infty}\frac{\int_{Q_{kN} \cap H^+(\tilde{x}_k, \rho)} f(x) f_X(x) dx+ \int_{Q_k\cap H^-(\tilde{x}_k,\rho)} f(x) f_X(x) dx}{\int_{Q_{kN}} f_X(x) dx},
\end{align*}
where $Q_{kN}$ is a sequence of cubes containing $\tilde{x}_k$ such that they converge to $\{\tilde{x}_k\}$ and $H^+(\tilde{x}_k,\rho)=\{y: \langle y-x,\rho\rangle\geq 0\}$ and $H^-(\tilde{x}_k,\rho) = \mathbb{R}^d\setminus H^+(\tilde{x}_k,\rho)$ are the two half spaces created by the Hyperplane $H(\tilde{x}_k,\rho)$ oriented by some $\rho\in S^{d-1}$ in the direction of the jump of $f$ at $\tilde{x}_k$, as depicted in Figure \ref{fig:jump}. Since each $Q_{kN}$ intersects both half balls for all $N$, the limit $\lim_{N\to+\infty}\hat{f}_{Nn}(\tilde{x}_k)$ converges by the Lebesgue-Besicovitch differentiation theorem to the convex combination
\[\lim_{N\to +\infty} \hat{f}_{Nn}(\tilde{x}_k) = (1-\theta)f^+(\tilde{x}_k) + \theta f^{-}(\tilde{x}_k),  \]
where 
\[0<\theta = \lim_{N\to+\infty} \frac{\int_{Q_{kn}\cap H^+(\tilde{x}_k,\rho)} f_X(x) dx}{\int_{Q_{kn}} f_X(x)dx}<1.\]
Note that since cubes are symmetric with respect to their center point, any hyperplane going through the center point of $Q_k$ and dissecting it splits $Q_k$ into two subsets of the same Lebesgue measure. Hence if $f_x$ is the density of the continuous uniform distribution and the grid consists of all cubes, then $\theta = 1/2$.

\begin{figure}[h!]
    \centering
\begin{tikzpicture}[scale=1]

        \begin{scope}
            \clip (-3,-1) rectangle (3,1);
            \fill[gray!30] (-3,-0.6) .. controls (-2,-0.9) and (-1,-0.2) .. (0,0)
                                  .. controls (0.4,0.2) and (0.8,0.8) .. (1.6,1) -- (3,1) -- (3,-1) -- (-3,-1) -- cycle;
        \end{scope}

        \fill[gray!30] (-1,-3) rectangle (1,-1);

        \draw[thick] (-3,-0.6) .. controls (-2,-0.9) and (-1,-0.2) .. (0,0)
                                  .. controls (0.4,0.2) and (0.8,0.8) .. (1.6,1);
        
    \draw[thick] (-1,-1) rectangle (1,1);

        \draw[thick] (-3,-1) rectangle (-1,1); 
        \draw[thick] (1,-1) rectangle (3,1);   
        \draw[thick] (-1,1) rectangle (1,3);   
        \draw[thick] (-1,-3) rectangle (1,-1); 

    \draw[thick,dashed] (-2,-1) -- (2,1);
    \node[below right] at (1.5,0.85) { $H(\tilde{x}_k,\rho)$};
    


    \filldraw[black] (0,0) circle (0.8pt);
    \node[above left] at (0,0) { $\tilde{x}_k$};

    \filldraw[black] (2,0) circle (0.8pt);
    \node[right] at (2,0) { $\tilde{x}_{\overset{\rightarrow}{k}}$};

    \filldraw[black] (0,-2) circle (0.8pt);
    \node[right] at (0,-2) { $\tilde{x}_{k\uparrow}$};

    \filldraw[black] (0,2) circle (0.8pt);
    \node[right] at (0,2) { $\tilde{x}_{k\downarrow}$};

    \filldraw[black] (-2,0) circle (0.8pt);
    \node[right] at (-2,0) { $\tilde{x}_{\overset{\leftarrow}{k}}$};

    \node[right] at (-2.4,-0.4) { $S_f$};



    \draw[->, thick] (0,0) -- (0.26,-0.5);
    \node[right] at (0.2,-0.4) {$\rho$};
\end{tikzpicture}
\caption{Representation of the case $\tilde{x}_k \in S_f$.}
    \label{fig:jump}
\end{figure}

\noindent\emph{Putting everything together:}\\
The argument is essentially the same as in the proof of Proposition \ref{prop:consistency_det}, except that we allow for $\lambda_n\to+\infty$ and the bounds are now
\begin{align*}
    \tilde{p}^x_k = p_k^x \left(\frac{\nu}{\nu+o(1)+\frac{L}{2N}}\right)\qquad\text{and}\qquad
    \tilde{p}^x_k = p_k^x\sqrt{\frac{S^2}{R^2 + 2RS + S^2}},
\end{align*}
with 
\begin{align*}
    R\coloneqq \frac{L}{2N}+\sqrt{\hat{f}_{X,h(n)}(\tilde{x}_k)\frac{2L}{N}},\,\,
    S\coloneqq\sqrt{4\hat{f}_{X,h(n)}(\tilde{x}_k)\left({}p_k^t+\lambda_n\hat{f}_{X,h(n)}(\tilde{x}_k)\left(t_k-\hat{f}_{Nn}(\tilde{x}_k)\right)^2\right)}.
\end{align*} 
Taking the minimum of these two expressions for $\tilde{p}_k^x$, and calling this minimum $0\leq m(N)\leq 1$ we first want to show that $m(N)$ converges almost surely to $1$ as $N(n)\to+\infty$. Since $\hat{f}_{X,h(n)}\overset{a.s.}{\to} f_X$ as $n\to+\infty$, where $f_X$ is bounded above by Assumption \ref{ass:stats}, $R^2\to0$ almost surely. Now analyze $2RS$:
\begin{align*}
    RS &= \left(\frac{L}{2N}+\sqrt{\hat{f}_{X,h(n)}(\tilde{x}_k)\frac{2L}{N}}\right)\sqrt{4\hat{f}_{X,h(n)}(\tilde{x}_k)\left({}p_k^t+\lambda_n\hat{f}_{X,h(n)}(\tilde{x}_k)\left(t_k-\hat{f}_{Nn}(\tilde{x}_k)\right)^2\right)}\\
    &=\sqrt{\frac{L^2}{N^2}\hat{f}_{X,h(n)}(\tilde{x}_k)\left({}p_k^t+\lambda_n\hat{f}_{X,h(n)}(\tilde{x}_k)\left(t_k-\hat{f}_{Nn}(\tilde{x}_k)\right)^2\right)} \\
    &\hspace{2cm}+ \sqrt{\frac{8L}{N}\hat{f}^2_{X,h(n)}(\tilde{x}_k)\left({}p_k^t+\lambda_n\hat{f}_{X,h(n)}(\tilde{x}_k)\left(t_k-\hat{f}_{Nn}(\tilde{x}_k)\right)^2\right)}.
\end{align*}
Since $f_X$ and $\tilde{f}$ are bounded and both $\hat{f}_{X,h(n)}$ and $\hat{f}_{Nn}$ are almost surely consistent for $f_X$ and $\tilde{f}$, it follows from $\lambda_n = o(N)$ that the both terms converge to $0$ almost surely, which implies that $RS\to0$ almost surely. This shows that $m(N)\to1$ almost surely. The rest of the argument now is the same as in Proposition \ref{prop:consistency_det}.
\end{proof}

\section{Implementation} \label{app:implement}

We solve \eqref{eq:main_problem_discrete} using a primal-dual algorithm \citep{chambolle2011first}. At each step, we need to project the iterand $v^n$, which denotes the estimated primal function $v$ at the $n$-th iteration of the algorithm,
onto the sets $C$ and $K$. We now discuss how these projections are implemented, before presenting the algorithm. The implementation follows those of \citet{strekalovskiy2012convex} and \citet{bauer2016solving}.

\subsection{Projection Onto Constraint Sets}

\paragraph*{Projection Onto C} The projection onto $C$ can be done using a straightforward clipping,
\begin{equation}
v^{n+1} = \min\{ 1, \max{0, v^n} \},
\end{equation}
where $v^n$ is the $n$-th iteration of the discretized function $v$. We also need to impose the discretized limits from $C$, by setting $v^{n+1}(k_1, \ldots, k_d, 1) = 1$ and $v^{n+1}(k_1, \ldots, k_d, N) = 0$. 

\paragraph*{Projection Onto K} The projection onto $K$ is more involved. The first constraint,
\begin{equation}
 K_1 = 
\left\{p^t(k) \geq \frac{\left\lvert p^x(k)\right\rvert _2^2}{4}-\lambda\left(\frac{k}{S}-f(k_1, \ldots, k_d)\right)^2\right\},
\end{equation}
constitutes a pointwise projection onto a parabola, which can be rewritten as an optimization program. To see this, let $\alpha>0, p^x \in \mathbb{R}^d, p^t \in \mathbb{R}$ and $p=\left(p^x, p^t\right)^T \in \mathbb{R}^d \times \mathbb{R}$. Assume that $p_0^t<\alpha\left\lvert p_0^x\right\rvert _2^2$ holds for a point $p_0 \in \mathbb{R}^d \times \mathbb{R}$. The projection of $p_0$ onto the parabola $\alpha\left\lvert p_0^x\right\rvert _2^2$ can then be written as the following optimization program:
$$
\begin{array}{cl}
\min _{p \in \mathbb{R}^d \times \mathbb{R}} & \frac{1}{2}\left\lvert p-p_0\right\rvert _2^2 \\
\text { subject to } & p^t \geq \alpha\left\lvert p^x\right\rvert _2^2,
\end{array}
$$
where $\alpha=\frac{1}{4 \hat{f}_{X,N_n}(k_1,\ldots,k_d)} $ and we leave out the constant $\lambda \hat{f}_{X,N_n}(\ldots) \left(\frac{k}{N}-f( \ldots)\right)^2$ for ease of notation. The first-order conditions of the corresponding Lagrangian are,
\begin{equation}
\left(\begin{array}{c}
p^x-p_0^x+\mu 2 \alpha p^x \\
p^t-p_0^t-\mu \\
\alpha\left\lvert p^x\right\rvert _2^2-p^t
\end{array}\right)=0.
\end{equation}
Solving for $\mu$ and recombining gives the cubic equation,
$$
t^3+3 b t-2 a=0
$$
with $a=2 \alpha\left\lvert p_0^x\right\rvert ^2_2, b=\frac{2}{3}\left(1-2 \alpha p_0^t\right)$ and $t=2 \alpha\left\lvert p^x\right\rvert ^2_2,$
which can be solved analytically as \citep{mckelvey1984simple},
$$
p^x= \begin{cases} p^x_0 & \text{ if } p_0^t \geq \alpha\left\lvert p_0^x\right\rvert _2^2 \\  \frac{w}{2 \alpha} \frac{p_0^x}{\left\lvert p_0^x\right\rvert ^2_2} & \text { if } p_0^t < \alpha\left\lvert p_0^x\right\rvert _2^2 \text{ and } p_0^x \neq 0 \\ 0 & \text { else }\end{cases}
$$
and $$ p^t= \begin{cases} p^t_0 & \text{ if } p_0^t \geq \alpha\left\lvert p_0^x\right\rvert _2^2 \\  \alpha\left\lvert p^x\right\rvert _2^2  & \text{ else} \end{cases},$$
where
$$
d= \begin{cases}a^2+b^3 & \text { if } b \geq 0 \\ \left(a-\sqrt{-b}^3\right)\left(a+\sqrt{-b}^3\right) & \text { else }\end{cases}
$$
and
$$
w= \begin{cases} 0 & \text{ if } c = 0 \\ c-\frac{b}{c}  & \text { if } d \geq 0 \text{ and } c > 0 \\ 2 \sqrt{-b} \cos \left(\frac{1}{3} \arccos \frac{a}{\sqrt{-b}^3}\right) & \text { else }\end{cases},
$$
with $c=\sqrt[3]{a+\sqrt{d}}$.

The second constraint, we decouple from the first by way of Lagrange multipliers \citep{strekalovskiy2012convex}. This lets us avoid the Dykstra projection originally used in \citet{pock2009algorithm}, which requires nested iterations (outer loop: primal-dual algorithm, inner loop: Dykstra's algorithm) and is thus computationally costly. In particular, we introduce a set of auxiliary variables $s_{s_1, s_2} \coloneqq \sum_{s_1 \leq k_{d+1} \leq s_2} p^x(k)$ and of Lagrange multipliers $\mu_{s_1,s_2} \in \mathbb{R}^d$ and write the second constraint set as, 
\begin{equation}
K_2 =\left\{\left|s_{s_1, s_2}\right| \leq \nu \text { s.t. } s_{s_1, s_2}=\sum_{s_1 \leq k_{d+1} \leq s_2} p^x(k)\right\},
\end{equation}
with corresponding Lagrangian,
\begin{equation}
\begin{aligned}
& \mathcal{L}(v, \mu, p, s)= \\
& \langle p, D_N v\rangle_N+\sum_{s_1=1}^S \sum_{s_2=k_1}^S\left\langle\mu_{s_1, s_2}, \sum_{s_1 \leq k_{d+1} \leq s_2} p^x(k)-s_{s_1, s_2}\right\rangle.
\end{aligned}
\end{equation}
Then, let $(v^*, p^*)$ be the solution to \eqref{eq:main_problem_discrete} and $(v^*, p^*, \mu_{s_1,s_2}^*, s_{s_1, s_2}^*)$ the solution to,
\begin{equation}
\min _{\substack{v \in C \\ \mu_{s_1, s_2}}} \max _{\substack{p \in K_1 \\\left|s_{s_1, s_2}\right| \leq \nu}} \mathcal{L}(u, \mu, p, s). 
\end{equation} 
We have that,
\begin{equation}
\left\langle p^*, D_N v^*\right\rangle_N =\mathcal{L}\left(v^*, \mu^*, p^*, s^*\right),
\end{equation}
because $p \in K$ implies $p \in K_2$ and hence the Lagrange constraint $s_{s_1, s_2}=\sum_{s_1 \leq k_{d+1} \leq s_2} p^x(k) $ always holds with equality at the optimum.

\subsection{Algorithm} 

 Putting everything together, we solve the discretized problem using a primal-dual gradient descent-ascent algorithm \citep{chambolle2011first}. To calculate the gradient updates, we just need the derivatives with respect to the Lagrangian above, which are,
\begin{equation}
\begin{aligned}
\frac{\partial \mathcal{L}(u, \mu, p, s)}{\partial u} & =D_N^T p \\
\frac{\partial \mathcal{L}(u, \mu, p, s)}{\partial s_{s_1, s_2}} & =- \mu_{s_1, s_2} \\
\frac{\partial \mathcal{L}(u, \mu, p, s)}{\partial \mu_{s_1, s_2}} & = p^x(k)-s_{s_1, s_2} \\
\frac{\partial \mathcal{L}(u, \mu, p, s)}{\partial p} & = D_N u+\tilde{p},
\end{aligned}
\end{equation}
where \begin{equation}
\tilde{p}=\left(\begin{array}{c}
\sum_{s_1=1}^l \sum_{s_2=l}^S \mu_{s_1, s_2}^1 \\
\sum_{s_1=1}^l \sum_{s_2=l}^S \mu_{s_1, s_2}^2 \\
0
\end{array}\right).
\end{equation}
We get the following algorithm,
\begin{algorithm} \label{algo:primaldual}
\caption{Primal-Dual Algorithm}
\begin{algorithmic}[1]
\STATE Choose $\left(v^0, p^0, \mu^0, s^0\right) \in C \times K_p \times \mathbb{R}^{d \times N_1 \times \ldots \times N_d \times I} \times \mathbb{R}^{d \times N_1 \times \ldots N_d \times I}$ and let $\bar{v}^0=u^0, \bar{\mu}^0=\mu^0=0, p^0=0$. 
\STATE Set $\tau_v=\sigma_p=\frac{1}{4 (d+1)}$ and $\tau_\mu=1/I, \sigma_s=1$. 
\STATE For each $n \geq 0$:
\STATE 
$$
\left\{\begin{array}{l}
p^{n+1}=\Pi_{K_p}\left(p^n+\sigma_p\left(D_N \bar{v}^n+\tilde{p}^n \right)\right)  \\
s_{s_1, s_2}^{n+1}=\Pi_{|\cdot| \leq \nu}\left(s_{s_1, s_2}^n-\sigma_s \bar{\mu}_{s_1, s_2}^n\right) \\
v^{n+1}=\Pi_C\left(v^n-\tau_v D_N^* p^{n+1}\right) \\
\mu_{s_1, s_2}^{n+1}=\mu_{s_1, s_2}^n+\tau_\mu \left(s_{s_1, s_2}^{n+1}- \sum_{s_1 \leq k_{d+1} \leq s_2} p^{x,n+1}(k) \right) \\
\bar{v}^{n+1}=2 v^{n+1}-v^n \\ 
\bar{\mu}_{s_1, s_2}^{n+1}=2 \mu_{s_1, s_2}^{n+1}-\mu_{s_1, s_2}^n,
\end{array} \right.
$$
where $\Pi_D(x)$ denotes the projection of $x$ onto the set $D$, $I \coloneqq \frac{S^2+S}{2}$ the number of $K$ constraints, and $D^*_N = D_N^T \coloneqq N \cdot D^T \coloneqq - N \cdot \operatorname{div}$ the adjoint of the discrete gradient operator.\footnote{Implemented in practice as backward differences with a Neumann boundary condition.}
\RETURN $v^*$

\end{algorithmic}
\end{algorithm}
 It was proved in \citet[Thm.1]{chambolle2011first} that $v \to v^*$ for this algorithm as the number of iterations $n$ goes to infinity, under the condition that the step sizes satisfy $\tau \sigma L^2 < 1$, where $L = \lvert K \rvert \coloneqq \max\{\lvert Kx \rvert : x \in X \text{ with } \lvert x\rvert \leq 1 \}$ with $K$ the continuous linear operator on the primal variable.

 \subsection{Computational Details and Code}

We implement the described algorithm in \href{https://pytorch.org}{PyTorch}, a modern deep learning framework that offers native support for GPU acceleration. Running the algorithm on a GPU, as opposed to a CPU, provides several advantages. GPUs are inherently designed for high-throughput parallelism, enabling the simultaneous processing of thousands of matrix operations. PyTorch offers native support for the whole spectrum of GPU architectures, including NVIDIA, Apple Silicon (for late 2020 Apple computers onward), and Intel GPUs. This ensures that researchers and practitioners with varied hardware setups can readily apply the estimator. For efficient computation of the SURE, which requires solving the algorithm a large number of times, we rely on the hyperparameter tuning suite in \href{http://ray.io}{Ray}, an open-source distributed computing framework that supports parallel and fractional GPU processing. This can drastically speed up the hyperparameter selection by efficiently allocating jobs across multiple partitions of a single GPU on a local machine or multiple GPUs on a high-performance computing cluster. 

The accompanying Python library can be found at \url{https://github.com/Davidvandijcke/fdr}. A fully developed Python package as well as extensions to R and STATA are in progress. Though the package is compatible with all major GPU architectures, researchers without GPUs on their local machines can refer to our Google Colab notebooks, which provide free access to cloud-based machines with GPU support. The current implementation of the algorithm takes 69.37 seconds on a Nvidia Tesla A100 Tensor Core GPU to converge on a 2D dataset with the number of raw data points $n=90,000$, the number of grid points $N=12,750$ (25\%), and the discretization of the lifted dimension $S=32$. This computation time can be further improved by approximating the objective function using ``sublabel-accurate'' relaxations \citep{mollenhoff2017sublabel}. As this requires further extensions of our statistical convergence arguments, we leave this to future work.

\section{Uncertainty Quantification} \label{app:uncertainty}
\subsection{Subsampling} \label{sec_app:subsampling}

\begin{algorithm}[htbp!]
\caption{Subsampling}
\label{algo:sub}
\begin{algorithmic}[1]
\STATE \textbf{Input}: Data $\left(X_i, Y_i\right), \mathcal{I} \coloneqq \{1, \ldots, n\}$, miscoverage level $\alpha \in(0,1)$, regression algorithm $\mathcal{A}$ 
\STATE \textbf{Output:} Confidence band, over $x \in \mathcal{Q_{N_-}}$
    \STATE $\hat{u} = \mathcal{A}(\{(X_i,Y_i):i\in \mathcal{I}\})$ 
    \FOR{$j = 1$ to $J$}
        \STATE Randomly sample $\mathcal{I}$ into $K$ subsets $\mathcal{I}_1, \ldots, \mathcal{I}_k$ of size $b_1, \ldots,b_K$
        \STATE $u^*_{j,k} \coloneqq \mathcal{A}(\{ (Y_i, X_i) : i \in \mathcal{I}_k \})$
        \STATE $Z[j,k] = \max\{|u^*_{j,k} - \hat{u}|\}$
    \ENDFOR
    \STATE $\bar{y}_{k} \coloneqq \frac{1}{L} \sum_{\ell=1}^L \log \left[G_{b_k}^{*-1}\left(t_\ell\right)-G_{b_k}^*{ }^{-1}\left(s_\ell\right)\right]$, the mean of $L$ log differences of the empirical quantile functions $G_b^{*-1}(t)$ at quantiles $s_\ell,t_\ell$
    \STATE $\hat{\beta}\coloneqq-\frac{\operatorname{cov}\{\bar{y_k}, \log (b_k)\}}{\operatorname{var}\{\log (b_k)\}}$ the rate of convergence
    \STATE For some $k \in 1,\ldots,K$:
    \STATE $Z^* = b_k^{\hat{\beta}} \cdot \max\{|\mathbf{u}^*_{k} - \hat{u})|\}$ with $\mathbf{u}^*_{k} = (u^*_{1,k}, \ldots, u^*_{J,k})$ 
    \STATE $z_\alpha \coloneqq \operatorname{sort}(Z^*)[(J + 1) \cdot (\alpha)]$ the critical value
    \RETURN $C_{\text{sub}}(x, \alpha) = [\hat{u}(x) - z_\alpha / n^{b_k}, \hat{u}(x) +  z_\alpha / n^{b_k}]$, for all $x \in \mathcal{Q}_{N_-}$
\end{algorithmic}
\end{algorithm}
We construct uniform confidence bands by way of a subsampling approach with an estimated rate of convergence \citep[Ch.8]{politis1999subsampling}. The only assumption required for this approach to be consistent is that the limiting distribution exists and is non-degenerate for the rate of convergence we estimate. Under this assumption, we construct uniform confidence bounds using Algorithm \ref{algo:sub}, where $\mathcal{Q}_{N_-}$ is the grid on the domain only, that is, without the lifted dimension. Note that, since we aim to construct confidence bands around a non-parametric estimate $\hat{u}$ on a grid of size $N$, we keep the grid size fixed when calculating the subsampled estimates $u^*_{j,k}$. With this algorithm, we obtain the confidence bands and 95\% significant jump sets depicted in Figures \ref{fig:simulated_examples} and \ref{fig:india_econ}. To obtain confidence bands on the jump sizes, we simply repeat lines 8--14 in Algorithm \ref{algo:sub} for the forward difference of $\hat{u}$ and of the subsampled $u^*$s and conclude that a jump at point $x \in S_u$ is significant at the $\alpha$\% level if either $\max\{y : y\in  C^D_{\text{sub, lower}}(x, \alpha) \} > 0,$ or $\min\{y : y\in  C^D_{\text{sub, upper}}(x, \alpha) \} < 0$, where $C^D$ indicates the confidence bands for the forward differences and ``lower'', ``upper'' indicate whether we consider the upper or lower bound. We take the max and min since in practice we set the jump size equal to the largest forward difference along any of the axes.

\subsection{Conformal Inference}
As a computationally efficient, though more conservative, alternative to quantifying the uncertainty of the function $\hat{u}$ and the jump set $S_u$, we rely on distribution-free conformal prediction methods developed for the non-parametric regression setting \citep{lei2018distribution}. Conformal prediction, originally, is a method for constructing bands around predictions produced by machine learning methods. Its appeal lies in the fact that it can produce prediction bands for any general estimator without requiring assumptions about its properties or about the distribution of the data-generating process. In particular, we construct confidence bands around $u$ and $\lvert \nabla u \rvert$ using the Split Conformal Prediction Algorithm \ref{algo:scp} proposed in \citet[Algorithm 2]{lei2018distribution},
\begin{algorithm}[htbp!]
\caption{Split Conformal Prediction \citep{lei2018distribution} } \label{algo:scp} 
\begin{algorithmic}[1]
\STATE \textbf{Input:} Data $(X_i, Y_i), i=1, \ldots, n$, miscoverage level $\alpha \in(0,1)$, regression algorithm $\mathcal{A}$
\STATE \textbf{Output:} Prediction band, over $x \in \mathbb{R}^d$
\STATE Randomly split $\{1, \ldots, n\}$ into two equal-sized subsets $\mathcal{I}_1, \mathcal{I}_2$
\STATE $\widehat{\mu}=\mathcal{A}\left(\left\{\left(X_i, Y_i\right): i \in \mathcal{I}_1\right\}\right)$
\STATE $R_i=\left|Y_i-\widehat{\mu}\left(X_i\right)\right|, i \in \mathcal{I}_2$
\STATE $d(\alpha)=$ the $k$th smallest value in $\left\{R_i: i \in \mathcal{I}_2\right\}$, where $k=\lceil(n / 2+1)(1-\alpha)\rceil$
\RETURN $C_{\text{split}}(x, \alpha)=[\widehat{\mu}(x)-d(\alpha), \widehat{\mu}(x)+d(\alpha)]$, for all $x \in \mathbb{R}^d$
\end{algorithmic}
\end{algorithm}
The authors show that, for $(X_i, Y_i)_{i=1,\ldots,n}$ i.i.d. and assuming that the residuals $R_i, i \in \mathcal{I}_2$ have a continuous joint distribution, which is guaranteed in the limit under the Gaussian assumption for the SURE, 
\[\mathbb{P}\left(Y_{n+1} \in C_{\text {split}}\left(X_{n+1}, \alpha \right)\right) \leq 1-\alpha+\frac{2}{n+2}. \]
Thus, $C_{\text{split}}$ delivers prediction bands for every grid point $x \in \mathcal{Q}_{N_-}$ on which we estimate $\hat{u}(x)$. 

In practice, we adapt Algorithm \ref{algo:scp} to our grid-based setting as follows. Keeping the hyperparameters $\lambda, \nu$ and the grid $\mathcal{Q}_{N}$ fixed, randomly split $\{1,\ldots,n\}$ into equal subsets $\mathcal{I}_1, \mathcal{I}_2$ and cast $(X_i, Y_i), i \in \mathcal{I}_1$ onto $\mathcal{Q}_{N}$ as in Section \ref{sec:consistency}. Solve problem \eqref{eq:main_problem_discrete} on the grid, and then calculate $R_i$ as the difference between $Y_i$ and the prediction $\hat{u}(X'_i)$ for the grid point $X_i' \in \mathcal{Q}_{N_-}$ closest to $X_i$. This gives prediction bands for the outcome variable $Y_i$. To obtain prediction bands for the jump set $J_f$, we also calculate the residual vector $R^J_i$ for the implied predictions $D \hat{u}$ by computing $ D Y'_i $ where $Y'_i$ is the outcome value corresponding to the data point $X_i$ closest to the grid point $X'_i$. Inference on the jump set then proceeds identically to the subsampling case.
Finally, we stress that these \textit{prediction} bands allow us to do inference on $Y$, but not on the conditional mean $f(X)$. Confidence bands for the latter will always be smaller than prediction bands for the former, so this approach provides a computationally efficient but very conservative hypothesis test regarding $f(X)$ and the corresponding jump set $J_f$. If the goal is to do inference on $f(X)$, we recommend using the subsampling approach described in the previous section.

\section{Additional Results} \label{app:additional_results}



\clearpage
\subsection{Figures}

\begin{figure}[htbp!]
\includegraphics[width=0.7\textwidth]{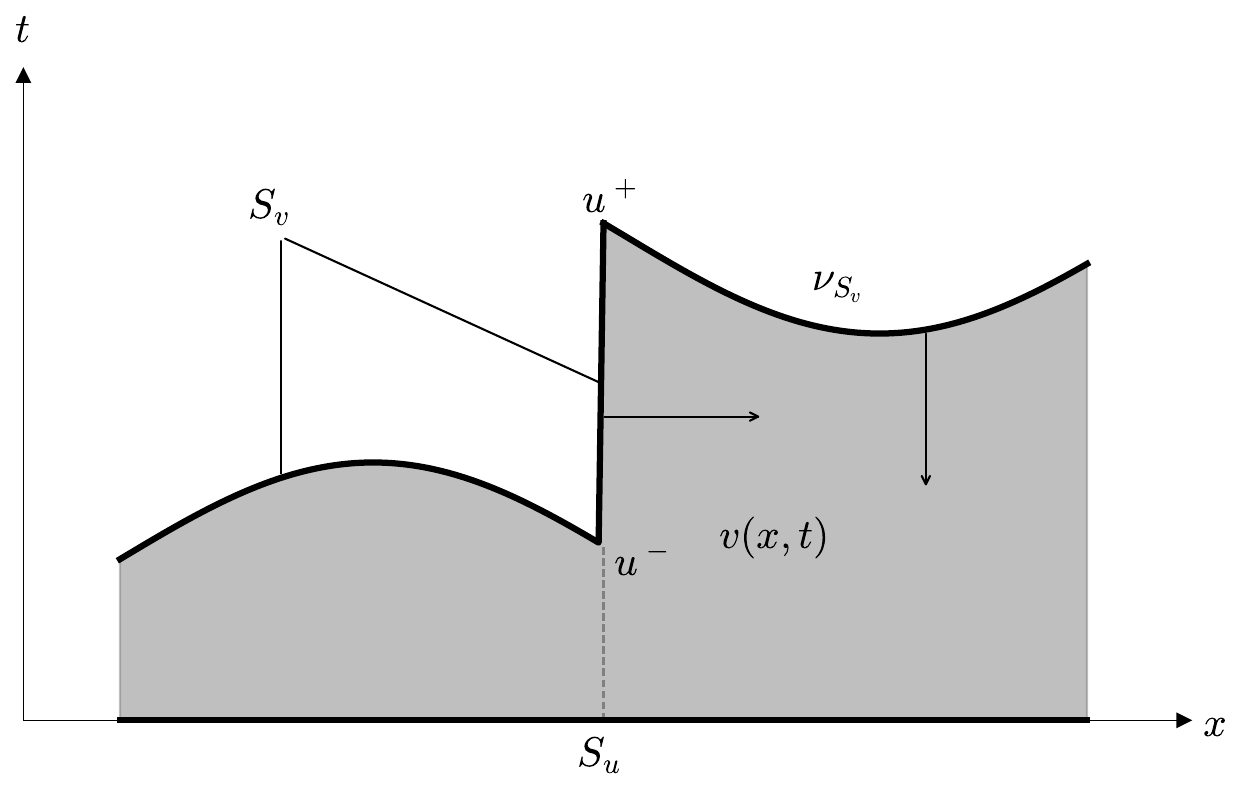}
\caption{Convex Relaxation Through Functional Lifting}
\end{figure}

\begin{figure}[htbp!]
\includegraphics[width=0.7\textwidth]{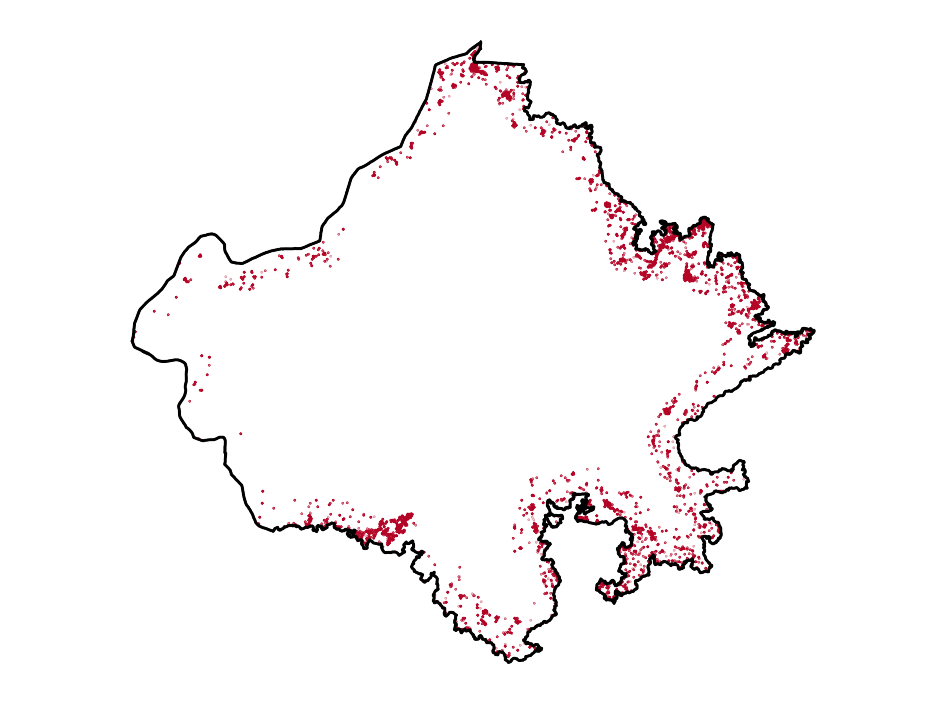}
\caption{Catchment Area for Estimating Degree of Selection}
\floatfoot{\textit{Note}: Figure plots pings inside 40km band within Rajasthan around the state boundary during the shutdown period. The degree of self-selection is estimated by comparing the share of devices associated with these pings that cross into neighboring states with its average in the prior month.}
\label{fig:crossers}
\end{figure}

\begin{figure}[htbp!]
\begin{subfigure}[b]{0.49\linewidth}
    \includegraphics[width=\textwidth]{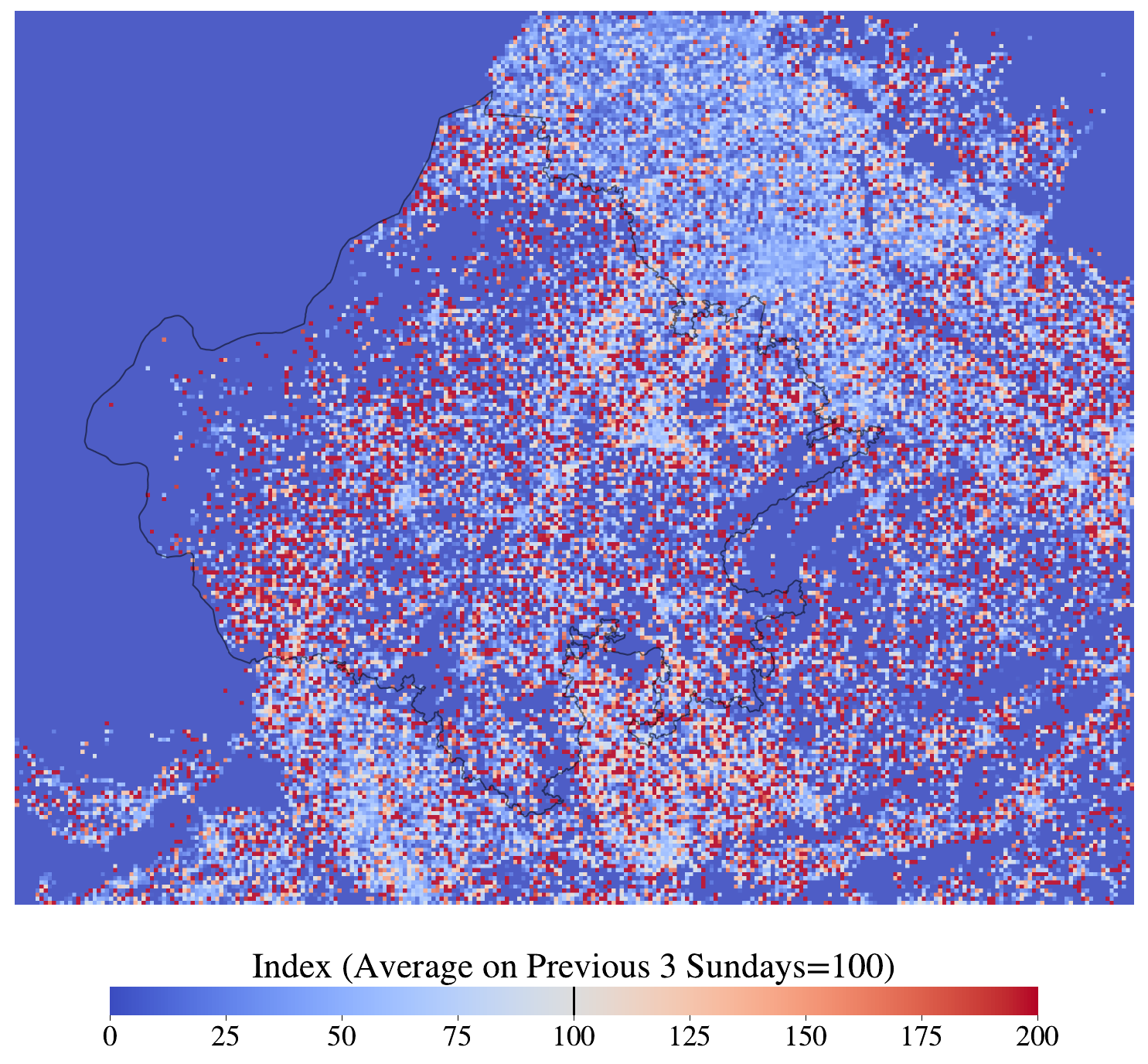}
\caption{Mobile: Raw Data, Day Of, 6 pm -- Midnight}
\end{subfigure}
\begin{subfigure}[b]{0.49\linewidth}
    \includegraphics[width=\textwidth]{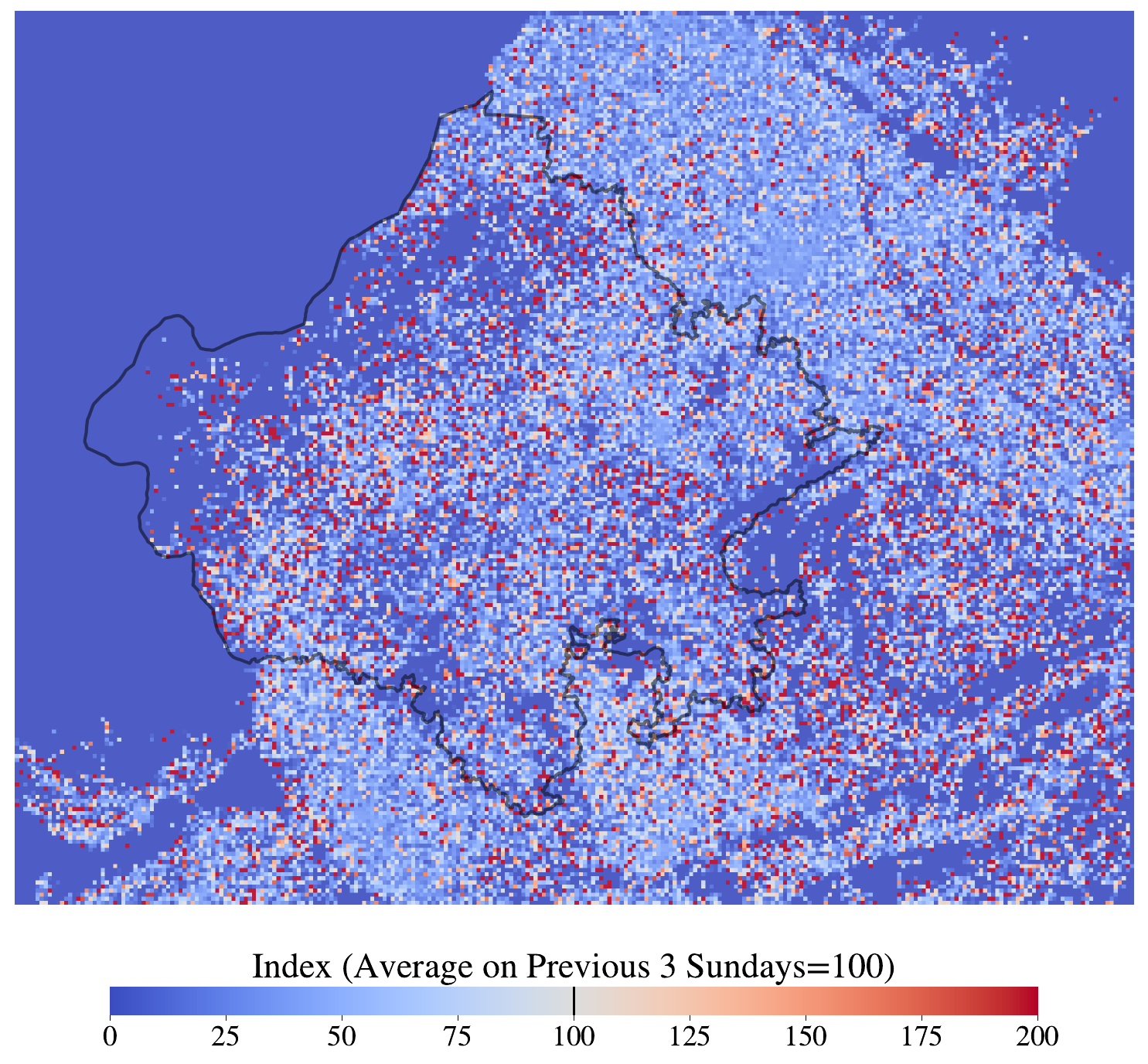}
\caption{Mobile: Raw Data, Day After, 6 am -- 6 pm}
\end{subfigure}
\caption{Post-Shutdown Activity}
\label{fig:india_overshooting}
\floatfoot{\textit{Note}: Plot shows the raw mobile device data on a 5x5km grid with the fill color of each cell indicating the value of $\overline{Pings}$, for the hours between 6 pm and midnight on the day of the shutdown in (a) and for the time spanning the shutdown window the day after the shutdown in (b). The outline of Rajasthan state is indicated by black lines.}
\end{figure}

\begin{figure}[htbp!]
\begin{subfigure}[b]{0.49\linewidth}
    \includegraphics[width=\textwidth]{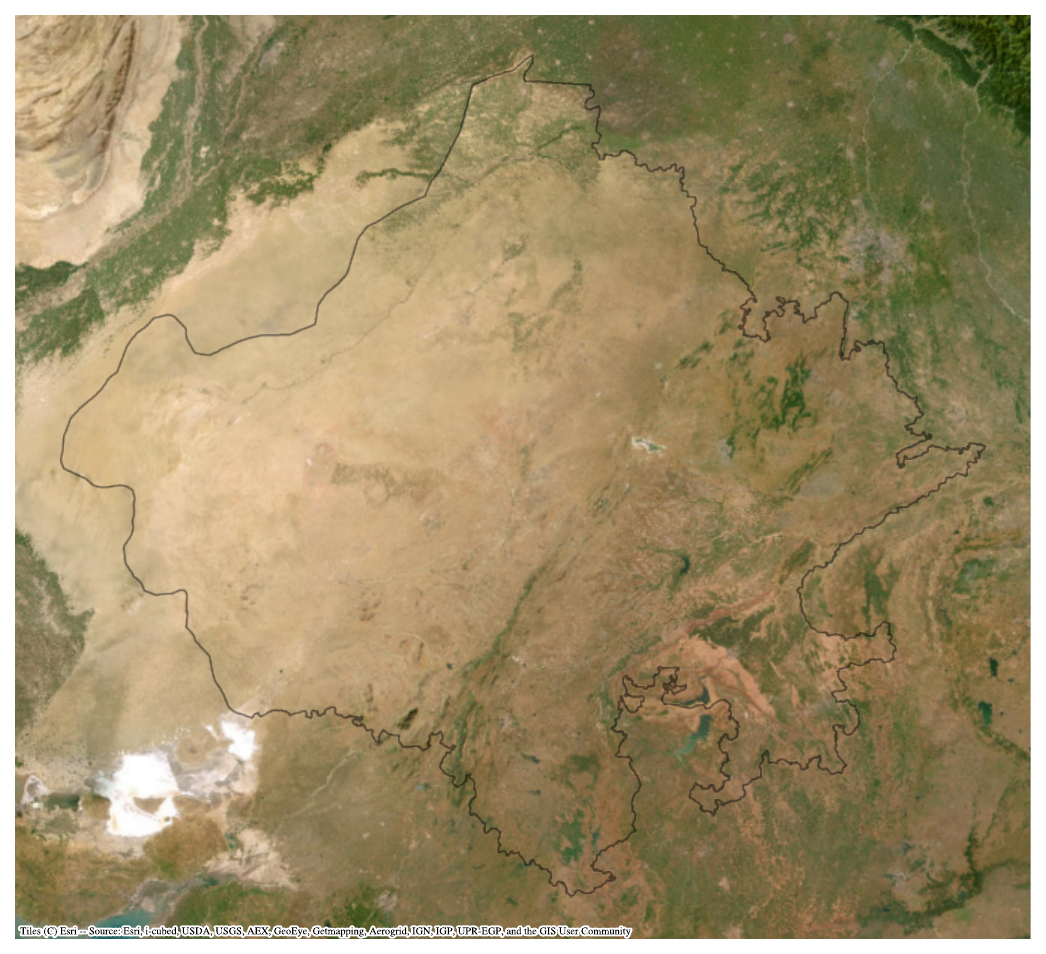}
        \caption{Esri Satellite Imagery}
        \label{fig:india_terrain_satellite}
\end{subfigure}
\begin{subfigure}[b]{0.49\linewidth}
    \includegraphics[width=\textwidth]{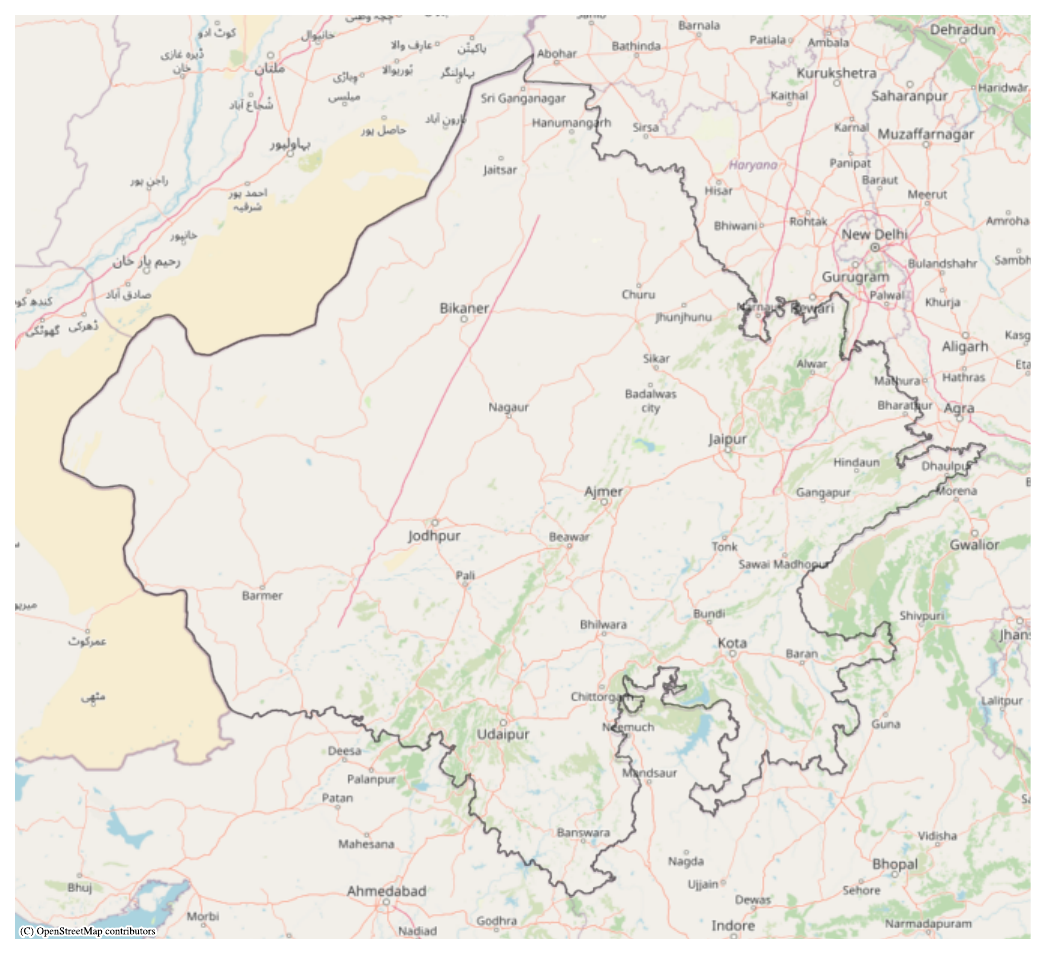}
    \caption{OpenStreetMap}
            \label{fig:india_terrain_street}
\end{subfigure}
\caption{Rajasthan: Terrain View}
\label{fig:india_terrain}
\floatfoot{\textit{Note}: \ref{fig:india_terrain_satellite} shows the satellite view of Rajasthan, obtained from Esri; \ref{fig:india_terrain_street} shows the street map of Rajasthan obtained from OpenStreetMap. The outline of Rajasthan is depicted in black.}
\end{figure}

\end{appendix}

\end{appendices}
\end{supplement}

\singlespacing
\bibliographystyle{apalike}
\bibliography{references}
\end{document}